\documentclass[preprint]{imsart}
\pdfoutput=1
\usepackage{array,epsfig,fancyhdr,rotating}
\usepackage[]{hyperref}  
\usepackage{sectsty, secdot}
\sectionfont{\fontsize{12}{14pt plus.8pt minus .6pt}\selectfont}

\subsectionfont{\fontsize{12}{14pt plus.8pt minus .6pt}\selectfont}
\usepackage{tcolorbox}
\usepackage{graphicx}
\usepackage{caption}
\usepackage{subcaption}
\usepackage{upgreek}

\usepackage{enumitem}

\usepackage{amsmath}
\usepackage{amssymb}
\usepackage{amsfonts}
\usepackage{multirow}
\usepackage{amsthm}
\usepackage{mathtools}
\usepackage{bbm}
\usepackage[textwidth=4cm,textsize=footnotesize]{todonotes}

\newtheorem{theorem}{Theorem}
\newtheorem{lemma}{Lemma}

\newtheorem{proposition}{Proposition}
\theoremstyle{definition}

\newtheorem{assumption}{\textbf{A}}[section]

\usepackage{ifthen}
\usepackage{stmaryrd}
\usepackage{xargs}
\usepackage{accents}
\usepackage{algorithm2e}
\usepackage{cleveref}

\crefname{assumption}{{\textbf{A}}}{{\textbf{A}}}
\Crefname{algorithm}{{Algorithm}}{Algorithm}

\newcommand{\af}[1]{h_{#1}}
\newcommand{\afterm}[1]{\tilde{h}_{#1}}
\newcommand{\bcdot}{\cdot}
\newcommand{\bdpart}[1]{\mathbb{B}_{#1}}
\newcommandx{\bk}[2][1=]{
\ifthenelse{\equal{#1}{}}
{\overleftarrow{Q}_{#2}}
{\overleftarrow{Q}_{#2}^{#1}}
}
\def\cstcondparisc{\mathsf{c}}
\def\cstcondparisd{\mathsf{d}}

\newcommand{\bmf}{\mathsf{F}}
\newcommand{\borel}{\mathcal{B}}
\newcommand{\bpart}[2]{\upsilon_{#1}^{#2}}
\newcommand{\bpartmb}[1]{\boldsymbol{\upsilon}_{#1}}
\newcommandx{\canlaw}[2][1=]{\prob^{#1}_{#2}}
\newcommandx{\canlawexp}[2][1=]{\E^{#1}_{#2}}
\newcommand{\ck}[1]{\mbletter{S}_{#1}}
\newcommand{\ckjt}[1]{\mathbb{S}_{#1}}
\newcommand{\condparis}{\mathsf{CondPaRIS}}
\newcommand{\E}{\mathbb{E}}
\newcommand{\efd}[1]{\boldsymbol{\mathcal{E}}_{#1}}
\newcommand{\eg}{\emph{e.g.}}
\newcommand{\epart}[2]{\xi_{#1}^{#2}}
\newcommand{\eparttd}[2]{\tilde{\xi}_{#1}^{#2}}
\newcommand{\epartmb}[1]{\boldsymbol{\xi}_{#1}}
\newcommand{\eqdef}{\coloneqq}
\newcommand{\esp}[1]{\boldsymbol{\mathsf{E}}_{#1}}

\newcommand{\gibbs}[1]{K_{#1}}
\newcommand{\init}{\eta_0}
\newcommandx{\initmb}[1][1=]{\ifthenelse{\equal{#1}{}}{\boldsymbol{\eta}_0}{\boldsymbol{\eta}_{0}\langle #1 \rangle}}
\newcommand{\intvect}[2]{\llbracket #1, #2 \rrbracket}
\newcommand{\koskimat}[1]{\mathbf{J}}

\newcommand{\M}{M}
\newcommand{\mbjt}[1]{\mathbb{C}_{#1}}
\newcommand{\mbletter}[1]{\boldsymbol{#1}}
\newcommand{\meas}{\mathsf{M}}
\newcommand{\md}[1]{m_{#1}}
\newcommand{\mdlow}[1]{\ubar{\sigma}_{#1}}
\newcommand{\mdhigh}[1]{\bar{\sigma}_{#1}}
\newcommand{\mixrate}[2]{\kappa_{#2,#1}}
\newcommand{\mk}[1]{M_{#1}}
\newcommandx{\mkmb}[2][2=]{
\ifthenelse{\equal{#2}{}}
{{\mbletter{M}}_{#1}}
{{\mbletter{M}}_{#1} \langle #2 \rangle}
}

\newcommand{\N}{N}

\newcommand{\nset}{\mathbb{N}}
\newcommand{\nsets}{\mathbb{N}^*}
\newcommand{\nsetpos}{\mathbb{N}^\ast}
\newcommand{\occm}{\mu}
\def\1{\mathbbm{1}}
\newcommand{\zpath}{z}
\newcommand{\parisgibbs}[1]{\mathbb{K}_{#1}}

\newcommand{\partfilt}[1]{\mathcal{F}_{#1}}
\newcommand{\partfiltbar}[1]{\tilde{\mathcal{F}}_{#1}}
\newcommand{\pd}[1]{\Phi_{#1}}

\newcommand{\pinit}{\mbletter{\psi}_0}
\newcommand{\pkl}{\mbletter{P}}
\newcommandx{\pk}[2][2=]{
\ifthenelse{\equal{#2}{}}
{\pkl_{#1}}
{\pkl_{#1} \langle #2 \rangle}
}
\newcommandx{\pkjt}[2][1=]{
\ifthenelse{\equal{#1}{}}
{\mbletter{\psi}_{#2}}
{\mbletter{\psi}_{#2} \langle #1 \rangle}
}
\newcommand{\pot}[1]{g_{#1}}
\newcommand{\pothigh}[1]{\bar{\tau}_{#1}}
\newcommand{\potlow}[1]{\ubar{\tau}_{#1}}
\newcommand{\potmb}[1]{\mbletter{g}_{#1}}
\newcommand{\prob}{\mathbb{P}}
\newcommand{\probmeas}{\mathsf{M}_1}
\newcommand{\refm}[1]{\lambda_{#1}}
\newcommandx{\rk}[2][1=]{
\ifthenelse{\equal{#1}{}}
{B_{#2}}
{B_{#2} \langle #1 \rangle}
}
\newcommand{\rmd}{\mathrm{d}} 
\newcommand{\rset}{\mathbb{R}}
\newcommand{\stat}[2]{\beta_{#1}^{#2}}
\newcommand{\statl}{b}
\newcommandx{\statlmb}[1][1=]{\ifthenelse{\equal{#1}{}}{\mbletter{b}}{\mbletter{b}_{#1}}}
\newcommand{\statmb}[1]{\boldsymbol{\beta}_{#1}}

\newcommand{\stattd}[2]{\tilde{\beta}_{#1}^{#2}}

\newcommandx{\targ}[2][1=]{
\ifthenelse{\equal{#1}{}}
{\eta_{#2}}
{\eta_{#2} \langle #1 \rangle}
}
\newcommand{\targmb}[1]{\boldsymbol{\eta}_{#1}}

\newcommand{\tensprod}{\varotimes}
\newcommand{\ubar}[1]{\underaccent{\bar}{#1}}
\newcommand{\ud}[1]{q_{#1}}
\newcommandx{\uk}[2][1=]{
\ifthenelse{\equal{#1}{}}
{Q_{#2}}
{Q_{#2} \langle #1 \rangle}
}
\newcommand{\ukmb}[1]{\boldsymbol{Q}_{#1}}
\newcommandx{\ukp}[2][1=]{
\ifthenelse{\equal{#1}{}}
{\mathcal{Q}_{#2}}
{\mathcal{Q}_{#2} \langle #1 \rangle}
}
\newcommandx{\untarg}[2][1=]{
\ifthenelse{\equal{#1}{}}
{\gamma_{{#2}}}
{\gamma_{{#2}} \langle #1 \rangle}
}

\def\PARIS{\texttt{PARIS}}
\def\PPG{\texttt{PPG}}
\newcommand{\coint}[1]{\left[#1\right)}

\newcommand{\untargmb}[1]{\boldsymbol{\gamma}_{#1}}
\newcommand{\xsp}[1]{\mathsf{X}_{#1}}
\newcommand{\xfd}[1]{\mathcal{X}_{#1}}
\newcommand{\zsp}[1]{\mathsf{Z}_{#1}}
\newcommand{\zfd}[1]{\mathcal{Z}_{#1}}

\newcommand{\xfdmb}[1]{\boldsymbol{\mathcal{X}}_{#1}}
\newcommandx{\xmb}[1][1=]{\ifthenelse{\equal{#1}{}}{\mbletter{x}}{\mbletter{x}_{#1}}}
\newcommandx{\xmbtd}[1][1=]{\ifthenelse{\equal{#1}{}}{\tilde{\mbletter{x}}}{\tilde{\mbletter{x}}_{#1}}}
\newcommand{\xpfd}[2]{\xfd{#1:#2}}
\newcommand{\xpsp}[2]{\xsp{#1:#2}}
\newcommand{\xspmb}[1]{\boldsymbol{\mathsf{X}}_{#1}}
\newcommand{\yfd}[1]{\mathcal{Y}_{#1}}
\newcommand{\yfdmb}[1]{\boldsymbol{\mathcal{Y}}_{#1}}
\newcommand{\ymb}[1][1=]{\ifthenelse{\equal{#1}{}}{\mbletter{y}}{\mbletter{y}_{#1}}}
\newcommand{\ymbtd}[1][1=]{\ifthenelse{\equal{#1}{}}{\tilde{\mbletter{y}}}{\tilde{\mbletter{y}}_{#1}}}
\newcommand{\ysp}[1]{\mathsf{Y}_{#1}}
\newcommand{\yspmb}[1]{\boldsymbol{\mathsf{Y}}_{#1}}
\newcommand{\indi}[1]{\1_{{#1}}}
\newcommand{\indin}[1]{\1\left\{#1\right\}}
\newcommand{\chunk}[3]{#1_{{#2:#3}}}

\newcommandx{\CPE}[4][1=,2=]{{\mathbb E}^{#2}_{#1}\left[ #3 \mid #4 \right]}
\newcommandx{\cPE}[4][1=,2=]{{\mathbb E}^{#2}_{#1}[ #3 \mid #4 ]} 
\newcommandx{\CPP}[3][1=]{{\mathbb P}_{#1}\left(\left. #2 \, \right| #3 \right)} 
\newcommandx{\cPP}[3][1=]{{\mathbb P}_{#1}[ #2 | #3 ]} 
\newcommand{\gsupbound}[1]{\pothigh{#1}}
\newcommand{\ginfbound}[1]{\potlow{#1}}
\newcommand{\constpdist}{\mathsf{c}}
\newcommandx{\mserunconst}[2][1=n]{\upzeta^{\scriptsize{\mbox{{\it mse}}}}_{#1,#2}}
\newcommandx{\biasrunconst}[2][1=n]{\upzeta^{\scriptsize{\mbox{{\it bias}}}}_{#1,#2}}
\newcommandx{\cstcondparisbias}[1][1=n]{\bar{\mathsf{c}}_{#1}^{\scriptsize{\mbox{{\it bias}}}}}
\newcommandx{\cstparisbias}[1][1=n]{\mathsf{c}_{#1}^{\scriptsize{\mbox{{\it bias}}}}}
\newcommandx{\cstparismse}[1][1=n]{\mathsf{c}_{#1}^{\scriptsize{\mbox{{\it mse}}}}}
\newcommandx{\cstrollingbias}[1][1=n]{\mathsf{c}_{\scriptsize{\mbox{{\it roll}}},#1}^{\scriptsize{\mbox{{\it bias}}}}}
\newcommandx{\cstrollingmse}[1][1=n]{\mathsf{c}_{\scriptsize{\mbox{{\it roll}}},#1}^{\scriptsize{\mbox{{\it mse}}}}}
\newcommandx{\cstpariscov}[1][1=n]{\mathsf{c}_{#1}^{\scriptsize{\mbox{{\it cov}}}}}

\newcommand{\ki}{{k}}
\def\burningloss{\upsilon}
\newcommandx{\rollingestim}[4][1=N,2=K_0,3=K,4=f]{\Pi_{(#1,#2),#3}(#4)}
\def\totalbudget{C} 
\def\Id{\operatorname{Id}}

\usepackage[numbers]{natbib}

\numberwithin{equation}{section}

\begin{document}

\begin{frontmatter}
	\title{Particle-based, rapid incremental smoother meets particle Gibbs}
	\runtitle{Particle-based, rapid incremental smoother meets particle Gibbs}

	\begin{aug}
		\author[A, C]{\fnms{Gabriel} \snm{Cardoso}\ead[label=e1]{gabriel.victorino-cardoso@polytechnique.edu}},
		\author[A]{\fnms{Eric} \snm{Moulines}\ead[label=e2]{eric.moulines@polytechnique.edu}}
		\and
		\author[B]{\fnms{Jimmy} \snm{Olsson}\ead[label=e3]{jimmyol@kth.se}}
	
		\address[A]{Ecole polytechnique, Centre de Math{\'e}matiques Appliqu{\'e}es, Palaiseau, France \\ \printead{e1,e2}}
		\address[B]{Department of Mathematics, KTH Royal Institute of Technology, Stockholm, Sweden \\ \printead{e3}}
		\address[C]{IHU-Liryc, Université de Bordeaux, Pessac, France}
	\end{aug}
	
	\begin{abstract}
	The particle-based, rapid incremental smoother (\PARIS) is a sequential Monte Carlo technique allowing for efficient online approximation of expectations of additive functionals under Feynman--Kac path distributions. Under weak assumptions, the algorithm has linear computational complexity and limited memory requirements. It also comes with a number of non-asymptotic bounds and convergence results. However, being based on self-normalised importance sampling, the {\PARIS} estimator is biased; its bias is inversely proportional to the number of particles but has been found to grow linearly with the time horizon under appropriate mixing conditions. In this work, we propose the Parisian particle Gibbs (\PPG) sampler, whose complexity is essentially the same as that of the {\PARIS} and which significantly reduces the bias for a given computational complexity at the price of a modest increase in the variance. This method is a wrapper in the sense that it uses the {\PARIS}  algorithm in the inner loop of particle Gibbs to form a bias-reduced version of the targeted quantities. We substantiate the {\PPG} algorithm with  theoretical results, including new bounds on bias and variance as well as deviation inequalities. We illustrate our theoretical results with numerical experiments supporting our claims.
\end{abstract}
	
	
\end{frontmatter}

\section{Introduction}
\label{sec:introduction}
\emph{Feynman--Kac formulas} play a key role in various models used in statistics, physics, and many other fields; see \cite{delmoral:2004,delmoral:2013,chopin2020introduction} and the references therein.
Let $\{ (\xsp{n}, \xfd{n}) \}_{n \in \nset}$ be a sequence of general state spaces, and define, for every $n \in \nset$, $\xpsp{0}{n} \eqdef \prod_{m = 0}^n \xsp{m}$ and $\xpfd{0}{n} \eqdef \bigotimes_{m = 0}^n \xfd{m}$. For a sequence $\{ \mk{n} \}_{n \in \nset}$ of Markov kernels $\mk{n} : \xsp{n} \times \xfd{n + 1} \rightarrow [0, 1]$, an initial distribution $\init \in \probmeas(\xfd{0})$, and a sequence $\{ \pot{n} \}_{n \in \nset}$ of bounded measurable potential functions $\pot{n} : \xsp{n} \rightarrow \rset_+$, a sequence $\{ \targ{0:n} \}_{n \in \nset}$ of \emph{Feynman--Kac path measures} is defined by
\begin{equation}
\label{eq:normalized:F-K:measures} \targ{0:n} : \xpfd{0}{n} \ni A \mapsto \frac{\untarg{0:n}(A)}{\untarg{0:n}(\xpsp{0}{n})}, \quad n \in \nset,
\end{equation}
where
\begin{equation}
\label{eq:unnormalized:F-K:measures} \untarg{0:n} : \xpfd{0}{n} \ni A \mapsto \int \indi{A}(\chunk{x}{0}{n}) \, \targ{0}(\rmd x_0) \prod_{m = 1}^{n - 1} \uk{m}(x_m, \rmd x_{m + 1}) ,
\end{equation}
with
\begin{equation} \label{eq:unnormalised:kernel}
\uk{m} : \xsp{m} \times \xfd{m + 1} \ni (x, A) \mapsto \pot{m}(x) \mk{m}(x, A)
\end{equation}
being unnormalised kernels. By convention, $\targ{0:0} \eqdef \eta_0$.
Note that each $\targ{0:n}$ is a probability measure, while $\untarg{0:n}$ is not normalised. 
For every $n \in \nsetpos$ we also define the marginal distribution
$\targ{n} : \xfd{n} \ni A \mapsto \targ{0:n}(\xsp{0:n - 1} \times A)$.
In the context of nonlinear filtering in \emph{general state-space hidden Markov models}, $\targ{0:n}$ is the \emph{joint-smoothing distribution} and $\targ{n}$ is the \emph{filter distribution}; see \cite{delmoral:2004,cappe:moulines:ryden:2005,chopin2020introduction}.

For most problems of interest in practice, the Feynman--Kac path or marginal measures are intractable, and so is also any expectation associated with the same. Therefore, \emph{particle approximations} of such measures have been developed.
A particle filter approximates the flow $\left\{\targ{n}\right\}_{n \in \nset}$ of marginals by a sequence of occupation measures associated with a swarm of particles $\{\epart{n}{i}\}_{i=1}^{N}$, $n \in \nset$, each particle $\epart{n}{i}$ being a random draw in $\xsp{n}$. Particle filters revolve around two operations: a \emph{selection step} duplicating/discarding particles with large/small importance weights, respectively, and a \emph{mutation step} evolving randomly the selected particles in the state space. Applying alternatingly and iteratively selection and mutation results in a swarm of $\N$  particles that are both serially and spatially dependent. Feynman--Kac path models can also be interpreted as laws associated with a certain kind of Markovian backward dynamics; this interpretation is useful, for instance, for the smoothing problem in nonlinear filtering \cite{douc:garivier:moulines:olsson:2009,delmoral:doucet:singh:2010}.
Several convergence results have been proved when the number $\N$ of particles tends to infinity; see, \emph{e.g.}, \cite{delmoral:2004,douc:moulines:2008,delmoral:2013,chopin2020introduction}. A number of non-asymptotic results have also been established, including the bias of the particle approximation of the Feynman--Kac formula and associated chaos propagation. Extensions to backward interpretation can also be found in \cite{douc:garivier:moulines:olsson:2009,delmoral:doucet:singh:2010}.

In this paper, we focus on the problem of recursively computing smoothed expectations
\begin{equation*}
\targ{0:n} h_{n}=\int h_{n}(\chunk{x}{0}{n}) \, \targ{0:n}(\rmd \chunk{x}{0}{n}), \quad n \in \nset,
\end{equation*}
where we have introduced the vector notation $\chunk{x}{0}{n} = (x_0, \ldots, x_n) \in \xpsp{0}{n} \eqdef \xsp{0} \times \cdots \times \xsp{n}$,
for additive functionals $h_{n}$ in the form
\begin{equation} \label{eq:add:functional}
h_{n}(\chunk{x}{0}{n}) \eqdef \sum_{m=0}^{n-1} \tilde{h}_m(\chunk{x}{m}{m + 1}), \quad \chunk{x}{0}{n} \in \xpsp{0}{n}.
\end{equation}
In nonlinear filtering problems, such expectations appear in the context of maximum-likelihood parameter estimation, for instance, when computing the score function (the gradient of the log-likelihood function) or the EM surrogate; see \cite{cappe:2001,andrieu:doucet:2003,poyiadjis:doucet:singh:2005,cappe:2009,poyiadjis:doucet:singh:2011}.
In \cite{olsson:westerborn:2017}, an efficient \emph{particle-based, rapid incremental smoother} ({\PARIS}) was proposed. {\PARIS} has linear computational complexity in the number of particles under weak assumptions and limited memory requirements. 
An interesting feature of {\PARIS}, which samples on-the-fly from the backward dynamics induced by the particle filter, is that it requires two or more backward draws per particle to cope with the degeneracy of the sampled trajectories and remain numerically stable in the long run, with an asymptotic variance that grows only linearly with time.

In this paper, we highlight a method to reduce the bias of the {\PARIS} estimator of $\targ{0:n} h_n$. The idea is to mix the {\PARIS} algorithm (to introduce a conditional {\PARIS} algorithm) and a version of the particle Gibbs algorithm with backward sampling \cite{andrieu:doucet:holenstein:2010,lindsten:jordan:schoen:2014,chopin:singh:2015,del2016particle,del2018sharp}.
This leads to the {\PARIS}ian particle Gibbs (\PPG) algorithm, for which we provide an upper bound of the bias which decreases inversely proportionally to the number of particles and exponentially fast with the iteration index (under the assumption that the particle Gibbs sampler is uniformly ergodic).

The paper is structured as follows. In \Cref{sec:Feynman-Kac} we recall the Feynman--Kac model and its backward interpretation and introduce the particle Gibbs. Our presentation is  inspired by \cite{del2016particle}, but differs in that it avoids the use of quotient spaces, while it avoids the extension of the distribution to the ancestral indices of the original derivation in \cite{andrieu:doucet:holenstein:2010}. In \Cref{sec:paris} we introduce the {\PARIS} algorithm and its conditional version. We show how {\PARIS} can be coupled with the particle Gibbs method with backward sampling, yielding the {\PPG} algorithm. In \Cref{sec:main-results}, we present the central result of this paper, an upper bound on the bias of the {\PPG} estimator as a function of the number of particles and the iteration index of the Gibbs algorithm. In addition, we provide an upper bound on the mean-squared error (MSE). In \Cref{sec:main-proofs}, the most important and original proofs  are collected. In \Cref{sec:numerics}, we provide a limited Monte Carlo experiment to illustrate our results. Finally, \Cref{sec:algorithms} and \Cref{sec:proofs} contain some pseudocode and additional proofs of technical character, respectively.

\paragraph{Notation.}
Let $\rset_+ \eqdef \coint{0, \infty}$, $\rset_+^* \eqdef (0, \infty)$, $\nset \eqdef \{0,1,2, \ldots\}$, and $\nset^* \eqdef \{1,2,3, \ldots\}$ denote the sets of nonnegative and positive real numbers and integers, respectively. 
We denote by $\Id_{\N}$ the $\N \times \N$ identity matrix. For any quantities $\left\{a_{\ell}\right\}_{\ell=m}^{n}$ we denote vectors as $\chunk{a}{m}{n} \eqdef (a_m, \ldots, a_n)$ and for any $(m, n) \in \nset^2$ such that $m \leq n$ we let $\intvect{m}{n} \eqdef \{m, m+1, \ldots, n\}$. For a given measurable space $(\xsp{}, \xfd{})$, where $\xfd{}$ is a countably generated $\sigma$-algebra, we denote by $\bmf(\xfd{})$ the set of bounded $\xfd{} / \mathcal{B}(\rset)$-measurable functions on $\xsp{}$. For any $h \in \bmf(\xfd{})$, we let
$\| h \|_\infty \eqdef \sup _{x \in \xsp{}}|h(x)|$ and $\operatorname{osc}(h) \eqdef \sup _{(x, x') \in \xsp{}^2}| h(x) - h(x')|$ denote the supremum and oscillator norms of $h$, respectively. Let $\meas(\xfd{})$ be the set of $\sigma$-finite measures on $(\xsp{}, \xfd{})$ and $\probmeas(\xfd{}) \subset \meas(\xfd{})$ the probability measures.

Let $(\ysp{}, \yfd{})$ be another measurable space. A possibly unnormalised transition kernel $K$ on $\xsp{} \times \yfd{}$ induces two integral operators, one acting on measurable functions and the other on measures; more specifically, for $h \in \bmf(\xfd{} \otimes \yfd{})$ and $\nu \in \probmeas(\xfd{})$, define the measurable function
$$
	K h : \xsp{} \ni x \mapsto \int h(x, y) \, K(x, \rmd y)
$$
and the measure
$$
	\nu K : \yfd{} \ni A \mapsto \int K(x, A) \, \nu(\rmd x),
$$
whenever these quantities are well defined. Now, let $(\zsp{}, \zfd{})$ be a third measurable space and $L$ another possibly unnormalised transition kernel on $\ysp{} \times \zfd{}$; we then define, with $K$ as above, two different products of $K$ and $L$, namely
$$
	K L: \xsp{} \times \zfd{} \ni (x, A) \mapsto \int  L(y, A) \, K(x, \rmd y)
$$
and
$$
	K \tensprod L: \xsp{} \times(\yfd{} \tensprod \zfd{}) \ni (x, A) \mapsto \iint \indi{A}(y, z) \, K(x, \rmd y) \, L(y, \rmd z),
$$
whenever these are well defined. This also defines the $\tensprod$ products of a kernel $K$ on $\xsp{} \times \yfd{}$ and a measure $\nu$ on $\xfd{}$ as well as of a kernel $L$ on $\ysp{} \times \xfd{}$ and a measure $\mu$ on $\yfd{}$ as the measures  
\begin{align*}
\nu \tensprod K &: \xfd{} \tensprod \yfd{} \ni A \mapsto \iint \indi{A}(x, y) \, K(x, \rmd y) \, \nu(\rmd x), \\
L \tensprod \mu &: \xfd{} \tensprod \yfd{} \ni A \mapsto \iint \indi{A}(x, y) \, L(y, \rmd x) \, \mu(\rmd y).
\end{align*}

\section{Feynman--Kac models}
\label{sec:Feynman-Kac}
In the next sections, we recall \emph{Feynman--Kac models}, \emph{many-body Feynman--Kac models}, \emph{backward interpretations}, and \emph{conditional dual processes}. Our presentation follows closely \cite{del2016particle} but with a different and hopefully more transparent definition of the many-body extensions. We restate (in \Cref{thm:duality} below) a duality formula of \cite{del2016particle} relating these concepts. This formula provides a foundation for the \emph{particle Gibbs sampler} described in \Cref{sec:dual:and:Gibbs} and is pivotal for the coming developments.
\subsection{Many-body Feynman--Kac models}
\label{sec:mb:FK:models}
In the following we assume that all random variables are defined on a common probability space $(\Omega, \mathcal{F}, \prob)$. The distribution flow $\{ \targ{m} \}_{m \in \nset}$ is intractable in general, but can be approximated  by random samples $\epartmb{m} = \{ \epart{m}{i} \}_{i = 1}^\N$, $m \in \nset$, of \emph{particles},
where $\N \in \nsetpos$ is a fixed Monte Carlo sample size and each particle $\epart{m}{i}$ is an $\xsp{m}$-valued random variable. Such particle approximation is based on the recursion $\targ{m + 1} = \pd{m}(\targ{m})$, $m \in \nset$, where $\pd{m}$ denotes the mapping
\begin{equation} \label{eq:def:pd:mapping}
\pd{m} : \probmeas(\xfd{m}) \ni \eta \mapsto \frac{\eta \uk{m}}{\eta \pot{m}}
\end{equation}
taking on values in $\probmeas(\xfd{m + 1})$. In order to describe recursively the evolution of the particle population, let $m \in \nset$ and assume that the particles $\epartmb{m}$ form a consistent approximation of $\targ{m}$ in the sense that $\occm(\epartmb{m}) h$, where $\occm(\epartmb{m}) \eqdef \sum_{i = 1}^\N \delta_{\epart{m}{i}} / \N$ (where $\delta_x$ denotes the Dirac measure located at $x$) is the occupation measure formed by $\epartmb{m}$, which serves as a proxy for $\targ{m} h$ for any $\targ{m}$-integrable test function $h$. (Under general conditions, $\occm(\epartmb{m}) h$ converges in probability to $\targ{m}$ with $N \rightarrow \infty$; see \cite{delmoral:2004,chopin2020introduction} and references therein.) Then, in order to generate an updated particle sample  approximating $\targ{m + 1}$, new particles $\epartmb{m + 1} = \{ \epart{m + 1}{i} \}_{i = 1}^\N$ are drawn conditionally independently given $\epartmb{m}$ according to
$$
\epart{m + 1}{i} \sim \pd{m}(\occm(\epartmb{m})) = \sum_{\ell = 1}^\N \frac{\pot{m}(\epart{m}{\ell})}{\sum_{\ell' = 1}^\N \pot{m}(\epart{m}{\ell'})} \mk{m}(\epart{m}{\ell}, \cdot), \quad i \in \intvect{1}{\N}.
$$
Since this process of particle updating involves sampling from the mixture distribution $ \pd{m}(\occm(\epartmb{m}))$, it can be naturally decomposed into two substeps: \emph{selection} and \emph{mutation}. The selection step consists of randomly choosing the $\ell$-th mixture stratum with probability $\pot{m}(\epart{m}{\ell}) / \sum_{\ell' = 1}^\N \pot{m}(\epart{m}{\ell'})$ and the mutation step consists of drawing a new particle $\epart{m + 1}{i}$ from the selected stratum $\mk{m}(\epart{m}{\ell}, \cdot)$. In \cite{del2016particle}, the term \emph{many-body Feynman--Kac models} is related to the law of process $\{ \epartmb{m} \}_{m \in \nset}$. For all $m \in \nset$, let $\xspmb{m} \eqdef \xsp{m}^\N$ and $\xfdmb{m} \eqdef \xfd{m}^{\varotimes \N}$; then $\{ \epartmb{m} \}_{m \in \nset}$ is an inhomogeneous Markov chain on $\{ \xspmb{m} \}_{m \in \nset}$ with transition kernels
$$
\mkmb{m} : \xspmb{m} \times \xfdmb{m + 1} \ni (\xmb_m, A) \mapsto \pd{m}(\occm(\xmb_m))^{\tensprod \N}(A)
$$
and initial distribution $\initmb = \init^{\tensprod \N}$. Now, denote $\xspmb{0:n} \eqdef \prod_{m = 0}^n \xspmb{m}$ and $\xfdmb{0:n} \eqdef \bigotimes_{m = 0}^n \xfdmb{m}$. (Here and in the following we use a bold symbol to stress that a quantity is related to the many-body process.) The \emph{many-body Feynman--Kac path model} refers to the flows $\{ \untargmb{m} \}_{m \in \nset}$ and $\{ \targmb{m} \}_{m \in \nset}$ of the unnormalised and normalised, respectively, probability distributions on $\{ \xfdmb{0:m} \}_{m \in \nset}$ generated by \eqref{eq:normalized:F-K:measures} and \eqref{eq:unnormalized:F-K:measures} for the Markov kernels $\{ \mkmb{m} \}_{m \in \nset}$, the initial distribution $\initmb$, the potential functions
$$
\potmb{m} : \xspmb{m} \ni \xmb_m \mapsto \occm(\xmb_m) \pot{m} = \frac{1}{N} \sum_{i = 1}^\N \pot{m}(x_m^i), \quad m \in \nset,
$$
and the corresponding unnormalised transition kernels
$$
\ukmb{m} : \xspmb{m} \times \xfdmb{m + 1} \ni (\xmb[m], A) \mapsto \potmb{m}(\xmb[m]) \mkmb{m} (\xmb[m], A), \quad m \in \nset.
$$

\subsection{Backward interpretation of Feynman--Kac path flows}
Suppose that each kernel $\uk{n}$, $n \in \nset$, defined in \eqref{eq:unnormalised:kernel}, has a transition density $\ud{n}$ with respect to some dominating measure $\refm{n+1} \in \meas(\xfd{n + 1})$. Then for $n \in \nset$ and $\eta \in \probmeas(\xfd{n})$ we may define the \emph{backward kernel}
\begin{equation} \label{def:backward:kernel}
\bk{n, \eta} : \xsp{n + 1} \times \xfd{n} \ni (x_{n + 1}, A) \mapsto \frac{\int \1_A(x_n) \ud{n}(x_n, x_{n + 1}) \, \eta(\rmd x_n)}{\int \ud{n}(x_n', x_{n + 1}) \, \eta(\rmd x_n')}.
\end{equation}
Now, denoting, for $n \in \nsetpos$,
\begin{equation}
\label{eq:bckwd-kernels}
\rk{n} : \xsp{n} \times \xfd{0:n - 1} \ni (x_n, A) \mapsto \idotsint \1_A(x_{0:n - 1}) \prod_{m = 0}^{n - 1} \bk{m, \targ{m}}(x_{m + 1}, \rmd x_m),
\end{equation}
we may state the following---now classical---\emph{backward decomposition} of the Feynman--Kac path measures, a result that will play a pivotal role in the following.

\begin{proposition} \label{prop:backward:decomposition}
For every $n \in \nsetpos$ it holds that $\untarg{0:n} = \untarg{n} \tensprod \rk{n}$ and $\targ{0:n} = \targ{n} \tensprod \rk{n}$.
\end{proposition}

Although the decomposition in \Cref{prop:backward:decomposition} is well known (see, \eg, \cite{delmoral:doucet:singh:2010,del2016particle}), we provide a proof in \Cref{sec:proof:backward:decomposition} for completeness. Using the backward decomposition, one can obtain a particle approximation of a given Feynman--Kac path measure $\targ{0:n}$ by first sampling, in an initial forward pass, particle clouds $\{ \epartmb{m} \}_{m = 0}^n$ from $\initmb \tensprod \mkmb{0} \tensprod \cdots \tensprod \mkmb{n - 1}$ and then sampling, in a subsequent backward pass, say $N$ conditionally independent paths $\{\eparttd{0:n}{i} \}_{i = 1}^\N$ from $\bdpart{n}(\epartmb{0}, \ldots, \epartmb{n}, \bcdot)$, where
\begin{equation} \label{eq:def:bdpart}
\bdpart{n} : \xspmb{0:n} \times \xfd{0:n} \ni (\xmb[0:n], A) \mapsto \idotsint  \1_A(x_{0:n}) \left( \prod_{m = 0}^{n - 1} \bk{m,\occm(\xmb[m])}(x_{m + 1}, \rmd x_m) \right) \occm(\xmb[n])(\rmd x_n)
\end{equation}
is a Markov kernel describing the time-reversed dynamics induced by the particle approximations generated in the forward pass. (Here and in the following we use blackboard notation to denote kernels related to many-body path spaces.) Finally, $\occm(\{\eparttd{0:n}{i} \}_{i = 1}^\N) h$ is returned as an estimator of $\targ{0:n} h$ for any $\targ{0:n}$-integrable test function $h$. This algorithm is in the literature referred to as the \emph{forward--filtering backward--simulation} (FFBSi) \emph{algorithm} and was introduced in \cite{godsill:doucet:west:2004}; see also \cite{cappe:godsill:moulines:2007,douc:garivier:moulines:olsson:2009}.
More precisely, given the forward particles $\{ \epartmb{m} \}_{m = 0}^n$, each path $\eparttd{0:n}{i}$ is generated by first drawing $\eparttd{n}{i}$ uniformly among the particles $\epartmb{n}$ in the last generation and then drawing, recursively,
\begin{equation} \label{eq:backward:sampling:operation}
\eparttd{m}{i} \sim \bk{m,\occm(\epartmb{m})}(\eparttd{m + 1}{i}, \cdot) = \sum_{j = 1}^\N \frac{\ud{m}(\epart{m}{j}, \eparttd{m + 1}{i})}{\sum_{\ell=1}^{N} \ud{m}(\epart{m}{\ell}, \eparttd{m + 1}{i})} \delta_{\epart{m}{j}},
\end{equation}
\emph{i.e.}, given $\eparttd{m + 1}{i}$, $\eparttd{m}{i}$ is picked at random among the $\epartmb{m}$ according to weights proportional to $\{ \ud{m}(\epart{m}{j}, \eparttd{m + 1}{i}) \}_{j = 1}^\N$.
Note that in this basic formulation of the FFBSi algorithm, each backward-sampling operation \eqref{eq:backward:sampling:operation} requires the computation of the normalising constant $\sum_{\ell=1}^{N} \ud{m}(\epart{m}{\ell}, \eparttd{m + 1}{i})$, which implies an overall quadratic complexity of the algorithm. Still, this heavy computational burden can eased by means of an effective accept--reject technique discussed in \Cref{sec:std:paris}.

\subsection{Conditional dual processes and particle Gibbs}
\label{sec:dual:and:Gibbs}
The \emph{dual process} associated with a given Feynman--Kac model (\ref{eq:normalized:F-K:measures}--\ref{eq:unnormalized:F-K:measures}) and a given trajectory $\{ z_n \}_{n \in \nset}$, where $z_n \in \xsp{n}$ for every $n \in \nset$, is defined as the canonical Markov chain with kernels
\begin{multline} \label{eq:kernel:cond:dual}
\mkmb{n}[z_{n + 1}] : \xspmb{n} \times \xfdmb{n + 1} \ni \\
(\xmb_n, A) \mapsto \frac{1}{\N} \sum_{i = 0}^{\N - 1} \left( \Phi_n(\occm(\xmb_n))^{\tensprod i} \tensprod \delta_{z_{n + 1}} \tensprod \Phi_n(\occm(\xmb_n))^{\tensprod (\N - i - 1)} \right)(A),
\end{multline}
for $n \in \nset$, and initial distribution
\begin{equation} \label{eq:init:cond:dual}
\initmb[z_0] \eqdef \frac{1}{\N} \sum_{i = 0}^{\N - 1} \left( \init^{\tensprod i} \tensprod \delta_{z_0} \tensprod \init^{\tensprod (\N - i - 1)} \right).
\end{equation}
As clear from (\ref{eq:kernel:cond:dual}--\ref{eq:init:cond:dual}), given $\{ z_n \}_{n \in \nset}$, a realisation $\{ \epartmb{n} \}_{n \in \nset}$ of the dual process is generated as follows. At time zero, the process is initialised by inserting $z_0$ at a randomly selected position in the vector $\epartmb{0}$ while drawing independently the remaining elements in the same vector from $\targ{0}$. After this, the process proceeds in a Markovian manner by, given $\epartmb{n}$, inserting $z_{n + 1}$ at a randomly selected position in $\epartmb{n + 1}$ while drawing independently the remaining elements from $\Phi_n(\occm(\epartmb{n}))$.

In order to describe compactly the law of the conditional dual process we define the Markov kernel
$$
\mbjt{n} : \xsp{0:n} \times \xfdmb{0:n} \ni (\chunk{z}{0}{n}, A) \mapsto \initmb[z_0] \tensprod \mkmb{0}[z_1] \tensprod \cdots \tensprod \mkmb{n - 1}[z_n](A).
$$
The following result elegantly combines the underlying model (\ref{eq:normalized:F-K:measures}--\ref{eq:unnormalized:F-K:measures}), the many-body Feynman--Kac model, the backward decomposition, and the conditional dual process.
\begin{theorem}[\cite{del2016particle}] \label{thm:duality}
For all $n \in \nset$ it holds that
\begin{equation}  \label{eq:duality}
\bdpart{n} \tensprod \untargmb{0:n}  = \untarg{0:n} \tensprod \mbjt{n}.
\end{equation}
\end{theorem}
In \cite{del2016particle}, each state $\epartmb{n}$ of the many-body process maps an outcome $\omega$ of the sample space $\Omega$ into an \emph{unordered set} of $N$ elements in $\xsp{n}$. However, we have chosen to let each $\epartmb{n}$ take on values in the standard \emph{product space} $\xsp{n}^\N$ for two reasons: first, the construction of \cite{del2016particle} requires sophisticated measure-theoretic arguments to endow such unordered sets with suitable $\sigma$-fields and appropriate measures; second, we see no need to ignore the index order of the particles as long as the Markovian dynamics (\ref{eq:kernel:cond:dual}--\ref{eq:init:cond:dual}) of the conditional dual process is symmetrised over the particle cloud. Therefore, in \Cref{sec:proof:thm:duality}, we include our own proof of duality \eqref{eq:duality} for completeness. Note that the measure \eqref{eq:duality} on $\xfd{0:n} \tensprod \xfdmb{0:n}$ is unnormalised, but since the kernels $\bdpart{n}$ and $\mbjt{n}$ are both Markov, normalising the identity with $\untarg{0:n}(\xsp{0:n}) = \untargmb{0:n}(\xspmb{0:n})$ yields immediately
\begin{equation} \label{eq:duality:normalised}
\bdpart{n} \tensprod \targmb{0:n}   = \targ{0:n} \tensprod \mbjt{n}.
\end{equation}
Since the two sides of \eqref{eq:duality:normalised} provide the full conditionals, it is natural to take a data-augmentation approach and sample the target \eqref{eq:duality:normalised} using a two-stage deterministic-scan Gibbs sampler \cite{andrieu:doucet:holenstein:2010,chopin:singh:2015}. More specifically, assume that we have generated a state
$(\epartmb{0:n}[\ell], \zeta_{0:n}[\ell])$ comprising a dual process with associated path on the basis of $\ell \in \nset$ iterations of the sampler; then the next state $(\epartmb{0:n}[\ell+1], \zeta_{0:n}[\ell+1])$ is generated in a Markovian fashion by, first, sampling $\epartmb{0:n}[\ell+1] \sim \mbjt{n}(\zeta_{0:n}[\ell], \cdot)$ and, second, sampling $\zeta_{0:n}[\ell+1] \sim \bdpart{n}(\epartmb{0:n}[\ell+1], \bcdot)$. After arbitrary initialisation (and the discard of possible burn-in), this procedure produces a Markov trajectory $\{ (\epartmb{0:n}[\ell], \chunk{\zeta}{0}{n}[\ell]) \}_{\ell \in \nset}$, and under weak additional technical conditions this Markov chain admits \eqref{eq:duality:normalised} as its unique
invariant distribution. In such case, the Markov chain is ergodic \cite[Chapter~5]{douc2018markov}, and the marginal distribution of the conditioning path $\chunk{\zeta}{0}{n}[\ell]$ converges to the target distribution $\targ{0:n}$. Therefore, for every $h \in \bmf(\xfd{0:n})$, it holds that 
$$
\lim_{L \to \infty} \frac{1}{L} \sum_{\ell=1}^L h(\chunk{\zeta}{0}{n}[\ell]) =
\targ{0:n} h, \quad \prob\mbox{-a.s.}
$$

\subsection{The {\PARIS} algorithm}
\label{sec:std:paris}
In the following we assume that we are given a sequence $\{ h_n \}_{n \in \nset}$ of \emph{additive state functionals} of type \eqref{eq:add:functional}. This problem is particularly relevant in the context of maximum-likelihood-based parameter estimation in general state-space models, \emph{e.g.}, when computing the \emph{score-function} (the gradient of the log-likelihood function) via the Fisher identity or when computing the intermediate quantity
of the \emph{expectation--maximization} (EM) \emph{algorithm}, in which case $\targ{0:n}$ and $h_n$ correspond to the joint state posterior and an element of some sufficient statistic, respectively; see \cite{cappe:moulines:2005,douc:garivier:moulines:olsson:2009,delmoral:doucet:singh:2010,poyiadjis:doucet:singh:2011,olsson:westerborn:2017}
and the references therein.

Interestingly, as noted in \cite{cappe:2009,delmoral:doucet:singh:2010}, the backward decomposition allows, when applied to additive state functionals, a forward recursion for the expectations $\{ \targ{0:n} \af{n} \}_{n \in \nset}$. More specifically, using the forward decomposition $\af{n + 1}(\chunk{x}{0}{n + 1}) = \af{n}(\chunk{x}{0}{n}) + \afterm{n}(x_n, x_{n + 1})$ and the backward kernel $\rk{n + 1}$ defined in \eqref{eq:bckwd-kernels}, we may write, for $x_{n + 1} \in \xsp{n + 1}$,
\begin{align}
\rk{n + 1} \af{n + 1}(x_{n + 1})  &= \int \bk{n, \targ{n}}(x_{n + 1}, \rmd x_n) \int \left( \af{n}(x_{0:n}) + \afterm{n}(x_n, x_{n + 1}) \right) \rk{n}(x_n, \rmd x_{0:n - 1}) \nonumber \\
&= \bk{n, \targ{n}}(\rk{n} \af{n} + \afterm{n})(x_{n + 1}), \label{eq:backward:forward:recursion}
\end{align}
which by \Cref{prop:backward:decomposition} implies that
\begin{equation}
\label{eq:additive-smoothing-recursion}
\targ{0:n + 1} h_{n + 1} = \targ{n + 1} \bk{n, \targ{n}}(\rk{n} \af{n} + \afterm{n}).
\end{equation}
Since, as we have seen, the marginal flow $\{\targ{n} \}_{n \in \nset}$ can be expressed recursively via the mappings $\{\pd{n} \}_{n \in \nset}$, \eqref{eq:additive-smoothing-recursion} provides, in principle, a basis for online computation of $\{ \targ{0:n} \af{n} \}_{n \in \nset}$. To handle the fact that the marginals are generally intractable we may, following \cite{delmoral:doucet:singh:2010}, plug particle approximations $\occm(\epartmb{n + 1})$ and $\bk{n,\occm(\epartmb{n})}$ (see \eqref{eq:backward:sampling:operation}) of $\targ{n + 1}$ and $\bk{n, \occm(\targ{n})}$, respectively, into the recursion \eqref{eq:additive-smoothing-recursion}. More precisely, we proceed recursively and assume that we at time $n$ have at hand a sample $\{(\epart{n}{i}, \stat{n}{i})\}_{i=1}^\N$ of particles with associated statistics, where each statistic $\stat{n}{i}$ serves as an approximation of $\rk{n} \af{n}(\epart{n}{i})$; then evolving the particle cloud according to $\epartmb{n + 1} \sim \mkmb{n}(\epartmb{n}, \cdot)$ and updating the statistics using \eqref{eq:backward:forward:recursion}, with $\bk{n, \targ{n}}$ replaced by $\bk{n,\occm(\epartmb{n})}$, yields the particle-wise recursion
\begin{equation} \label{eq:FFBSm:forward:update}
\stat{n + 1}{i} = \sum_{\ell = 1}^\N \frac{\ud{n}(\epart{n}{\ell}, \epart{n + 1}{i})}{\sum_{\ell' = 1}^\N \ud{n}(\epart{n}{\ell'}, \epart{n + 1}{i})} \left( \stat{n}{\ell} + \afterm{n}(\epart{n}{\ell}, \epart{n + 1}{i}) \right), \quad i \in \intvect{1}{\N},
\end{equation}
and, finally, the estimator
\begin{equation}
\label{eq:bckwd-interpretation}
\occm(\statmb{n})(\operatorname{id}) = \frac{1}{N} \sum_{i=1}^\N \stat{n}{i}
\end{equation}
of $\targ{0:n} \af{n}$, where we have set $\statmb{n} \eqdef (\stat{n}{1}, \ldots, \stat{n}{\N})$, $i \in \intvect{1}{\N}$. The procedure is initialised by simply letting $\stat{0}{i}=0$ for all $i \in \intvect{1}{N}$.
Note that \eqref{eq:bckwd-interpretation} provides a particle interpretation of the backward decomposition in \Cref{prop:backward:decomposition}. This algorithm is a special case of the \emph{forward--filtering backward--smoothing} (FFBSm) \emph{algorithm} (see \cite{andrieu:doucet:2003,godsill:doucet:west:2004,douc:garivier:moulines:olsson:2009,sarkka2013bayesian}) for additive functionals satisfying \eqref{eq:add:functional}. It allows for online processing of the sequence $\{\targ{0:n} \af{n}\}_{n \in \nset}$, but has also the appealing property that only the current particles $\epartmb{n}$ and statistics $\statmb{n}$ need to be stored in the memory. However, since each update \eqref{eq:FFBSm:forward:update} requires the summation of $\N$ terms, the scheme has an overall \emph{quadratic} complexity in the number of particles, leading to a computational bottleneck in applications to complex models that require large particle sample sizes $\N$.

In order to detour the computational burden of this forward-only implementation of {FFBSm}, the {\PARIS} algorithm \cite{olsson:westerborn:2017} updates the statistics $\statmb{n}$ by replacing each sum \eqref{eq:FFBSm:forward:update} by a Monte Carlo estimate
\begin{equation}
\label{eq:paris-update}
\stat{n+1}{i} = \frac{1}{\M} \sum_{j=1}^\M \left( \stattd{n}{i, j} + \afterm{n}(\eparttd{n}{i, j}, \epart{n+1}{i}) \right), \quad i \in \intvect{1}{N},
\end{equation}
where $\{(\eparttd{n}{i, j}, \stattd{n}{i, j})\}_{j = 1}^\M$ are drawn randomly among $\{(\epart{n}{i}, \stat{n}{i})\}_{i = 1}^\N$ with replacement, by assigning $(\eparttd{n}{i, j}, \stattd{n}{i, j})$ the value of $(\epart{n}{\ell}, \stat{n}{\ell})$ with probability $\ud{n}(\epart{n}{\ell}, \epart{n + 1}{i}) / \sum_{\ell' = 1}^\N \ud{n}(\epart{n}{\ell'}, \epart{n + 1}{i})$, and the Monte Carlo sample size $\M \in \nsetpos$ is supposed to be much smaller than $\N$ (say, less than $5$). Formally,
$$
\{(\eparttd{n}{i, j}, \stattd{n}{i, j})\}_{j = 1}^\M \sim \left( \sum_{\ell = 1}^\N \frac{\ud{n}(\epart{n}{\ell}, \epart{n + 1}{i})}{\sum_{\ell' = 1}^\N \ud{n}(\epart{n}{\ell'}, \epart{n + 1}{i})} \delta_{(\epart{n}{\ell}, \stat{n}{\ell})} \right)^{\tensprod \M}, \quad i \in \intvect{1}{\N}.
$$
The resulting procedure, summarised in \Cref{alg:paris}, allows for online processing with constant memory requirements, since it only needs to store the current particle cloud and the estimated auxiliary statistics at each iteration. Moreover, in the case where the Markov transition densities of the model can be uniformly bounded, \emph{i.e.}, there exists, for every $n \in \nset$, an upper bound $\mdhigh{n} > 0$ such that for all $(x_n, x_{n + 1}) \in \xsp{n} \times \xsp{n + 1}$, $\md{n}(x_n, x_{n + 1}) \leq \mdhigh{n}$ (a weak assumption satisfied for most models of interest), a sample $(\eparttd{n}{i, j}, \stat{n}{i, j})$ can be generated by drawing, with replacement and until acceptance, candidates $(\eparttd{n}{i, \ast}, \stattd{n}{i, \ast})$ from $\{(\epart{n}{i}, \stat{n}{i})\}_{i = 1}^\N$ according to the normalised particle weights $\{\pot{n}(\epart{n}{\ell}) / \sum_{\ell'} \pot{n}(\epart{n}{\ell'})\}_{\ell = 1}^\N$ (obtained as a by-product in the generation of $\epartmb{n + 1}$) and accepting the same with probability $\md{n}(\eparttd{n}{i, \ast}, \epart{n + 1}{i}) /  \mdhigh{n}$. As this sampling procedure bypasses completely the calculation of the normalising constant $\sum_{\ell' = 1}^\N \ud{n}(\epart{n}{\ell'}, \epart{n + 1}{i})$ of the targeted categorical distribution, it yields an overall $\mathcal{O}(\M \N)$ complexity of the algorithm as a whole; see \cite{douc:garivier:moulines:olsson:2009} for details.

Increasing $\M$ improves the accuracy of the algorithm at the cost of additional computational complexity. As shown in \cite{olsson:westerborn:2017}, there is a qualitative difference between the cases $\M=1$ and $\M \geq 2$, and it turns out that the latter is required to keep \PARIS\  numerically stable. More precisely, in the latter case, it can be shown that the \PARIS\  estimator $\occm(\statmb{n})$ satisfies, as $\N$ tends to infinity while $\M$ is held fixed, a central limit theorem (CLT) at the rate $\sqrt{N}$ and with an $n$-normalised asymptotic variance of order $\mathcal{O}(1 - 1/(\M - 1))$. As clear from this bound, using a large $\M$ is only to waste the computational work, and setting $\M$ to $2$ or $3$ typically works well in practice.

\section{The Parisian particle Gibbs (\PPG) sampler}
\label{sec:paris}
We now introduce the \emph{Parisian particle Gibbs} (\PPG) \emph{algorithm}.
For all $n \in \nsetpos$, let $\ysp{n} \eqdef \xsp{0:n} \times \rset$ and $\yfd{n} \eqdef \xfd{0:n} \tensprod \borel(\rset)$. Moreover, let $\ysp{0} \eqdef \xsp{0} \times \{ 0 \}$ and $\yfd{0} \eqdef \xfd{0} \tensprod \{ \{0\}, \emptyset \}$. An element of $\ysp{n}$ will always be denoted by $y_n = (x_{0:n|n}, b_n)$. The Parisian particle Gibbs sampler comprises, as a key ingredient, a \emph{conditional \PARIS\  step}, which updates recursively a set of $\ysp{n}$-valued random variables
$\bpart{n}{i} \eqdef (\epart{0:n|n}{i}, \stat{n}{i})$, $i \in \intvect{1}{\N}$.
Let $(\bpartmb{n})_{n \in \nset}$ denote the corresponding many-body process, each $\bpartmb{n} \eqdef \{(\epart{0:n|n}{i}, \stat{n}{i})\}_{i = 1}^\N$  taking on values in the space $\yspmb{n} \eqdef \ysp{n}^\N$, which we furnish with a $\sigma$-field $\yfdmb{n} \eqdef \yfd{n}^{\tensprod \N}$. The space $\yspmb{0}$ and the corresponding $\sigma$-field $\yfdmb{0}$ are defined accordingly.
For every $n \in \nset$, we write $\epartmb{0:n|n}$ for the collection $\{ \epart{0:n|n}{i} \}_{i = 1}^\N$ of paths in $\bpartmb{n}$, and $\epartmb{n|n}$ for the collection $\{ \epart{n|n}{i} \}_{i = 1}^\N$ of end points of the same.

In the following we let $n \in \nset$ be a fixed time horizon, and describe in detail how the \PPG\ approximates $\targ{0:n} \af{n}$ iteratively. In short, at each iteration $\ell$, the \PPG\ produces, given an input conditional path $\chunk{\zeta}{0}{n}[\ell]$, a many-body system $\bpartmb{n}[\ell+1]$ by means of a series of conditional \PARIS\  operations; after this, an updated path $\chunk{\zeta}{0}{n}[\ell+1]$, serving as input at the next iteration, is generated by picking one of the paths $\epartmb{0:n|n}[\ell+1]$ in $\bpartmb{n}[\ell+1]$ at random. At each iteration, the produced statistics $\statmb{n}$ (in $\bpartmb{n}$) provides an approximation of $\targ{0:n} \af{n}$ according to \eqref{eq:bckwd-interpretation}.

More precisely, given the path $\chunk{\zeta}{0}{n}[\ell]$, the conditional \PARIS\  operations are executed as follows. In the initial step, $\epartmb{0|0}[\ell+1]$ are drawn from $\initmb[{\zeta_0[\ell]}]$ defined in \eqref{eq:init:cond:dual} and $\bpart{0}{i}[\ell+1] \gets (\epart{0|0}{i}[\ell+1], 0)$ for all $i \in \intvect{1}{\N}$; then, recursively for $m \in \intvect{0}{n}$, assuming access to $\bpartmb{m}[\ell+1]$, we 
\begin{itemize}[noitemsep]
\item[\textbf{(1)}] generate an updated particle cloud $\epartmb{m + 1}[\ell+1] \sim \mkmb{m}[{\zeta_{m + 1}[\ell]}](\epartmb{m|m}[\ell+1], \cdot)$,
\item[\textbf{(2)}] pick at random, for each $i \in \intvect{1}{N}$, an ancestor path with associated statistics $(\tilde{\xi}_{0:m}^{i, 1}[\ell+1], \stattd{m}{i, 1}[\ell+1])$ among $\bpartmb{m}[\ell+1]$ by  drawing 
$$
(\tilde{\xi}_{0:m}^{i, 1}[\ell+1], \stattd{m}{i, 1}[\ell+1]) \sim
\sum_{s = 1}^\N \frac{\ud{m}(\epart{m|m}{s}[\ell+1], \epart{m + 1}{i}[\ell+1])}{\sum_{s' = 1}^\N \ud{m}(\epart{m|m}{s'}[\ell+1], \epart{m + 1}{i}[\ell+1])} \delta_{\bpart{m}{s}[\ell+1]}, \quad i \in \intvect{1}{N},
$$
\item[\textbf{(3)}] draw, with replacement, $\M - 1$ ancestor particles and associated statistics $\{ (\eparttd{m}{i, j}[\ell+1], \stattd{m}{i, j}[\ell+1]) \}_{j = 2}^\M$ at random from $\{(\epart{m|m}{s}[\ell+1], \stat{m}{s})[\ell+1]\}_{s = 1}^\N$ according to
\begin{multline*}
\{ (\eparttd{m}{i, j}[\ell+1], \stattd{m}{i, j}[\ell+1]) \}_{j = 2}^\M \\
\sim \left( \sum_{s = 1}^\N \frac{\ud{m}(\epart{m|m}{s}[\ell+1], \epart{m + 1}{i}[\ell+1])}{\sum_{s' = 1}^\N \ud{m}(\epart{m|m}{s'}[\ell+1], \epart{m + 1}{i}[\ell+1])} \delta_{(\epart{m|m}{s}[\ell+1], \stat{m}{s}[\ell+1])} \right)^{\tensprod (\M - 1)},
\end{multline*}
\item[\textbf{(4)}] set, for all $i \in \intvect{1}{\N}$, $\epart{0:m + 1|m + 1}{i}[\ell+1] \gets (\tilde{\xi}_{0:m}^{i, 1}[\ell+1], \epart{m + 1}{i}[\ell+1])$ and  $\bpart{m + 1}{i}[\ell+1] \gets (\epart{0:m + 1|m + 1}{i}[\ell+1], \stat{m + 1}{i}[\ell+1])$, where
\[
\stat{m + 1}{i}[\ell+1] \gets \M^{-1} \sum_{j = 1}^\M \left( \stattd{m}{i, j}[\ell+1] + \afterm{m}(\eparttd{m}{i, j}[\ell+1], \epart{m + 1}{i}[\ell+1]) \right).
\]
\end{itemize}
This conditional \PARIS\  procedure is summarised in \Cref{alg:parisian:Gibbs}.

Once the set of trajectories and associated statistics $\bpartmb{n}[\ell+1]$ is formed by means of $n$ recursive conditional {\PARIS}  updates, an updated path $\chunk{\zeta}{0}{n}[\ell+1]$ is drawn from $\occm(\epartmb{0:n|n}[\ell+1])$. A full sweep of the {\PPG} is summarised in \Cref{alg:parisian:particle:Gibbs}.

The following Markov kernels will play an instrumental role in the following. For a given path $\{ z_m \}_{m \in \nset}$, the conditional \PARIS\  update in \Cref{alg:parisian:Gibbs} defines an inhomogeneous Markov chain on the spaces $\{ (\yspmb{m}, \yfdmb{m}) \}_{m \in \nset}$ with kernels
$$
\yspmb{m} \times \yfdmb{m + 1} \ni (\ymb[m], A) \mapsto \int
\, \mkmb{m}[z_{m + 1}](\xmb[m|m], \rmd \xmb[m + 1]) \, \ck{m}(\ymb[m], \xmb[m + 1], A), \quad m \in \nset,
$$
where
\begin{align}
\lefteqn{ \ck{m} : \yspmb{m} \times \xspmb{m + 1} \times \yfdmb{m + 1} \ni
(\ymb[m], \xmb[m + 1], A)} \label{eq:def:ck} \\
&\mapsto \idotsint \indi{A} \left( \Big \{ \Big( (\tilde{x}_{0:m}^{i,1}, x_{m + 1}^i),
\frac{1}{\M} \sum_{j = 1}^\M \left( \tilde{\statl}_m^{i,j} + \afterm{m}(\tilde{x}_m^{i, j}, x_{m + 1}^i) \right) \Big) \Big\}_{i = 1}^\N \right) \nonumber \\
&\times \prod_{i = 1}^\N \left( \sum_{\ell = 1}^\N \frac{\ud{m}(x_{m|m}^\ell, x_{m + 1}^i)}{\sum_{\ell' = 1}^\N \ud{m}(x_{m|m}^{\ell'}, x_{m + 1}^i)}
\delta_{y_m^\ell}(\rmd(\tilde{x}_{0:m}^{i,1}, \tilde{\statl}_m^{i,1})) \right. \nonumber \\
&\left. \times \left( \sum_{\ell = 1}^\N \frac{\ud{m}(x_{m|m}^\ell, x_{m + 1}^i)}{\sum_{\ell' = 1}^\N \ud{m}(x_{m|m}^{\ell'}, x_{m + 1}^i)} \delta_{(x_{m|m}^\ell, \statl_m^{\ell})} \right)^{\tensprod (\M - 1)} (\rmd (\tilde{x}_m^{i,2}, \tilde{\statl}_m^{i, 2}, \ldots, \tilde{x}_{m}^{i, \M}, \tilde{\statl}_m^{i,\M})) \right) \, \nonumber.
\end{align}
In addition, we introduce the joint law
\begin{multline} \label{eq:def:ckjt}
\ckjt{n} : \xspmb{0:n} \times \yfdmb{n} \ni (\xmb[0:n], A) \\
\mapsto
\idotsint \indi{A}(\ymb[n]) \, \ck{0}(\koskimat{\N} \xmb[0], \xmb[1], \rmd \ymb[1])
\prod_{m = 1}^{n - 1} \ck{m}(\ymb[m], \xmb[m + 1], \rmd \ymb[m + 1]),
\end{multline}
where we have defined $\koskimat{\N} \eqdef  \Id_{\N} \tensprod (0, 1)^\intercal$.

The kernel $\ckjt{n}$ can be viewed as a \emph{superincumbent sampling kernel} describing the distribution of the output $\bpartmb{n}$ generated by a sequence of \PARIS\  iterates when the many-body process $\{ \epartmb{m} \}_{m = 0}^n$ associated with the underlying SMC algorithm is given. This allows us to describe alternatively the \PPG\ as follows: given $\zeta_{0:n}[\ell]$, draw $\epartmb{0:n}[\ell+1] \sim \mbjt{n}(\zeta_{0:n}[\ell], \cdot)$; then, draw $\bpartmb{n}[\ell+1] \sim \ckjt{n}(\epartmb{0:n}[\ell+1], \cdot)$ and pick a trajectory $\zeta_{0:n}[\ell+1]$ from $\epartmb{0:n|n}[\ell+1]$ at random. The following proposition, which will be instrumental in the coming developments, establishes that the conditional distribution of $\zeta_{0:n}[\ell+1]$ given $\epartmb{0:n}[\ell+1]$ coincides, as expected, with the particle-induced backward dynamics $\bdpart{n}$.
\begin{proposition} \label{prop:backward:marginal}
For all $n \in \nsetpos$, $\N \in \nsetpos$, $\xmb[0:n] \in \xspmb{0:n}$, and $h \in \bmf(\xfd{0:n})$,
$$
\int \ckjt{n}(\xmb[0:n], \rmd \ymb[n]) \, \occm(\xmb[0:n|n]) h = \bdpart{n}h (\xmb[0:n]).
$$
\end{proposition}
Finally, we define the Markov kernel induced by the \PPG\ as well as the extended probability distribution targeted by the same. For this purpose, we introduce the extended measurable space $(\esp{n}, \efd{n})$ with
$$
\esp{n} \eqdef \yspmb{n} \times \xsp{0:n}, \quad \efd{n} \eqdef \yfdmb{n} \tensprod \xfd{0:n}.
$$
The \PPG\ described in \Cref{alg:parisian:particle:Gibbs} defines a Markov chain on $(\esp{n},  \efd{n})$ with Markov transition kernel
\begin{multline*}
\parisgibbs{n} :  \esp{n} \times \efd{n} \ni (\ymb[n], z_{0:n}, A) \\
\mapsto \iiint \1_A(\ymbtd[n], \tilde{z}_{0:n}) \, \mbjt{n}(z_{0:n}, \rmd \xmbtd[0:n]) \, \ckjt{n}(\xmbtd[0:n], \rmd \ymbtd[n]) \, \occm(\xmbtd[0:n|n])(\rmd \tilde{z}_{0:n}).
\end{multline*}
Note that the values of $\parisgibbs{n}$ defined above do not depend on $\ymb[n]$, but only on $(z_{0:n}, A)$. For any given initial distribution $\upxi \in \probmeas(\xfd{0:n})$, let $\canlaw{\upxi}$ be the distribution of the canonical Markov chain induced by the kernel $\parisgibbs{n}$ and the initial distribution $\upxi$. In the special case where $\upxi = \delta_{z_{0:n}}$ for some given path $z_{0:n} \in \xsp{0:n}$, we use the short-hand notation $\canlaw{\delta_{z_{0:n}}} = \canlaw{z_{0:n}}$.
In addition, denote by
$$
\gibbs{n}  :  \xsp{0:n} \times \xfd{0:n} \ni (z_{0:n}, A)
\mapsto \iiint \1_A(\tilde{z}_{0:n}) \, \mbjt{n}(z_{0:n}, \rmd \xmbtd[0:n]) \, \ckjt{n}(\xmbtd[0:n], \rmd \ymbtd[n]) \, \occm(\xmbtd[0:n|n])(\rmd \tilde{z}_{0:n})
$$
the path-marginalised version of $\parisgibbs{n}$. By \Cref{prop:backward:marginal} it holds that $\gibbs{n} = \mbjt{n} \bdpart{n}$, which shows that $\gibbs{n}$ coincides with the Markov transition kernel of the backward-sampling-based particle Gibbs sampler discussed in \Cref{sec:dual:and:Gibbs}.

Finally, in order prepare for the statement of our theoretical results on the \PPG\ we need to introduce the following Feynman--Kac path model \emph{with a frozen path}. More precisely, for a given path $z_{0:n} \in \xpsp{0}{n}$, define, for every $m \in \intvect{0}{n - 1}$, the unnormalised kernel
$$
\uk[z_{m + 1}]{m} : \xsp{m} \times \xfd{m + 1} \ni (x_m, A) \mapsto \left( 1 - \frac{1}{\N} \right) \uk{m}(x_m, A) + \frac{1}{\N} \pot{m}(x_m) \, \delta_{z_{m + 1}}(A)
$$
and the initial distribution $\targ[z_0]{0} : \xfd{0} \ni A \mapsto (1 - 1 / \N) \targ{0}(A) + \delta_{z_0}(A) / \N$.
Given these quantities, define, for $m \in \intvect{0}{n}$,
$\untarg[z_{0:m}]{m} \eqdef \targ[z_0]{0} \uk[z_1]{0} \cdots \uk[z_m]{m - 1}$
along with the normalised counterpart
$\targ[z_{0:m}]{m} \eqdef \untarg[z_{0:m}]{m}/\untarg[z_{0:m}]{m} \indi{\xpsp{0}{m}}$.
Finally, we introduce, for $m \in \intvect{0}{n}$, the kernels
$$
\rk[z_{0:m - 1}]{m} : \xsp{m} \times \xfd{0:m - 1} \ni (x_m, A) \mapsto \idotsint \1_A(x_{0:n - 1}) \prod_{m = 0}^{n - 1} \bk{m, \targ[z_{0:m}]{m}}(x_{m + 1}, \rmd x_m),
$$
as well as the path model $\targ[z_{0:m}]{0:m} \eqdef  \rk[z_{0:m - 1}]{m} \tensprod \targ[z_{0:m}]{m}$.


\section{Main results}
\label{sec:main-results}
\subsection{Theoretical results}
\label{sec:thm:theoretical:results}

This section will be devoted to establishing our main result, namely the exponentially contracting bias bound stated in \Cref{thm:bias:bound} below. This result will be proved under the following strong mixing assumptions, which are standard in the literature (see \cite{delmoral:2004,douc:moulines:2008,delmoral:2013,del2016particle}).
\begin{assumption}[strong mixing] \label{assump:strong_mixing}
For every $n \in \nset$ there exist $\potlow{n}$, $\pothigh{n}$, $\mdlow{n}$, and $\mdhigh{n}$ in $\rset_+^\ast$ such that
 \begin{enumerate}[label=(\roman*),nosep]
 \item  $\potlow{n} \leq \pot{n}(x_n) \leq \pothigh{n}$ for every $x_n \in \xsp{n}$,
 \item  $\mdlow{n} \leq \md{n}(x_n, x_{n + 1}) \leq \mdhigh{n}$ for every $(x_n, x_{n + 1}) \in \xsp{n:n+1}$.
 \end{enumerate}
\end{assumption}

Under \Cref{assump:strong_mixing}, define, for every $n \in \nset$,
\begin{equation} \label{eq:def:rho:n}
\rho_n \eqdef \max_{m \in \intvect{0}{n}} \frac{\pothigh{m} \mdhigh{m}}{\potlow{m} \mdlow{m}}
\end{equation}
and, for every $\N \in \nsetpos$ and $n \in \nset$ such that $\N > \N_n \eqdef 1+ 5 \rho^2 n / 2$,
\begin{equation} \label{eq:def:kappa}
\kappa_{N, n} \eqdef 1 - \frac{1 - N_n / N}{1 + 4 n (1+2 \rho_n^2) / N}.
\end{equation}
Note that $\kappa_{N, n} \in (0, 1)$ for all $\N$ and $n$ as above. 
\begin{theorem} \label{thm:bias:bound}
Assume \Cref{assump:strong_mixing}. Then for every $n \in \nset$ there exist $\cstparisbias$, $\cstparismse$, and $\cstpariscov$ in $\rset_+^\ast$ such that for every $\M \in \nsetpos$, $\upxi \in \probmeas(\xpfd{0}{n})$, $\ell \in \nsetpos$, $s \in \nsetpos$, and $\N \in \nsetpos$ such that $\N > \N_n$,
\begin{align}
\label{eq:bias:bound}
\left|\E_\upxi \left[ \occm(\statmb{n}[\ell])(\operatorname{id}) \right] - \targ{0:n} \af{n} \right| &\leq
\cstparisbias \left(\sum_{m=0}^{n-1} \| \afterm{m} \|_\infty \right) \N^{-1} \kappa_{N, n}^\ell, \\
\label{eq:mse:bound}
\E_\upxi \left[ \left( \occm(\statmb{n}[\ell])(\operatorname{id})  - \targ{0:n} \af{n} \right)^2 \right] &\leq
\cstparismse \left(\sum_{m=0}^{n - 1} \| \afterm{m} \|_\infty \right)^2 \N^{-1}, \\
\nonumber
\lefteqn{
\left|\E_\upxi \left[ \left( \occm(\statmb{n}[\ell])(\operatorname{id})  - \targ{0:n} \af{n} \right)  \left( \occm(\statmb{n}[\ell+s])(\operatorname{id})  - \targ{0:n} \af{n} \right) \right]\right|} \hspace{41mm} \\
\label{eq:cov:bound} 
&\leq \cstpariscov  \left(\sum_{m=0}^{n - 1} \| \afterm{m} \|_\infty \right)^2 \N^{-3/2} \kappa_{N, n}^{s}.
\end{align}
\end{theorem}
\noindent
The constants $\cstparisbias$, $\cstparismse$, and $\cstpariscov$ are explicitly given in the proof.
Since the focus of this paper is on the dependence on $\N$ and the index $\ell$, we have made no attempt to optimise the dependence of these constants on $n$ in our proofs; still, we believe that it is possible to prove,
under the stated assumptions, that this dependence is linear.
The proof of the bound in \Cref{thm:bias:bound} is based on four key ingredients. The first is the following unbiasedness property of the {\PARIS} under the many-body Feynman--Kac path model.
\begin{theorem} \label{thm:unbiasedness}
For every $n \in \nset$, $\N \in \nsetpos$, and $\ell \in \nsetpos$,
$$
\E_{\targ{0:n}} \left[ \occm(\statmb{n}[\ell])(\operatorname{id}) \right] = \int \targ{0:n} \mbjt{n} \ckjt{n}(\rmd \statlmb_n) \, \occm(\statlmb[n])(\operatorname{id}) = \int \targmb{0:n}\ckjt{n}(\rmd \statlmb_n) \, \occm(\statlmb[n])(\operatorname{id})  = \targ{0:n} h_n.
$$
\end{theorem}
The proof of \Cref{thm:unbiasedness} is found in \Cref{sec:proof:unbiasedness}. The second is the uniform geometric ergodicity of the particle Gibbs with backward sampling established in \cite{del2018sharp}.
\begin{theorem} \label{thm:geometric:ergodicity:particle:Gibbs}
Assume \Cref{assump:strong_mixing}. Then for every $n \in \nset$, $(\mu, \nu) \in \probmeas(\xfd{0:n})^2$, $\ell \in \nsetpos$, and $\N \in \nsetpos$ such that $N > 1 + 5 \rho_n^2 n / 2$, $\| \mu \gibbs{n}^\ell - \nu \gibbs{n}^\ell \|_{\mathrm{TV}} \leq \mixrate{n}{\N}^\ell$, where $\mixrate{n}{\N}$ is defined in \eqref{eq:def:kappa}.
\end{theorem}

As a third ingredient, we require the following uniform exponential concentration inequality of the conditional {\PARIS} with respect to the frozen-path Feynman--Kac model defined in the previous section.
\begin{theorem} \label{cor:exp:conc:cond:paris}
For every $n \in \nset$ there exist $\cstcondparisc_n > 0$ and $\cstcondparisd_n > 0$ such that for every $\M \in \nsetpos$, $z_{0:n} \in \xpsp{0}{n}$, $\N \in \nsetpos$, and $\varepsilon > 0$,
$$
\int \mbjt{n} \ckjt{n}(z_{0:n}, \rmd \statlmb_n) \indin{\left| \occm( \statlmb_n)(\operatorname{id})  - \targ[z_{0:n}]{0:n}\af{n} \right|\geq \varepsilon } \leq \cstcondparisc_n \exp \left( -   \frac{\cstcondparisd_n \N \varepsilon^2}{2 (\sum_{m = 0}^{n - 1} \| \afterm{m} \|_\infty)^2}  \right).
$$
\end{theorem}
\Cref{cor:exp:conc:cond:paris}, whose proof is found in \Cref{sec:proof:exp:concentration}, implies, in turn, the following conditional variance bound.
\begin{proposition} \label{cor:Lp:cond:paris}
For every $n \in \nset$, $M \in \nsets$, $z_{0:n} \in \xpsp{0}{n}$, and $\N \in \nsetpos$,
$$
\int \mbjt{n} \ckjt{n}(z_{0:n}, \rmd \statlmb_n) \left| \occm( \statlmb_n)(\operatorname{id})  - \targ[z_{0:n}]{0:n}\af{n} \right|^2 \leq \frac{\cstcondparisc_n}{\cstcondparisd_n} \left( \sum_{m = 0}^{n - 1} \| \afterm{m} \|_\infty \right)^2 \N^{-1}.
$$
\end{proposition}
Using \Cref{cor:Lp:cond:paris}, we deduce, in turn, the following bias bound, whose proof is postponed to \Cref{subsec:prop:bias:cond:paris}.
\begin{proposition}
\label{prop:bias:cond:paris}
For every $n \in \nset$ there exists $\cstcondparisbias > 0$  such that for every $\M \in \nsetpos$, $z_{0:n} \in \xpsp{0}{n}$, and $\N \in \nsetpos$,
\begin{align*}
  \left|\int \mbjt{n} \ckjt{n}(z_{0:n}, \rmd \statlmb_n) \, \occm(\statlmb[n])(\operatorname{id}) - \targ[z_{0:n}]{0:n} \af{n} \right| \leq
\cstcondparisbias \N^{-1} \left( \sum_{m = 0}^{n - 1} \| \afterm{m} \|_\infty \right) .
\end{align*}
\end{proposition}
A fourth and last ingredient in the proof of \Cref{thm:bias:bound} is the following bound on the discrepancy between additive expectations under the original and frozen-path Feynman--Kac models. This bound is established using novel results in \cite{gloaguen:lecorff:olsson:2022}. More precisely, since for every $m \in \nset$, $(x, z) \in \xsp{m}^2$, $\N \in \nsetpos$, and $h \in \bmf(\xfd{m + 1})$, using \Cref{assump:strong_mixing},
$$
\left| \uk[z]{m}h(x) - \uk{m}h(x) \right| \leq \frac{1}{\N} \| \pot{m} \|_\infty \|h \|_\infty \leq \frac{1}{\N} \gsupbound{m} \| h \|_\infty,
$$
applying \cite[Theorem~4.3]{gloaguen:lecorff:olsson:2022} yields the following.
\begin{proposition} \label{prop:error:conditional:model}
Assume \Cref{assump:strong_mixing}. Then there exists $\constpdist > 0$ such that for every $n \in \nset$, $\N \in \nset$, and $z_{0:n} \in \xsp{0:n}$,
$$
\left| \targ[z_{0:n}]{0:n} \af{n} - \targ{0:n} \af{n} \right| \leq
\constpdist \N^{-1} \sum_{m = 0}^{n - 1} \| \afterm{m} \|_\infty. 
$$
\end{proposition}
Note that assuming, in addition, that $\sup_{n \in \nset} \| \afterm{n} \|_\infty < \infty$ yields an $\mathcal{O}(n / \N)$ bound in \Cref{prop:error:conditional:model}.

Finally, by combining these ingredients we are now ready to present a proof of \Cref{thm:bias:bound}.

\begin{proof}[Proof of \Cref{thm:bias:bound}]
Write, using the tower property,
$$
\E_{\upxi} \left[ \occm(\statmb{n}\left[\ell\right])(\operatorname{id}) \right] = \E_{\upxi} \left[\E_{\zeta_{0:n}\left[\ell\right]} \left[\occm(\statmb{n}\left[0\right])(\operatorname{id})\right] \right] = \int \upxi \gibbs{n}^{\ell}\mbjt{n} \ckjt{n}(\rmd \statlmb_n) \, \occm(\statlmb[n])(\operatorname{id}).
$$
Thus, by the unbiasedness property in \Cref{thm:unbiasedness},
\begin{align}
\lefteqn{\left|\E_{\upxi} \left[ \occm(\statmb{n}\left[\ell\right])(\operatorname{id}) \right] - \targ{0:n} \af{n} \right|} \hspace{15mm} \nonumber \\
&=  \left| \int \upxi \gibbs{n}^\ell \mbjt{n} \ckjt{n}(\rmd \statlmb_n) \, \occm(\statlmb[n])(\operatorname{id}) - \int \targ{0:n} \mbjt{n} \ckjt{n}(\rmd \statlmb_n) \, \occm(\statlmb[n])(\operatorname{id}) \right| \nonumber \\
&\leq \big\| \upxi \gibbs{n}^\ell - \targ{0:n} \big\|_{\mathrm{TV}}  \operatorname{osc} \left( \int \mbjt{n} \ckjt{n}(\cdot, \rmd \statlmb_n) \, \occm(\statlmb[n])(\operatorname{id}) \right), \nonumber
\end{align}
where, by \Cref{thm:geometric:ergodicity:particle:Gibbs}, $\| \upxi \gibbs{n}^\ell - \targ{0:n}\|_{\mathrm{TV}} \leq \mixrate{n}{\N}^\ell$. Moreover, to derive an upper bound on the oscillation, we consider the decomposition
\begin{multline*} \label{eq:osc:norm:bound}
\operatorname{osc} \left( \int \mbjt{n} \ckjt{n}(\cdot, \rmd \statlmb_n) \, \occm(\statlmb[n])(\operatorname{id}) \right) \\ 
\leq 2 \left( \left\| \int \mbjt{n} \ckjt{n}(\cdot, \rmd \statlmb_n) \, \occm(\statlmb[n])(\operatorname{id}) - \targ[\cdot]{0:n} \af{n} \right\|_\infty  + \left \| \targ[\cdot]{0:n} \af{n}  - \targ{0:n} \af{n} \right \|_\infty \right),
\end{multline*}
where the two terms on the right-hand side can be bounded using \Cref{prop:error:conditional:model} and \Cref{prop:bias:cond:paris}, respectively. This completes the proof of \eqref{eq:bias:bound}. We now consider the proof of \eqref{eq:mse:bound}. Writing 
\begin{equation*}
\E_\upxi \left[ \left( \occm(\statmb{n}[\ell])(\operatorname{id})  - \targ{0:n} \af{n} \right)^2 \right]
= \int \upxi \gibbs{n}^{\ell}(\rmd \chunk{z}{0}{n}) \mbjt{n} \ckjt{n}(z_{0:n}, \rmd \statlmb_n) \left( \occm( \statlmb_n)(\operatorname{id})  -  \targ{0:n} \af{n} \right)^2,
\end{equation*}
we may establish \eqref{eq:mse:bound} using \Cref{cor:Lp:cond:paris} and \Cref{prop:error:conditional:model}. We finally consider \eqref{eq:cov:bound}.
Using the Markov property we obtain 
\begin{multline*}
\E_\upxi \left[ \left( \occm(\statmb{n}[\ell])(\operatorname{id})  - \targ{0:n} \af{n} \right)  \left( \occm(\statmb{n}[\ell+s])(\operatorname{id})  - \targ{0:n} \af{n} \right) \right] 
\\ =
\E_\upxi \left[ \left( \occm(\statmb{n}[\ell])(\operatorname{id})  - \targ{0:n} \af{n} \right) \left(\E_{\chunk{\zeta}{0}{n}[\ell]}[\occm(\statmb{n}[s])(\operatorname{id})] - \targ{0:n} \af{n} \right)\right], 
\end{multline*}
from which \eqref{eq:cov:bound} follows by \eqref{eq:bias:bound} and \eqref{eq:mse:bound}.
\end{proof}

\subsection{The roll-out {\PPG} estimator}
\label{sec:the:roll-out:PPG:estimator}
In the light of the previous results, it is natural to consider an estimator formed by an average across successive conditional {\PPG} estimators $\{\occm(\statmb{n}[\ell])\}_{\ell \in \nset}$. To mitigate the bias, we remove a ``burn-in'' period whose length $\ki_0$ should be chosen proportionally to the mixing time of the particle Gibbs chain $\{\chunk{\zeta}{0}{n}[\ell]\}_{\ell \in \nsets}$. This yields the estimator
\begin{equation}
\label{eq:rolling-estimator}
\rollingestim[\ki_0][\ki][N][\af{n}]= (\ki-\ki_0)^{-1} \sum_{\ell=\ki_0+1}^\ki \occm(\statmb{n}[\ell])(\operatorname{id}). 
\end{equation}
The total number of particles  underlying this estimator is $\totalbudget = (\N - 1) \ki$.  We denote by $\burningloss= (\ki-\ki_0)/\ki$ the ratio of the number of particles used in the estimator to the total number of sampled particles.

Our final main result provides bounds on the bias and the MSE of the estimator \eqref{eq:rolling-estimator}. The proof is postponed to
\Cref{sec:proof:theo:bias-mse-rolling}.
\begin{theorem}
\label{theo:bias-mse-rolling}
Assume \Cref{assump:strong_mixing}. Then for every $n \in \nset$,  $\M \in \nsetpos$, $\upxi \in \probmeas(\xpfd{0}{n})$, $\ell \in \nsetpos$, $s \in \nsetpos$, and $\N \in \nsetpos$ such that $\N > \N_n$, 
\begin{align}
\label{eq:theo:bias-mse-rolling:bias}
&\left| \E_{\upxi}[\rollingestim[{\ki}_0][\ki][N][\af{n}]] - \targ{0:n} \af{n} \right| \leq  \cstparisbias  \left(\sum_{m=0}^{n-1} \| \afterm{m} \|_\infty \right) \frac{\kappa_{N, n}^{\ki_0}}{(\ki-\ki_0) (1 - \kappa_{N,n})\N},\\
\nonumber
&\E_{\upxi}\left[ \left(\rollingestim[{\ki}_0][\ki][N][\af{n}] - \targ{0:n} \af{n} \right)^2 \right] \\
\label{eq:theo:bias-mse-rolling:mse}
& \phantom{\E_{\upxi}[ (\rollingestim[{\ki}_0][\ki][N][\af{n}]  } \leq
  \left(\sum_{m=0}^{n - 1} \| \afterm{m} \|_\infty \right)^2 \{\N (\ki- \ki_0)\}^{-1} (\cstparismse + 2 \cstpariscov  \N^{-1/2} (1-\kappa_{N,n})^{-1}).
\end{align}
\end{theorem} 


\section{Proofs}
\label{sec:main-proofs}
\subsection{Proof of \Cref{prop:backward:decomposition}} \label{sec:proof:backward:decomposition}

Using the identity
$$
\targ{0} \uk{0} \cdots \uk{n - 1} \indi{\xsp{n}} = \prod_{m = 0}^{n - 1} \targ{m} \uk{m} \indi{\xsp{m + 1}}
$$
and the fact that each kernel $\uk{m}$ has a transition density, write, for $h \in \bmf(\xfd{0:n})$,
\begin{align}
\targ{0:n} h &= \idotsint h(x_{0:n}) \, \targ{0}(\rmd x_0) \prod_{m = 0}^{n - 1} \left(\frac{\targ{m}[\ud{m}(\bcdot, x_{m + 1})] \, \refm{m+1}(\rmd x_{m + 1})}{\targ{m} \uk{m} \indi{\xsp{m + 1}}} \right) \left( \frac{\ud{m}(x_m, x_{m + 1})}{\targ{m}[\ud{m}(\bcdot, x_{m + 1})]} \right) \nonumber \\
&= \idotsint h(x_{0:n}) \, \targ{n}(\rmd x_n) \prod_{m = 0}^{n - 1} \frac{\targ{m}(\rmd x_m) \, \ud{m}(x_m, x_{m + 1})}{\targ{m}[\ud{m}(\bcdot, x_{m + 1})]} \\
&= \left( \bk{0,\targ{0}}  \tensprod \cdots \tensprod \bk{n-1,\targ{n - 1}} \tensprod  \targ{n} \right) h, \nonumber
\end{align}
which was to be established.


\subsection{Proof of \Cref{thm:duality}}
\label{sec:proof:thm:duality}

\begin{lemma}
\label{lem:transport:id}
For all $n \in \nset$, $\xmb_n \in \xspmb{n}$, and $h \in \bmf(\xfdmb{n + 1} \varotimes \xfd{n + 1})$,
\begin{multline} \label{eq:transport:id}
\iint h(\xmb[n + 1], z_{n + 1}) \, \ukmb{n}(\xmb[n], \rmd \xmb[n + 1]) \, \occm(\xmb[n + 1])(\rmd z_{n + 1}) \\
= \iint h(\xmb[n + 1], z_{n + 1}) \, \occm(\xmb[n]) \uk{n}(\rmd z_{n + 1}) \, \mkmb{n}[z_{n + 1}](\xmb[n], \rmd \xmb[n + 1]).
\end{multline}
In addition, for all $h \in \bmf(\xfdmb{0} \varotimes \xfd{0})$,
\begin{equation} \label{eq:transport:id:init}
\iint h(\xmb[0], z_0) \, \initmb(\rmd \xmb[0]) \, \occm(\xmb[0])(\rmd z_0) = \iint h(\xmb[0], z_0) \, \initmb[z_0](\rmd \xmb[0]) \, \init(\rmd z_0).
\end{equation}
\end{lemma}

\begin{proof}
Since
$
\occm(\xmb[n]) \, \uk{n}(\rmd z_{n + 1}) = \potmb{n}(\xmb[n]) \, \pd{n}(\occm(\xmb[n]))(\rmd z_{n + 1}),
$
we may rewrite the right-hand side of \eqref{eq:transport:id} according to
\[
\begin{split}
\lefteqn{\iint h(\xmb[n + 1], z_{n + 1}) \, \occm(\xmb[n]) \uk{n}(\rmd z_{n + 1}) \, \mkmb{n}[z_{n + 1}](\xmb[n], \rmd \xmb[n + 1])} \hspace{15mm} \\
&= \potmb{n}(\xmb[n]) \frac{1}{\N} \sum_{i = 0}^{\N - 1} \iint h(\xmb[n + 1], z_{n + 1}) \, \pd{n}(\occm(\xmb[n]))(\rmd z_{n + 1}) \\
&\hspace{15mm} \times \left( \Phi_n(\occm(\xmb_n))^{\tensprod i} \tensprod \delta_{z_{n + 1}} \tensprod \Phi_n(\occm(\xmb_n))^{\tensprod (\N - i - 1)} \right)(\rmd \xmb[n + 1]) \\
&= \potmb{n}(\xmb[n]) \frac{1}{\N} \sum_{i = 1}^\N \idotsint h((x_{n + 1}^{1}, \ldots, x_{n + 1}^{i - 1}, z_{n + 1}, x_{n + 1}^{i + 1}, \ldots, x_{n + 1}^{\N}), z_{n + 1}) \\
&\hspace{15mm} \times \pd{n}(\occm(\xmb[n]))(\rmd z_{n + 1}) \prod_{\ell \neq i} \pd{n}(\occm(\xmb[n]))(\rmd x_{n + 1}^{\ell}) \\
&= \potmb{n}(\xmb_n) \frac{1}{\N} \sum_{i = 1}^\N \int h(\xmb[n + 1], x_{n + 1}^{i}) \, \mkmb{n}(\xmb[n], \rmd \xmb[n + 1]).
\end{split}
\]

On the other hand, note that the left-hand side of \eqref{eq:transport:id} can be expressed as
\begin{multline}
\iint h(\xmb[n + 1], z_{n + 1}) \, \ukmb{n}(\xmb[n], \rmd \xmb[n + 1]) \, \occm(\xmb[n + 1])(\rmd z_{n + 1}) \\
= \potmb{n}(\xmb_n) \frac{1}{\N} \sum_{i = 1}^\N \int h(\xmb[n + 1], x_{n + 1}^{i}) \, \mkmb{n}(\xmb[n], \rmd \xmb[n + 1]),
\end{multline}
which establishes the identity.
The identity \eqref{eq:transport:id:init} is established along similar lines.
\end{proof}

We establish \Cref{thm:duality} by induction; thus, assume that the claim holds true for $n$ and show that for all $h \in \bmf(\xfdmb{0:n + 1} \varotimes \xfd{0:n + 1})$,
\begin{multline} \label{eq:induction:step}
\iint h(\xmb[0:n + 1], z_{0:n + 1}) \, \untargmb{0:n + 1}(\rmd \xmb[0:n + 1]) \, \bdpart{n + 1}(\xmb[0:n + 1], \rmd z_{0:n + 1}) \\
= \iint h(\xmb[0:n + 1], z_{0:n + 1}) \, \untarg{0:n + 1}(\rmd \chunk{z}{0}{n + 1}) \, \mbjt{n + 1}(z_{0:n + 1}, \rmd \xmb[0:n + 1]).
\end{multline}
To prove this, we process, using definition \eqref{eq:def:bdpart}, the left-hand side of \eqref{eq:induction:step} according to
\begin{equation} \label{eq:duality:lhs}
\begin{split}
\lefteqn{\iint h(\xmb[0:n + 1], z_{0:n + 1}) \, \untargmb{0:n + 1}(\rmd \xmb[0:n + 1]) \, \bdpart{n + 1}(\xmb[0:n + 1], \rmd z_{0:n + 1})} \\
&= \iint \untargmb{0:n}(\rmd \xmb[0:n]) \, \bdpart{n}(\xmb[0:n], \rmd {\chunk{z}{0}{n}}) \\
&\hspace{20mm} \times \iint \bar{h}(\xmb[0:n + 1], z_{0:n + 1}) \, \ukmb{n}(\xmb[n], \rmd \xmb[n + 1]) \, \occm(\xmb[n + 1])(\rmd z_{n + 1}),
\end{split}
\end{equation}
where we have defined the function
$$
\bar{h} (\xmb[0:n + 1], z_{0:n + 1}) \eqdef \frac{\ud{n}(z_n, z_{n + 1}) h(\xmb[0:n + 1], z_{0:n + 1})}{\occm(\xmb[n])[\ud{n}(\cdot, z_{n + 1})]}.
$$
Now, applying \Cref{lem:transport:id} to the inner integral and using that
$$
\occm(\xmb[n]) \uk{n}(\rmd z_{n + 1}) = \occm(\xmb[n])[\ud{n}(\cdot, z_{n + 1})] \, \refm{n + 1}(\rmd z_{n + 1})
$$
yields, for every $\xmb[0:n]$ and ${\chunk{z}{0}{n}}$,
\[
\begin{split}
\lefteqn{\iint \bar{h}(\xmb[0:n + 1], z_{0:n + 1}) \, \ukmb{n}(\xmb[n], \rmd \xmb[n + 1]) \, \occm(\xmb[n + 1])(\rmd z_{n + 1})} \\
&= \iint \bar{h}(\xmb[0:n + 1], z_{0:n + 1}) \, \occm(\xmb[n]) \uk{n}(\rmd z_{n + 1}) \, \mkmb{n}[z_{n + 1}](\xmb[n], \rmd \xmb[n + 1]) \\
&= \iint h(\xmb[0:n + 1], z_{0:n + 1}) \, \uk{n}(z_n, \rmd z_{n + 1}) \, \mkmb{n}[z_{n + 1}](\xmb[n], \rmd \xmb[n + 1]).
\end{split}
\]
Inserting the previous identity into \eqref{eq:duality:lhs} and using the induction hypothesis provides
\[
\begin{split}
\lefteqn{\iint h(\xmb[0:n + 1], z_{0:n + 1}) \, \untargmb{0:n + 1}(\rmd \xmb[0:n + 1]) \, \bdpart{n + 1}(\xmb[0:n + 1], \rmd z_{0:n + 1})} \\
&= \iint \untarg{0:n}(\rmd {\chunk{z}{0}{n}}) \, \mbjt{n}({\chunk{z}{0}{n}}, \rmd \xmb[0:n ])   \\
& \hspace{20mm} \times \iint h(\xmb[0:n + 1], z_{0:n + 1}) \, \uk{n}(z_n, \rmd z_{n + 1}) \,  \mkmb{n}[z_{n + 1}](\xmb[n], \rmd \xmb[n + 1])  \\
&= \iint h(\xmb[0:n + 1], z_{0:n + 1}) \, \untarg{0:n + 1}(\rmd z_{0:n + 1}) \, \mbjt{n + 1}(z_{0:n + 1}, \rmd \xmb[0:n + 1]),
\end{split}
\]
which establishes \eqref{eq:induction:step}.

\subsection{Proof of \Cref{thm:unbiasedness}}
\label{sec:proof:unbiasedness}
First, define, for $m \in \nset$,
\begin{equation} \label{eq:def:P}
\pk{m} : \yspmb{m} \times \yfdmb{m + 1} \ni (\ymb[m], A) \mapsto \int
\, \mkmb{m}(\xmb[m|m], \rmd \xmb[m + 1]) \, \ck{m}(\ymb[m], \xmb[m + 1], A).
\end{equation}
For any given initial distribution $\pinit \in \probmeas(\yfdmb{0})$, let $\canlaw[\pkl]{\pinit}$ be the distribution of the canonical Markov chain induced by the Markov kernels $\{ \pk{m} \}_{m \in \nset}$ and the initial distribution $\pinit$. By abuse of notation we write, for $\initmb \in \probmeas(\xfdmb{0})$, $\canlaw[\pkl]{\initmb}$ instead of $\canlaw[\pkl]{\pinit[\initmb]}$, where we have defined the extension $\pinit[\initmb](A) = \int \1_A(\koskimat{\N} \xmb[0]) \, \initmb(\rmd \xmb_0)$, $A \in \yfdmb{0}$. We preface the proof of \Cref{thm:unbiasedness} by some technical lemmas and a proposition.
\begin{lemma} \label{lem:key:identity:untarg:affine}
For all $n \in \nset$ and $(f_{n + 1}, \tilde{f}_{n + 1}) \in \bmf(\xfd{n + 1})^2$,
$$
 \untarg{n + 1}(f_{n + 1} \rk{n + 1} \af{n + 1} + \tilde{f}_{n + 1}) = \untarg{n}\{ \uk{n} f_{n + 1} \rk{n} \af{n} + \uk{n}(\afterm{n} f_{n + 1} + \tilde{f}_{n + 1}) \}.
$$
\end{lemma}
\begin{proof}
Pick arbitrarily $\varphi \in \bmf(\xfd{n:n+1})$ and write, using definition \eqref{eq:bckwd-kernels} and the fact that $\uk{n}$ has a transition density, 
\begin{align}
    \lefteqn{\iint \varphi(x_{n:n+1}) \, \untarg{n}(\rmd x_n) \, \uk{n}(x_n, \rmd x_{n + 1})} \hspace{10mm} \nonumber \\
    &= \iint \varphi(x_{n:n+1}) \untarg{n}[\ud{n}(\cdot, x_{n + 1})] \, \refm{n + 1}(\rmd x_{n + 1}) \, \frac{\untarg{n}(\rmd x_n) \ud{n}(x_n, x_{n + 1})}{\untarg{n}[\ud{n}(\cdot, x_{n + 1})]} \nonumber \\
    &= \iint \varphi(x_{n:n+1}) \, \untarg{n+1}(\rmd x_{n+1}) \, \bk{n, \targ{n}}(x_{n+1}, \rmd x_{n}). \label{eq:key_relation_lemma2} 
\end{align}

    Now, by \eqref{eq:backward:forward:recursion} it holds that 
    \begin{equation*}
        \rk{n+1} \af{n+1}(x_{n+1}) = \int \bk{n, \targ{n}} (x_{n+1}, \rmd x_{n}) \left( \afterm{n}(x_{n:n+1}) + \int \af{n}(x_{0:n}) \, \rk{n}(x_{n}, \rmd x_{0:n-1}) \right);
    \end{equation*}
    therefore, by applying \eqref{eq:key_relation_lemma2} with
    \begin{equation*}
        \varphi(x_{n:n+1}) \eqdef f_{n+1}(x_{n+1}) \left( \afterm{n}(x_{n:n+1}) + \int \af{n}(x_{0:n}) \, \rk{n}(x_{n}, \rmd x_{0:n-1})\right)
    \end{equation*}
    we obtain that
    \begin{align*}
    \untarg{n + 1}(f_{n + 1} \rk{n + 1} \af{n + 1}) &= \iint \varphi(x_{n:n+1}) \, \untarg{n + 1}(\rmd x_{n + 1}) \, \bk{n, \targ{n}} (x_{n+1}, \rmd x_{n}) \\
    &= \iint \varphi(x_{n:n+1}) \, \untarg{n}(\rmd x_n) \, \uk{n}(x_n, \rmd x_{n + 1}) \\
    &= \untarg{n}(\uk{n} f_{n + 1} \rk{n} \af{n} + \uk{n}\afterm{n} f_{n + 1}).
    \end{align*}
    Now the proof is concluded by noting that since $\untarg{n+1} = \untarg{n} \uk{n}$, $\untarg{n + 1}\tilde{f}_{n + 1} = \untarg{n}\uk{n}\tilde{f}_{n + 1}$.
\end{proof}

\begin{lemma}
\label{lem:unbiasedness-1}
For every $n \in \nsetpos$, $h_n \in \bmf(\yfd{n})$, and $\initmb \in \probmeas(\xfdmb{0})$ it holds that
\[
\cPE[\initmb][\pkl]{h_n(\bpartmb{n})}{\epartmb{0|0},\dots,\epartmb{n|n}}=
\ckjt{n} h_n(\epartmb{0|0},\dots,\epartmb{n|n}), \quad \canlaw[\pkl]{\initmb}\mbox{-a.s.}
\]
\end{lemma}

\begin{proof}
Pick arbitrarily $v_n \in \bmf(\xfd{0:n})$. We show that
\begin{equation}
\label{eq:proof-objective}
\E_{\initmb}^{\pkl}[ v_n(\epartmb{0|0},\dots,\epartmb{n|n}) h_n(\bpartmb{n}) ]=
\E_{\initmb}^{\pkl}[ v_n(\epartmb{0|0},\dots,\epartmb{n|n}) \ckjt{n} h_n(\epartmb{0|0},\dots,\epartmb{n|n})], 
\end{equation}
from which the claim follows. Using the definition \eqref{eq:def:P}, the left-hand side
of the previous identity may be rewritten as
\begin{align*}
\lefteqn{\idotsint \pinit[\initmb](\rmd \ymb[0]) \prod_{m=0}^{n-1} \pk{m}(\ymb[m], \rmd \ymb[m+1]) \, h_n(\ymb[n])
v_n(\xmb[0|0],\dots,\xmb[n|n])} \hspace{10mm} \\
&=\idotsint \initmb( \rmd \xmb[0|0]) \prod_{m=0}^{n-1} \mkmb{m}(\xmb[m|m], \rmd \xmb[m+1]) \, \ck{0}(\koskimat{\N} \xmb[0|0], \xmb[1], \rmd \ymb[1])  \\
&\hspace{10mm} \times \prod_{m=0}^{n-1}  \ck{m}(\ymb[m], \xmb[m+1], \rmd \ymb[m+1]) \, h_n(\ymb[n]) v_n(\xmb[0|0],\dots,\xmb[n|n]) \\
&=\idotsint \initmb( \rmd \xmb[0]) \prod_{m=0}^{n-1} \mkmb{m}(\xmb[m], \rmd \xmb[m+1]) \, \ck{0}(\koskimat{\N} \xmb[0], \xmb[1], \rmd \ymb[1])  \\
&\hspace{10mm} \times \prod_{m=0}^{n-1}  \ck{m}(\ymb[m], \xmb[m+1], \rmd \ymb[m+1]) \, h_n(\ymb[n]) v_n(\xmb[0],\dots,\xmb[n]).
\end{align*}
Thus, we may conclude the proof by using the definition \eqref{eq:def:ckjt} of $\ckjt{n}$  together with Fubini's theorem.
\end{proof}

\begin{lemma}
\label{lem:unbiasedness-2}
For every $n \in \nsetpos$ and $h_n \in \bmf(\yfd{n})$ it holds that
\[
\E_{\initmb} \left[\left(\prod_{m=0}^{n-1} \potmb{m}(\epartmb{m|m})\right) h_n(\bpartmb{n})\right]
= \int  \untargmb{0:n} \ckjt{n}(\rmd  \ymb[n]) \, h_n(\ymb[n]).
\]
\end{lemma}
\begin{proof}
The claim of the lemma is a direct implication of \Cref{lem:unbiasedness-1}; indeed, by applying the tower property and the latter we obtain
\begin{align*}
\lefteqn{\E_{\initmb}^{\pkl} \left[\left(\prod_{m=0}^{n-1} \potmb{m}(\epartmb{m|m})\right) h_n(\bpartmb{n})\right]} \hspace{15mm} \\
&= \E_{\initmb}^{\pkl} \left[\left(\prod_{m=0}^{n-1} \potmb{m}(\epartmb{m|m})\right) \ckjt{n} h_n(\epartmb{0|0},\dots,\epartmb{n|n}) \right]\\
&= \idotsint \initmb(\rmd \xmb[0])  \prod_{m=0}^{n-1} \potmb{m}(\xmb[m]) \,  \mkmb{m}(\xmb[m],\rmd \xmb[m+1]) \, \ckjt{n} h_n(\xmb[0:n]) \\
&= \int  \untargmb{0:n} \ckjt{n}(\rmd  \ymb[n]) \, h_n(\ymb[n]).
\end{align*}
\end{proof}

\begin{proposition} \label{prop:unbiasedness:affine}
For all $n \in \nsetpos$, $(\N, \M) \in (\nsetpos)^2$, and $(f_n, \tilde{f}_n) \in \bmf(\xfd{n})^2$,
$$
\int \untargmb{0:n} \ckjt{n}(\rmd \ymb[n]) \,
\left( \frac{1}{\N} \sum_{i = 1}^\N \{\statl_n^i f_n(x_{n|n}^i) + \tilde{f}_n(x_{n|n}^i)\} \right)
= \untarg{n}(f_n \rk{n} \af{n} + \tilde{f}_n).
$$
\end{proposition}

\begin{proof}
Applying \Cref{lem:unbiasedness-2} yields
\begin{multline}
\int \untargmb{0:n} \ckjt{n}(\rmd \ymb[n]) \,
\left( \frac{1}{\N} \sum_{i = 1}^\N \{\statl_n^i f_n(x_{n|n}^i) + \tilde{f}_n(x_{n|n}^i)\} \right) \\
= \E_{\initmb}^{\pkl}\left[ \left( \prod_{m = 0}^{n - 1} \potmb{m}(\epartmb{m|m}) \right) \frac{1}{\N} \sum_{i = 1}^\N \{\stat{n}{i} f_n(\epart{n|n}{i}) + \tilde{f}_n(\epart{n|n}{i})\} \right] \, .
\end{multline}
In the following we will use repeatedly the following filtrations. Let
$\partfiltbar{n} \eqdef \sigma( \{ \bpartmb{m} \}_{m = 0}^n)$
be the $\sigma$-field generated by the output of the {\PARIS} (\Cref{alg:paris}) during the first $n$ iterations. In addition, let $\partfilt{n} \eqdef \partfiltbar{n - 1} \vee \sigma(\epartmb{n|n})$.

We proceed by induction. Thus, assume that the statement of the proposition holds true for a given $n \in \nsetpos$ and consider, for arbitrarily chosen $(f_{n + 1}, \tilde{f}_{n + 1}) \in \bmf(\xfd{n + 1})^2$,
\begin{multline*}
\CPE[\initmb][\pkl]{\left( \prod_{m = 0}^n \potmb{m}(\epartmb{m|m}) \right) \frac{1}{\N}
\sum_{i = 1}^\N \{\stat{n + 1}{i} f_{n + 1}(\epart{n + 1|n + 1}{i}) + \tilde{f}_{n + 1}(\epart{n + 1|n + 1}{i})\}}{\partfiltbar{n}} \\
= \left( \prod_{m = 0}^n \potmb{m}(\epartmb{m|m}) \right) \cPE[\initmb][\pkl]{\stat{n + 1}{1} f_{n + 1}(\epart{n + 1|n + 1}{1}) + \tilde{f}_{n + 1}(\epart{n + 1|n + 1}{1})}{\partfiltbar{n}} \,,
\end{multline*}
where we used that the variables $\{\stat{n + 1}{i} f_{n + 1}(\epart{n + 1|n + 1}{i}) + \tilde{f}_{n + 1}(\epart{n + 1|n + 1}{i})\}_{i = 1}^\N$ are conditionally i.i.d. given $\partfiltbar{n}$. Note that, by symmetry,
\begin{align}
\CPE[\initmb][\pkl]{\stat{n + 1}{1}}{\partfilt{n + 1}} &= \int \ck{n}(\bpartmb{n}, \epartmb{n + 1|n + 1}, \rmd \ymb[n + 1]) \, \statl_{n + 1}^1 \nonumber \\
&= \idotsint \left( \prod_{j = 1}^\M \sum_{\ell = 1}^\N \frac{\ud{n}(\epart{n|n}{\ell}, \epart{n + 1|n + 1}{1})}{\sum_{\ell' = 1}^\N \ud{n}(\epart{n|n}{\ell'}, \epart{n + 1|n + 1}{1})} \delta_{(\epart{n|n}{\ell}, \stat{n}{\ell})}  (\rmd \tilde{x}_n^{1,j}, \rmd \tilde{\statl}_n^{1,j}) \right) \nonumber \\
& \hspace{50mm} \times \frac{1}{\M} \sum_{j = 1}^\M \left( \tilde{\statl}_n^{1,j} + \afterm{n}(\tilde{x}_n^{1,j}, \epart{n + 1|n + 1}{1}) \right) \nonumber \\
&= \sum_{\ell = 1}^\N \frac{\ud{n}(\epart{n|n}{\ell}, \epart{n + 1|n + 1}{1})}{\sum_{\ell' = 1}^\N \ud{n}(\epart{n|n}{\ell'}, \epart{n + 1|n + 1}{1})} \left( \stat{n}{\ell} + \afterm{n}(\epart{n|n}{\ell}, \epart{n + 1|n + 1}{1}) \right). \label{eq:cond:exp:beta}
\end{align}
Thus, using the tower property,
\begin{multline*}
\CPE[\initmb][\pkl]{\stat{n + 1}{1} f_{n + 1}(\epart{n + 1|n + 1}{1})}{\partfiltbar{n}} \\
= \int \pd{n}(\occm(\epartmb{n|n}))(\rmd x_{n + 1}) \, f_{n + 1}(x_{n + 1}) \sum_{\ell = 1}^\N \frac{\ud{n}(\epart{n|n}{\ell}, x_{n + 1})}{\sum_{\ell' = 1}^\N \ud{n}(\epart{n|n}{\ell'}, x_{n + 1})} \left( \stat{n}{\ell} + \afterm{n}(\epart{n|n}{\ell}, x_{n + 1}) \right),
\end{multline*}
and consequently, using definition \eqref{eq:def:pd:mapping},
\begin{align}
&\left( \prod_{m = 0}^n \potmb{m}(\epartmb{m|m}) \right) \CPE[\initmb][\pkl]{\stat{n + 1}{1} f_{n + 1}(\epart{n + 1|n + 1}{1})}{\partfiltbar{n}} \nonumber \\
&=\left( \prod_{m = 0}^{n - 1} \potmb{m}(\epartmb{m|m}) \right) \int \frac{1}{\N} \sum_{i = 1}^\N \ud{n}(\epart{n|n}{i}, x_{n + 1}) \nonumber \\
&\hspace{6mm} \times f_{n + 1}(x_{n + 1}) \sum_{\ell = 1}^\N \frac{\ud{n}(\epart{n|n}{\ell}, x_{n + 1})}{\sum_{\ell' = 1}^\N \ud{n}(\epart{n|n}{\ell'}, x_{n + 1})} \left( \stat{n}{\ell} + \afterm{n}(\epart{n|n}{\ell}, x_{n + 1}) \right) \, \refm{n+1}(\rmd x_{n + 1}) \nonumber \\
&= \left( \prod_{m = 0}^{n - 1} \potmb{m}(\epartmb{m|m}) \right) \frac{1}{\N} \sum_{\ell = 1}^\N \left( \stat{n}{\ell} \uk{n} f_{n + 1} (\epart{n|n}{\ell}) + \uk{n}(\afterm{n} f_{n + 1})(\epart{n|n}{\ell}) \right). \nonumber
\end{align}
Thus, applying the induction hypothesis,
\begin{align}
\lefteqn{\E_{\initmb}^{\pkl} \left[ \left( \prod_{m = 0}^n \potmb{m}(\epartmb{m|m}) \right) \frac{1}{\N} \sum_{i = 1}^\N \stat{n + 1}{i} f_{n + 1}(\epart{n + 1|n + 1}{i}) \right]} \nonumber \\
&= \E_{\initmb}^{\pkl} \left[ \left( \prod_{m = 0}^{n - 1} \potmb{m}(\epartmb{m|m}) \right) \frac{1}{\N} \sum_{\ell = 1}^\N \left( \stat{n}{\ell} \uk{n} f_{n + 1} (\epart{n|n}{\ell}) +\uk{n}(\afterm{n} f_{n + 1})(\epart{n|n}{\ell}) \right) \right] \nonumber \\
&= \untarg{n} \left( \uk{n} f_{n + 1} \rk{n} \af{n} + \uk{n}(\afterm{n} f_{n + 1}) \right). \label{eq:unbiasedness:first:term}
\end{align}
In the same manner, it can be shown that
\begin{equation} \label{eq:unbiasedness:second:term}
\E_{\initmb}^{\pkl} \left[ \left( \prod_{m = 0}^n \potmb{m}(\epartmb{m|m}) \right) \frac{1}{\N} \sum_{i = 1}^\N  \tilde{f}_{n + 1}(\epart{n + 1|n + 1}{i}) \right] = \untarg{n} \uk{n} \tilde{f}_{n + 1}.
\end{equation}
Now, by (\ref{eq:unbiasedness:first:term}--\ref{eq:unbiasedness:second:term}) and \Cref{lem:key:identity:untarg:affine},
\begin{align}
\lefteqn{\hspace{-30mm}\E_{\initmb}^{\pkl} \left[ \left( \prod_{m = 0}^n \potmb{m}(\epartmb{m|m}) \right) \frac{1}{\N} \sum_{i = 1}^\N \{\stat{n + 1}{i} f_{n + 1}(\epart{n + 1|n + 1}{i}) + \tilde{f}_{n + 1}(\epart{n + 1|n + 1}{i})\} \right]} \nonumber \\
&\hspace{-30mm} = \untarg{n} \left( \uk{n} f_{n + 1} \rk{n} \af{n} + \uk{n}(\afterm{n} f_{n + 1} + \uk{n} \tilde{f}_{n + 1}) \right) \nonumber \\
&\hspace{-30mm} = \untarg{n + 1}(f_{n + 1} \rk{n + 1} \af{n + 1} + \tilde{f}_{n + 1}), \nonumber
\end{align}
which shows that the claim of the proposition holds at time $n + 1$.

It remains to check the base case $n = 0$, which holds trivially true as $\statmb{0} = \boldsymbol{0}$, $\rk{0} \af{0} = 0$ by convention, and the initial particles $\epartmb{0|0}$ are drawn from $\initmb$. This completes the proof.
\end{proof}

\begin{proof}[Proof of  \Cref{thm:unbiasedness}]
The identity $\int \targmb{0:n}(\rmd \xmb[0:n]) \, \ckjt{n}(\xmb[0:n], \rmd \statlmb_n) \, \occm(\statlmb[n])(\operatorname{id}) =  \targ{0:n} h_n$ follows immediately by letting $f_n \equiv 1$ and $\tilde{f}_n \equiv 0$ in \Cref{prop:unbiasedness:affine} and using that $\untargmb{0:n}(\xspmb{0:n}) = \untarg{0:n}(\xsp{0:n})$. Moreover, applying  \Cref{thm:duality} yields
\begin{align}
\int \targ{0:n} \mbjt{n} \ckjt{n}(\rmd \statlmb_n) \, \occm(\statlmb[n])(\operatorname{id}) &= \iint \targ{0:n}(\rmd {\chunk{z}{0}{n}}) \, \mbjt{n}({\chunk{z}{0}{n}}, \rmd \xmb[0:n]) \int \ckjt{n}(\xmb[0:n], \rmd \statlmb_n) \, \occm(\statlmb[n])(\operatorname{id}) \nonumber \\
&= \iint \targmb{0:n}(\rmd \xmb[0:n]) \, \bdpart{n}(\xmb[0:n], \rmd {\chunk{z}{0}{n}}) \int \ckjt{n}(\xmb[0:n], \rmd \statlmb_n) \, \occm(\statlmb[n])(\operatorname{id}) \nonumber \\
&= \int \targmb{0:n} \ckjt{n}(\rmd \statlmb_n) \, \occm(\statlmb[n])(\operatorname{id}). \nonumber
\end{align}
Finally, the first identity holds true since $\gibbs{n}$ leaves $\targ{0:n}$ invariant.
\end{proof} 

\begin{funding}
    The work of J.~Olsson is supported by the Swedish Research Council, Grant 2018-05230.
\end{funding}

\bibliographystyle{plain}
\bibliography{cmo2022}


\appendix

\section{Numerical results}
\label{sec:numerics}
Given measurable spaces $(\xsp{},\xfd{})$ and $(\zsp{}, \zfd{})$, a \emph{hidden Markov model} (HMM) is a bivariate stochastic process $\{(X_m, Z_m)\}_{m \in \nset}$ where $(X_m, Z_m)$ takes on values in $(\xsp{} \times \zsp{}, \xfd{} \tensprod \zfd{})$ and where $\{X_n\}_{n \in \nset}$ (the \emph{state sequence}) is a (possibly unnormalised) Markov chain which is only partially observed through the observation process $\{Z_m\}_{m \in \nset}$. We denote by $\{\mk{n}\}_{n \in \nset}$ and $\chi$ the Markov transition kernels and initial distribution of $\{X_n\}_{n \in \nset}$, respectively. Conditionally on the unobserved state sequence, the observations are assumed to be independent and such that the conditional distribution of each $Z_m$ only depends on the corresponding state $X_m$ and has a density $\pot{m}(X_m, \cdot)$ with respect to some dominating measure. HMMs are used today in a variety of scientific and engineering disciplines; see \cite{andrieu:doucet:2002,cappe:moulines:ryden:2005,chopin2020introduction}. Inference in HMMs typically involves the computation of conditional distributions of unobserved states given observations. Of particular interest are the sequences of \emph{filter distributions}, where the filter at time $m \in \nset$, denoted $\eta_m$, is defined as the conditional distribution of $X_m$ given $\chunk{Z}{0}{m} \eqdef (Z_{0}, \ldots, Z_{m})$, and the \emph{joint-smoothing distributions}, where the joint-smoothing distribution at time $m$, denoted $\eta_{0:m}$, is defined as the joint conditional distribution of the states $\chunk{X}{0}{m} = (X_0, \ldots, X_m)$ given $\chunk{Z}{0}{m}$. Consequently, $\eta_m$ is the marginal of $\eta_{0:m}$ with respect to the last state $X_m$. Given a sequence $\{z_m\}_{m \in \nset}$ of fixed observations, it can be shown (see  \cite[Section~3]{cappe:moulines:ryden:2005} for details)  that $\{\eta_{0:m}\}_{m \in \nset}$ forms a Feynman--Kac model as defined in \Cref{sec:introduction} with Markov kernels $\{\mk{m}\}_{m \in \nset}$ potential functions $\pot{m} \eqdef \pot{}(\cdot, z_m)$, $m \in \nset$, on $\xsp{}$.

We now numerically assess the proposed algorithm on two different models, namely (i) a linear Gaussian state-space model (for which the filter and joint-smoothing distribution flows are available in a closed form) and (ii) a stochastic volatility model proposed in \cite{hull:white:1987}.


\paragraph{Linear Gaussian state-space model (LGSSM).}
We first consider a linear Gaussian HMM
\begin{equation}
X_{m + 1} = A X_m + Q \epsilon_{m + 1}, \quad Z_m = B X_m + R \zeta_m, \quad m \in \nset, 
\end{equation}
where $\{\epsilon_m\}_{m \in \nsetpos}$ and $\{\zeta_m\}_{m \in \nset}$ are sequences of independent standard normally distributed random variables. The coefficients $A$, $Q$, $B$, and $R$ are assumed to be known and equal to $0.97$, $0.60$, $0.54$, and $0.33$, respectively. Using this parameterisation, we generate, by simulation, a record $z_{0:n}$ of observations with $n = 1000$.

In this setting, we aim at computing smoothed expectations of the state one-lag covariance
$\af{n}(\chunk{x}{0}{n}) \eqdef \sum_{m=0}^{n-1} x_m x_{m + 1}$. In the linear Gaussian case, the \emph{disturbance smoother} (see \cite[Algorithm 5.2.15]{cappe:moulines:ryden:2005}) provides the exact values of the smoothed sufficient statistics, which allows us to study the bias of the estimator for a given computational budget $\totalbudget$.

To illustrate the bias bounds provided by \Cref{thm:bias:bound,theo:bias-mse-rolling}, we calculate the bias after different numbers $k$ of iterations on the basis of $1000$ independent replicates of the algorithm for $\N \in \{10, 25, 50, 100\}$ with $\totalbudget = N k = 500$.
\Cref{fig:LGSSM:bias_evolution} displays the evolution of the bias of the roll-out estimator (given by \eqref{eq:rolling-estimator}) in the cases where $\ki_0=\ki-1$ and $\ki_0 = \lfloor \ki / 2 \rfloor$.
Moreover, for comparison the estimated bias of the {\PARIS} algorithm (\Cref{alg:paris}) with $\N = \totalbudget$ particles is shown.
For each $N$, we fit a curve of type $e^{ak+b}$ to the produced bias estimates to illustrate the exponentially (with $k$) decreasing {\PPG} bias bound established by \Cref{thm:bias:bound,theo:bias-mse-rolling}.
We also note that the most economical strategy in terms of smallest possible bias for a given computational cost $kN$ is to use a large $N$ and a small $k$.
Indeed, we find that for $N \in \{50, 100\}$, only two iterations of the {\PPG} are sufficient to be less biased than the reference estimator {\PARIS} with $\N = \totalbudget$ particles.
\begin{figure}[h]
    \centering
    \begin{subfigure}{0.45\textwidth}
        \includegraphics[width=\textwidth]{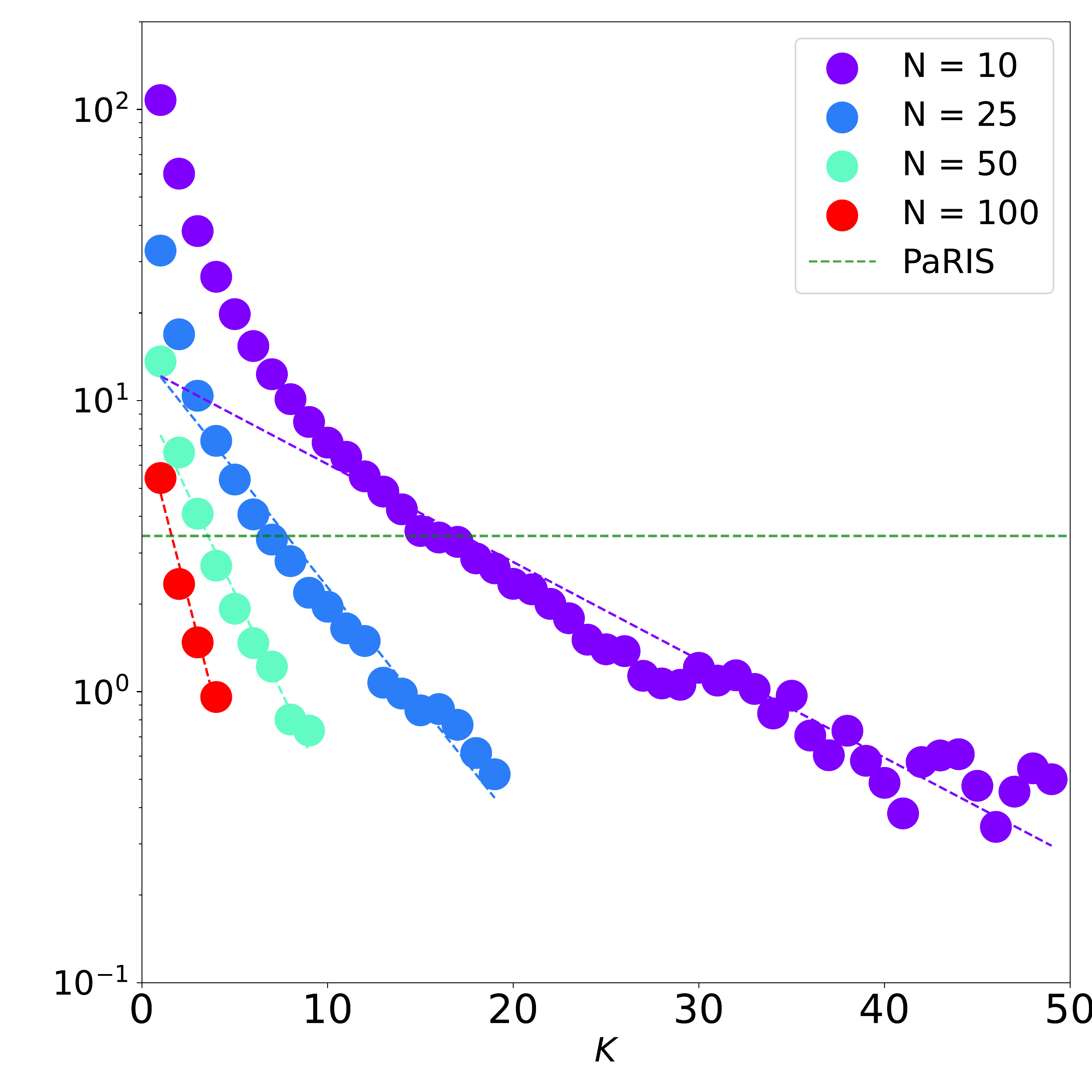}
        \label{fig:LGSSM:dim_1:bias_evolution_no_rollout}
    \end{subfigure}
    \begin{subfigure}{0.45\textwidth}
        \includegraphics[width=\textwidth]{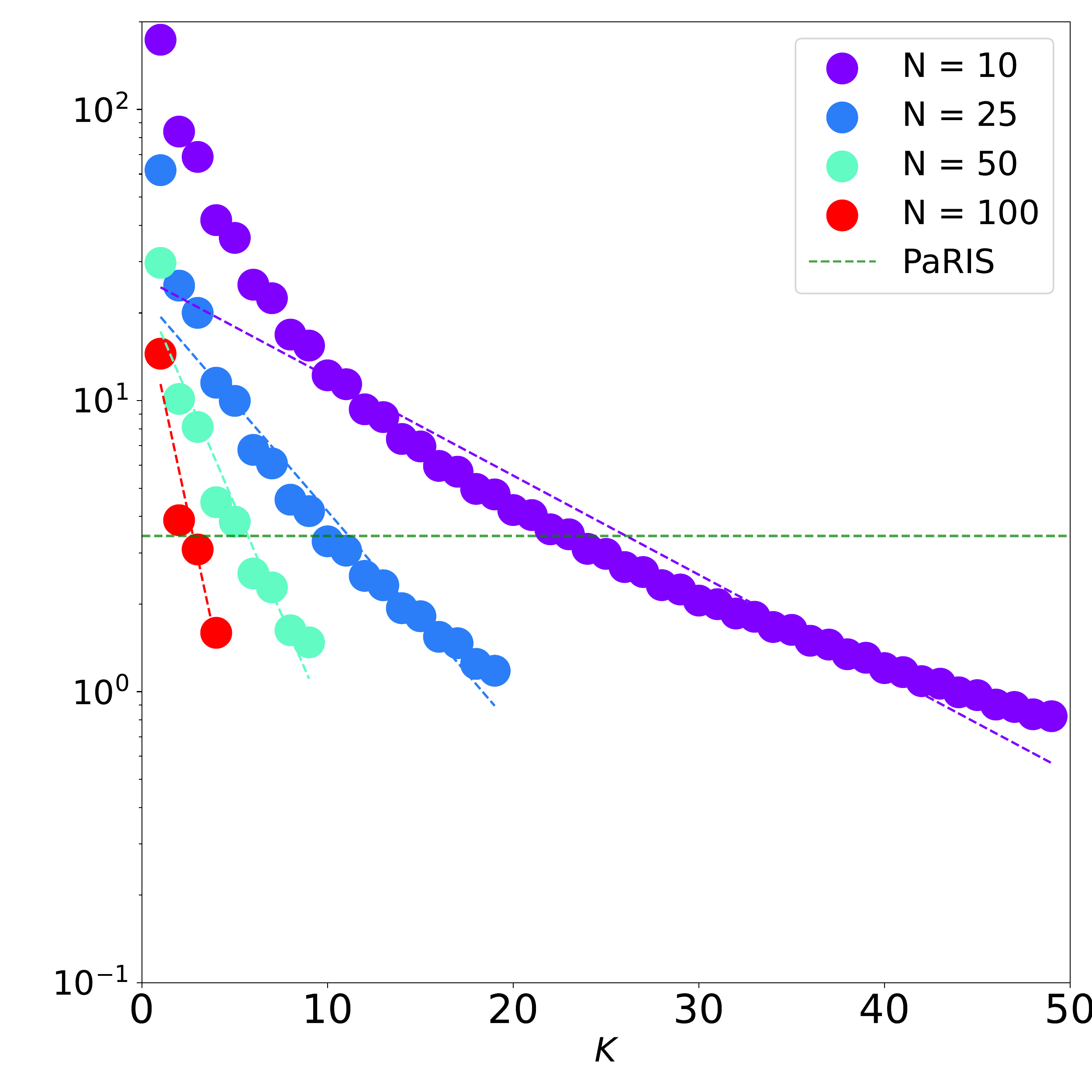}
        \label{fig:LGSSM:dim_1:bias_evolution_rollout_05}
    \end{subfigure}
    \caption{Output of the {\PPG} roll-out estimator for the LGSSM. The curves describe the evolution of the bias with increasing $\ki$ for different particle sample sizes $\N$. The left and right panels correspond to $\ki_0=\ki-1$ and $\ki_0=\lfloor\ki/2\rfloor$, respectively.}
    \label{fig:LGSSM:bias_evolution}
\end{figure}
The preceding remark is also supported by \Cref{fig:LGSSM:comparison_component_2},
which displays, for three different total budgets $\totalbudget$, the distribution of estimates of $\targ{0:n} h_{n}$ using the {\PARIS} as well as three different configurations of the {\PPG} corresponding to $\ki \in \{2, 4, 10\}$ (and $N = \totalbudget / \ki$) with $\ki_0 = \ki  - 1$.
The reference value is shown as a red-dashed line and the mean value of each distribution is shown as a black-dashed line.
Each boxplot is based on $1000$ independent replicates of the corresponding estimator. We observe that in this example, all configurations of the {\PPG} are less biased than the equivalent {\PARIS} estimator.
We also observe that a larger $k$ does not lead to smaller bias in all configurations.
\begin{figure}[h]
    \centering
    \begin{subfigure}{0.32\textwidth}
        \includegraphics[width=\textwidth]{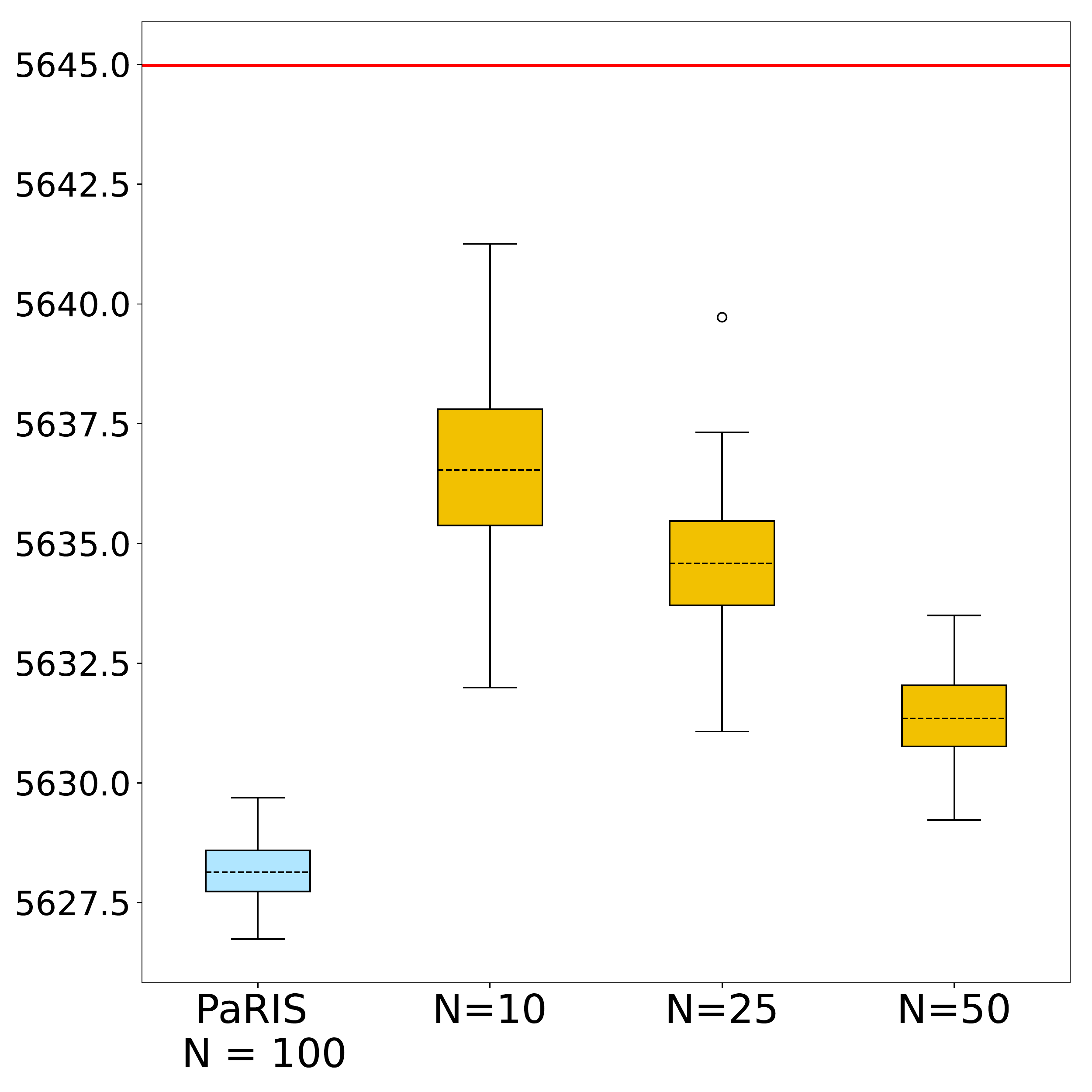}
    \end{subfigure}
    \begin{subfigure}{0.32\textwidth}
        \includegraphics[width=\textwidth]{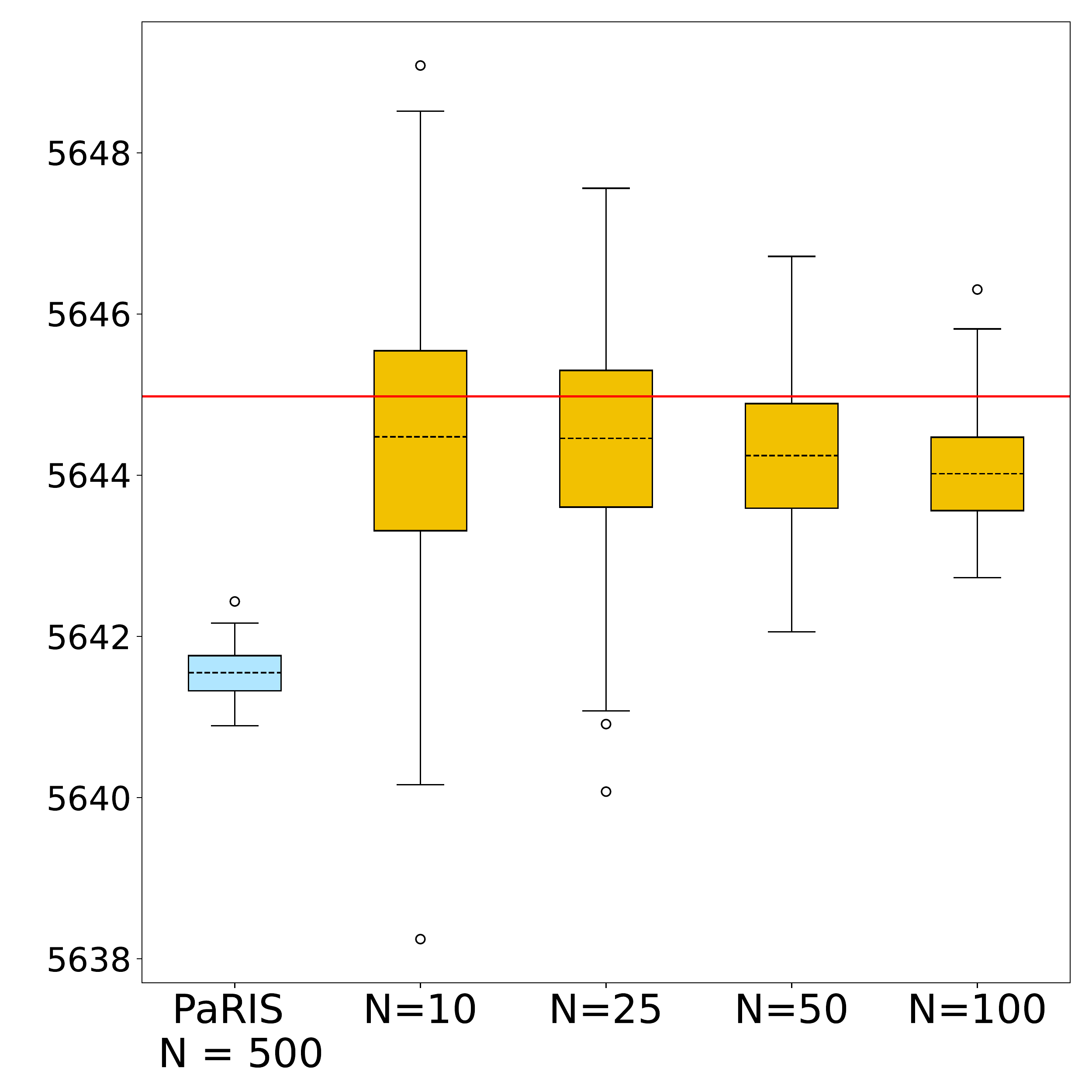}
    \end{subfigure}
    \begin{subfigure}{0.32\textwidth}
        \includegraphics[width=\textwidth]{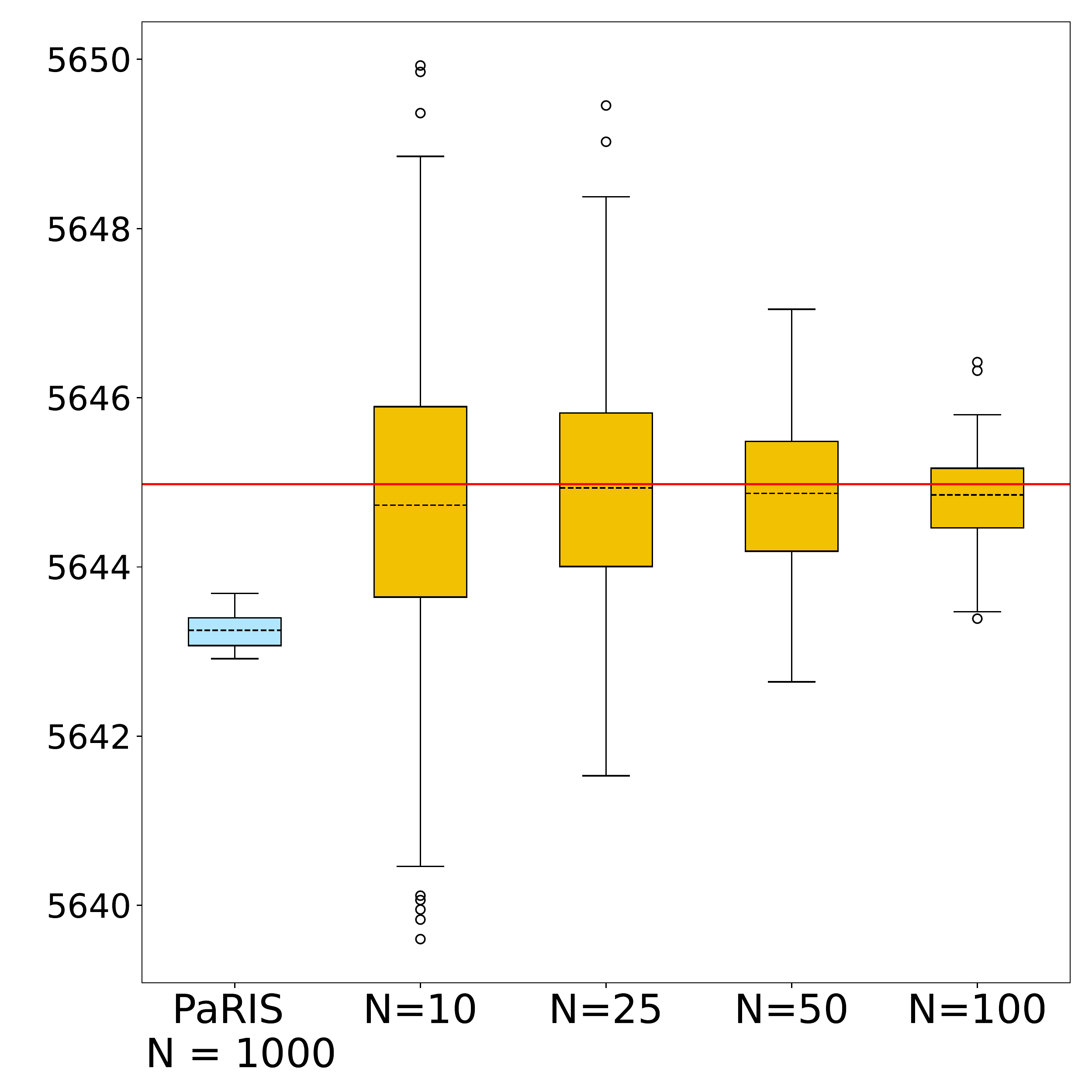}
    \end{subfigure}
    \caption{{\PARIS} and {\PPG} outputs for the LGSSM. The different panels correspond to the different budgets $\totalbudget \in \{100, 500, 1000\}$ and for each panel, yellow boxes correspond to {\PPG} outputs produced using $\ki \in \{10, 4, 2\}$ iterations and $\N \in \{C/10, C/4, C/2\}$ particles.}
    \label{fig:LGSSM:comparison_component_2}
\end{figure}

\paragraph{Stochastic Volatility.}
As a second example, we consider the stochastic volatility model
\begin{equation}
X_{m + 1} = \phi X_m + \sigma_{\epsilon} \epsilon_{m + 1}, \quad Z_m = \beta \exp(X_m /2) \zeta_m, \quad m \in \nset, 
\end{equation}
where $\{ \epsilon_m \}_{m \in \nsetpos}$ and $\{\zeta_m \}_{m \in \nset}$ are as in the previous example and the model parameters $\phi$, $\beta$, and $\sigma_{\epsilon}$, are set to $0.975$, $0.63$, and $0.16$, respectively. Again, an observation record $z_{0:n}$ of length $n = 1000$ is generated through simulation. The reference value is calculated using the {\PARIS} with $\N = 10000$ particles.
Similarly to the LGSSM example, the bias decay with respect to $k$ is shown in \Cref{fig:stovol:bias_evolution} and the comparison with the {\PARIS} algorithm in \Cref{fig:stovol:comparison_component_2}, and the same remarks made for the LGSSM model apply to the stochastic volatility model.
In \Cref{fig:stovol:comparison_component_2}, we observe that for small budgets ($\N=100$), the choice of taking a large number $k$ of steps leads to a higher bias than for the {\PARIS} estimator, indicating that in this case the chain has not yet mixed properly.
\begin{figure}[h]
    \centering
    \begin{subfigure}{0.45\textwidth}
        \includegraphics[width=\textwidth]{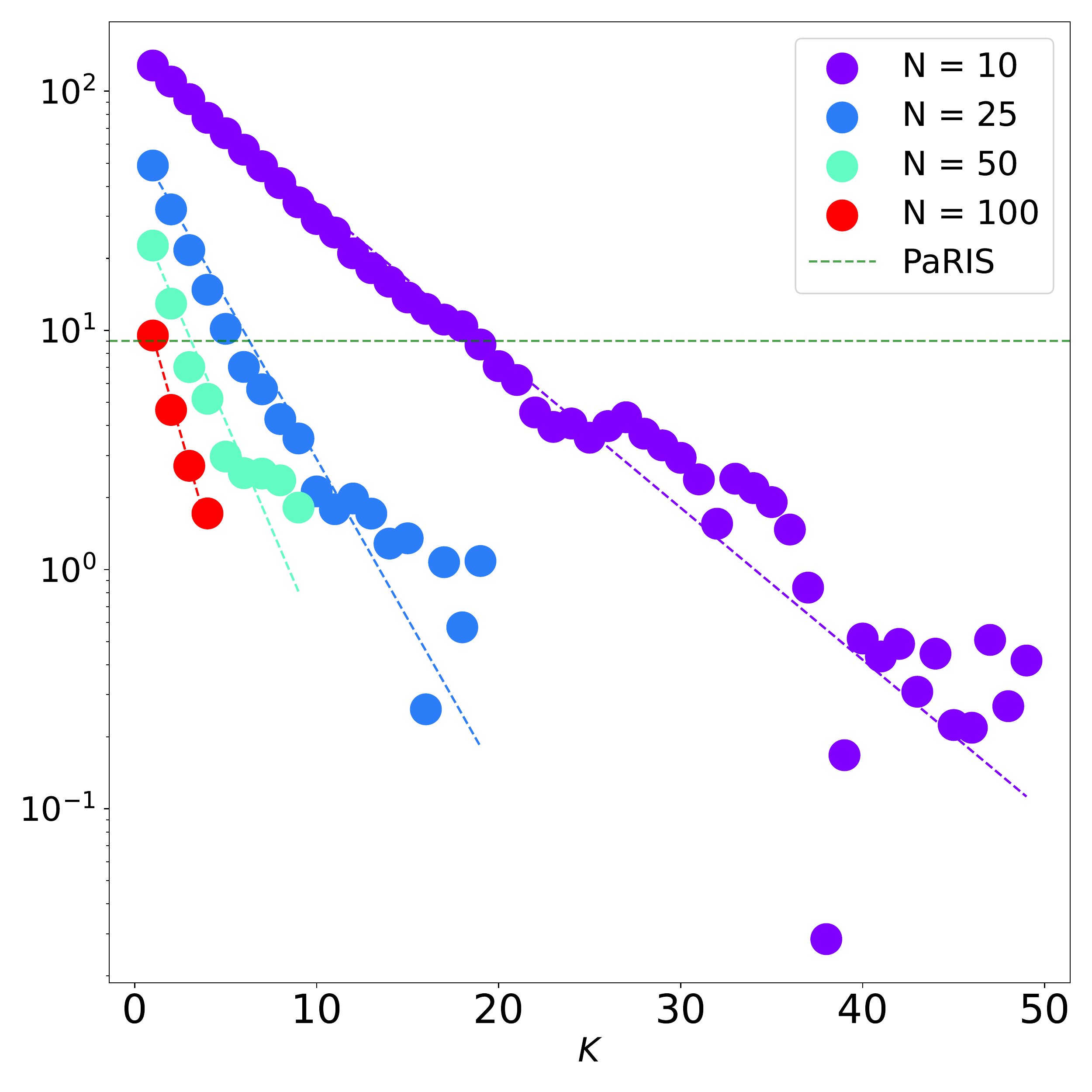}
        \label{fig:stovol:dim_1:bias_evolution_no_rollout}
    \end{subfigure}
    \begin{subfigure}{0.45\textwidth}
        \includegraphics[width=\textwidth]{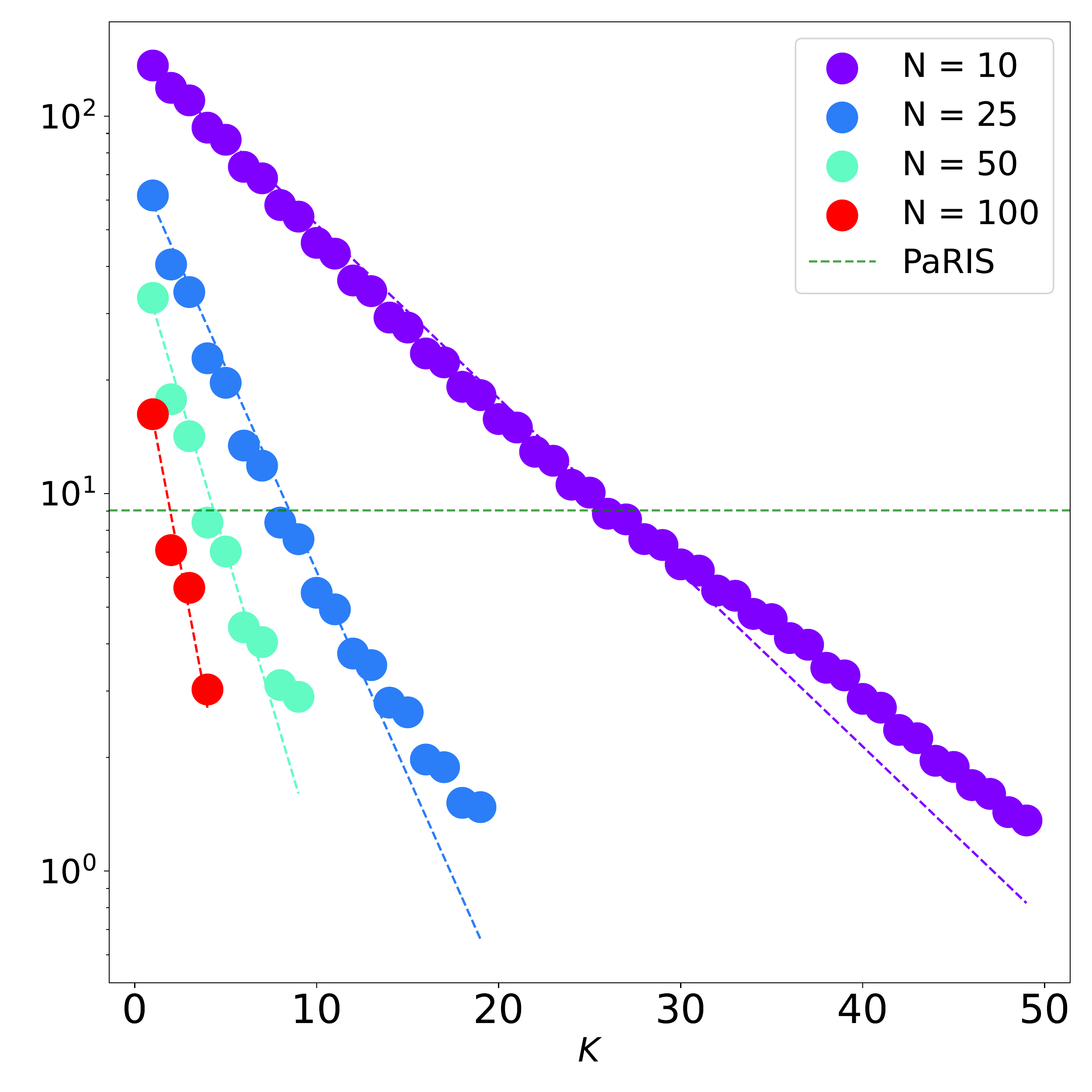}
        \label{fig:stovol:dim_1:bias_evolution_rollout_05}
    \end{subfigure}
    \caption{Output of the {\PPG} roll-out estimator for the stochastic volatility model. The curves describe the evolution of the bias with increasing $\ki$ for different particle sample sizes $\N$. The left and right panels correspond to $\ki_0=\ki-1$ and $\ki_0=\lfloor\ki/2\rfloor$, respectively.}
    \label{fig:stovol:bias_evolution}
\end{figure}
\begin{figure}[h]
    \centering
    \begin{subfigure}{0.32\textwidth}
        \includegraphics[width=\textwidth]{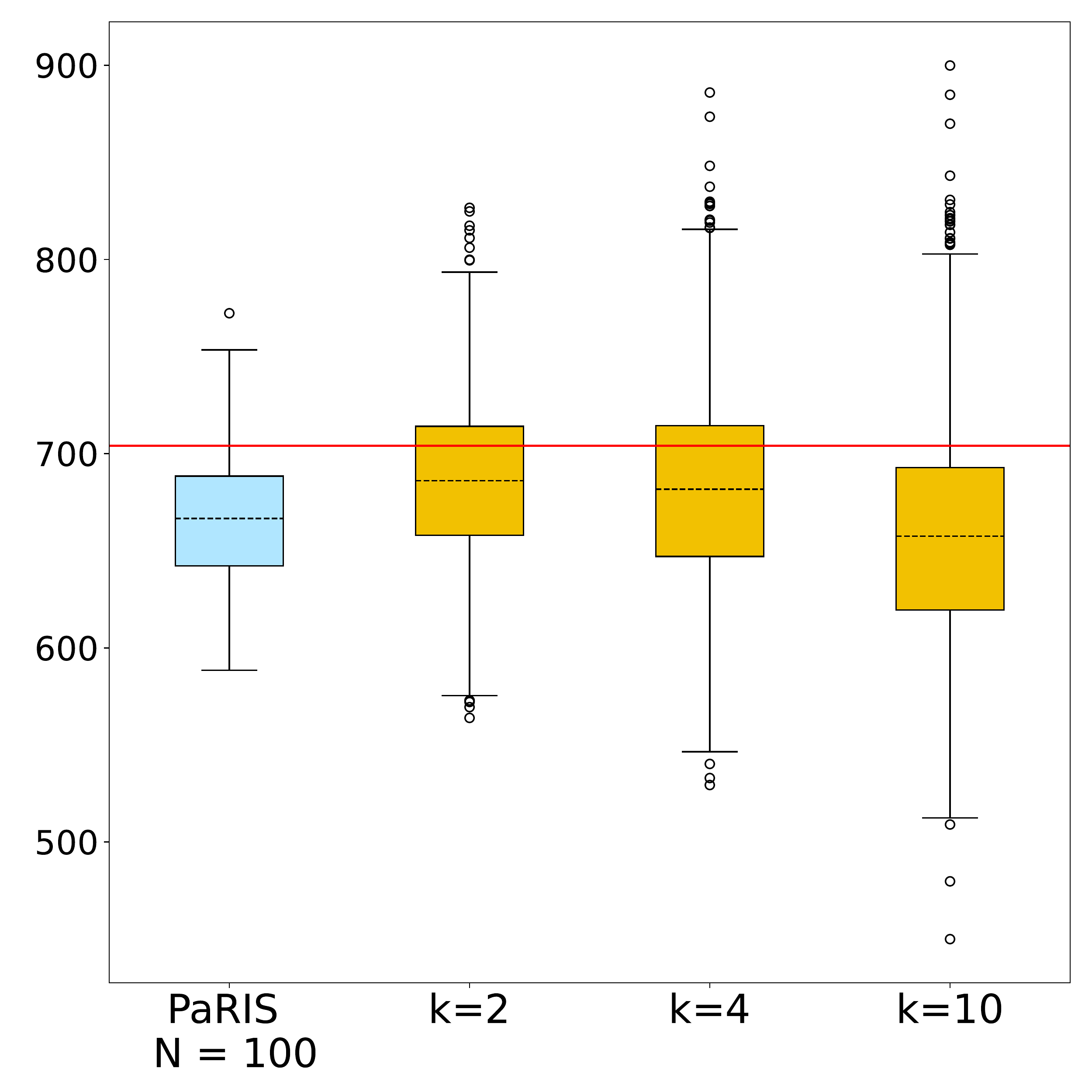}
    \end{subfigure}
    \begin{subfigure}{0.32\textwidth}
        \includegraphics[width=\textwidth]{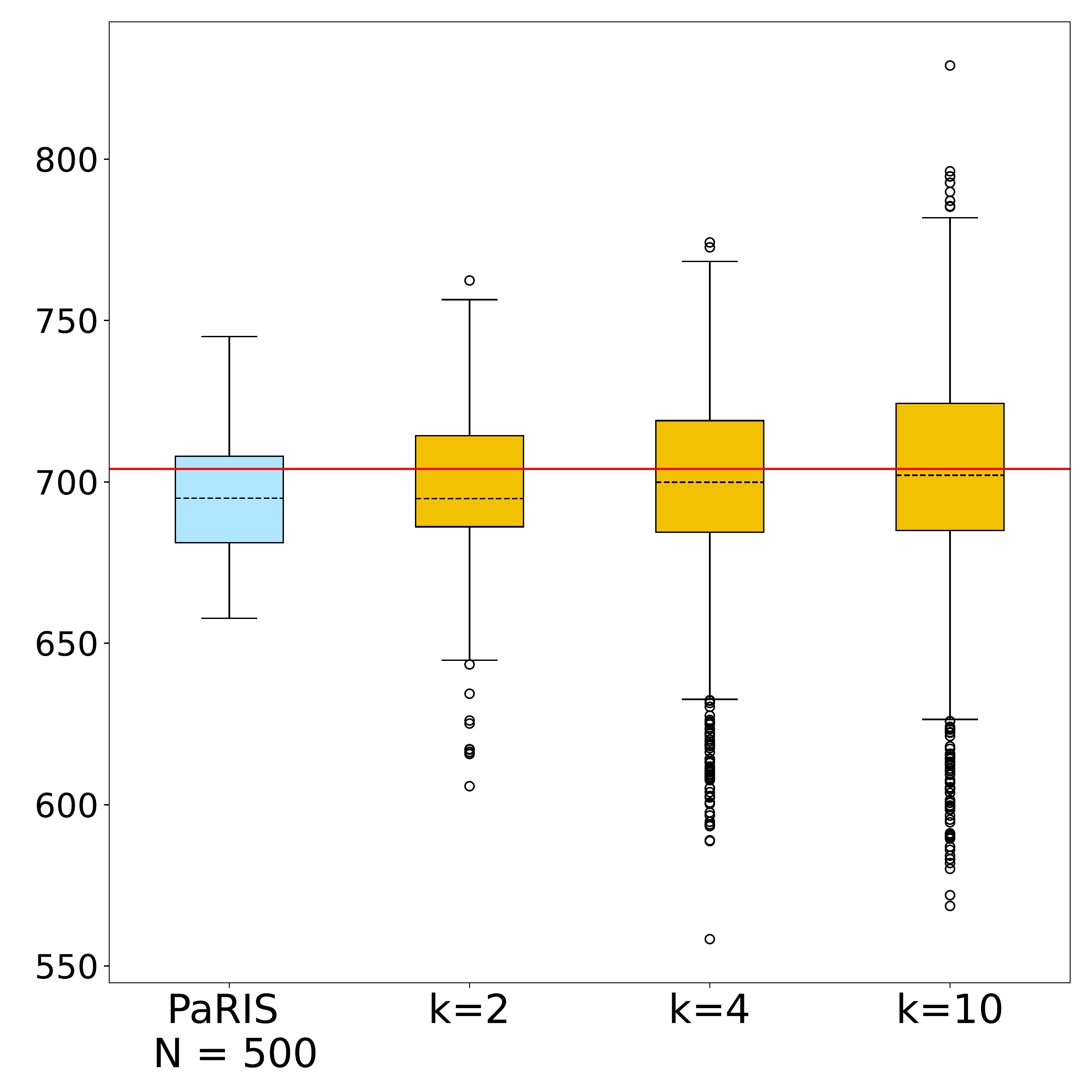}
    \end{subfigure}
    \begin{subfigure}{0.32\textwidth}
        \includegraphics[width=\textwidth]{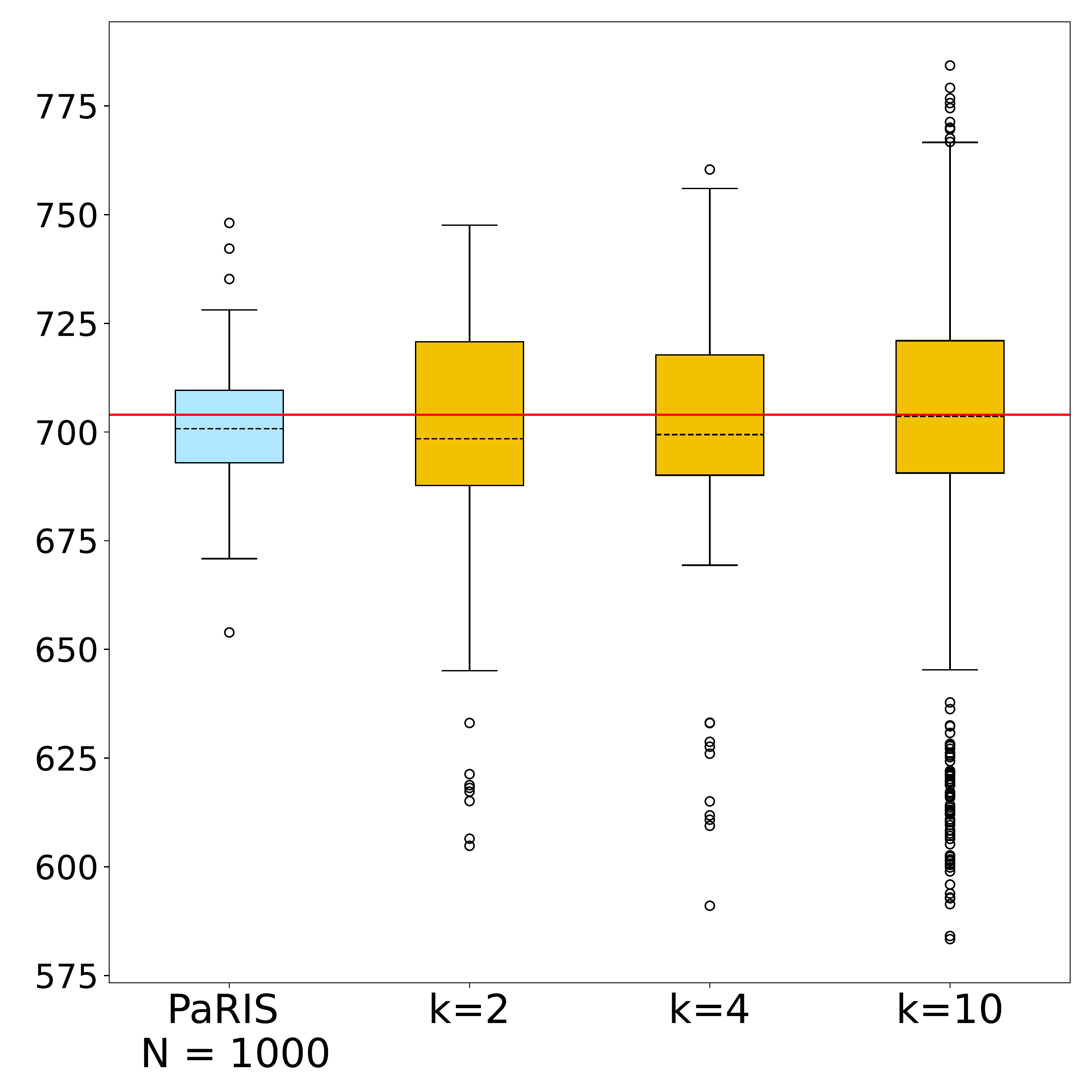}
    \end{subfigure}
    \caption{{\PARIS} and {\PPG} outputs for the stochastic volatility model. The different panels correspond to the different budgets $\totalbudget \in \{100, 500, 1000\}$ and for each panel, yellow boxes correspond to {\PPG} outputs produced using $\ki \in \{2, 4, 10\}$ iterations and $\N \in \{C/2, C/4, C/10\}$ particles.}
    \label{fig:stovol:comparison_component_2}
\end{figure}

\section{Algorithms}
\label{sec:algorithms}
\begin{algorithm}[H]
\SetAlFnt{\small\sf}
\KwData{$\{ (\epart{n}{i}, \stat{n}{i}) \}_{i = 1}^\N$}
 \KwResult{$\{ (\epart{n + 1}{i}, \stat{n + 1}{i}) \}_{i = 1}^\N$}
 draw $\epartmb{n + 1} \sim \mkmb{n}(\epartmb{n}, \cdot)$\;
 \For{$i \gets 1$ \KwTo $\N$}{
draw $ \{ (\eparttd{n}{i, j}, \stattd{n}{i, j}) \}_{j = 1}^\M \sim \left( \sum_{\ell = 1}^\N \frac{\ud{n}(\epart{n}{\ell}, \epart{n + 1}{i})}{\sum_{\ell' = 1}^\N \ud{n}(\epart{n}{\ell'}, \epart{n + 1}{i})} \delta_{(\epart{n}{\ell}, \stat{n}{\ell})} \right)^{\tensprod \M}$\;
set $ \stat{n + 1}{i} \gets \frac{1}{\M} \sum_{j = 1}^\M \left( \stattd{n}{i, j} + \afterm{n}(\eparttd{n}{i, j}, \epart{n + 1}{i}) \right)$\;
}
\caption{\scriptsize{One update of the \PARIS\ algorithm}}
\label{alg:paris}
\end{algorithm}

\begin{algorithm}[H]
\KwData{$\bpartmb{n}$, $\zeta_{n + 1}$}
\KwResult{$\bpartmb{n + 1}$}
draw $\epartmb{n + 1} \sim \mkmb{n}[\zeta_{n + 1}](\epartmb{n|n}, \cdot)$\;
\For{$i \gets 1$ \KwTo $\N$}{
draw $ (\tilde{\xi}_{0:n}^{i, 1}, \stattd{n}{i, 1}) \sim \sum_{\ell = 1}^\N \frac{\ud{n}(\epart{n|n}{\ell}, \epart{n + 1}{i})}{\sum_{\ell' = 1}^\N \ud{n}(\epart{n|n}{\ell'}, \epart{n + 1}{i})} \delta_{\bpart{n}{\ell}}$\;
draw $ \{ (\eparttd{n}{i, j}, \stattd{n}{i, j}) \}_{j = 2}^\M \sim \left( \sum_{\ell = 1}^\N \frac{\ud{n}(\epart{n|n}{\ell}, \epart{n + 1}{i})}{\sum_{\ell' = 1}^\N \ud{n}(\epart{n|n}{\ell'}, \epart{n + 1}{i})} \delta_{(\epart{n|n}{\ell}, \stat{n}{\ell})} \right)^{\tensprod (\M - 1)}$\;
set $\stat{n + 1}{i} \gets \frac{1}{\M} \sum_{j = 1}^\M \left( \stattd{n}{i, j} + \afterm{n}(\eparttd{n}{i, j}, \epart{n + 1}{i}) \right)$\;
set $\epart{0:n + 1|n + 1}{i} \gets (\tilde{\xi}_{0:n}^{i, 1}, \epart{n + 1}{i})$ and $\bpart{n + 1}{i} \gets (\epart{0:n + 1|n + 1}{i}, \stat{n + 1}{i})$\;
}
\vspace{-10pt}
\caption{\scriptsize{One conditional \PARIS\  update}} \label{alg:parisian:Gibbs}
\end{algorithm}

\begin{algorithm}[H]
\KwData{$\chunk{\zeta}{0}{n}$}
\KwResult{$\bpartmb{n}$, $\chunk{\zeta'}{0}{n}$}
draw $\epartmb{0} \sim \initmb[\zeta_0]$\;
set $\bpartmb{0} \gets \{(\epart{0}{i}, 0)\}_{i = 1}^\N$\;
\For{$m \gets 0$ \KwTo $n - 1$}{
set $\bpartmb{m + 1} \gets
\condparis(\bpartmb{m}, \zeta_{m + 1})$\;
}
draw $\zeta_{0:n}' \sim \occm(\epartmb{0:n|n})$\;
\caption{\scriptsize{One iteration of the {\PARIS}ian particle Gibbs (\PPG)}} \label{alg:parisian:particle:Gibbs}
\end{algorithm}

\section{Additional Proofs}
\label{sec:proofs}

\subsection{Proof of \Cref{prop:backward:marginal}}
\label{sec:proof:prop:backward:marginal}
First, note that, by definitions \eqref{eq:def:ck} and \eqref{eq:def:ckjt}, 
\begin{align*} 
H_n(\xmb[0:n]) &\eqdef \int \ckjt{n}(\xmb[0:n], \rmd \ymb[n]) \, \occm(\xmb[0:n|n]) h \\
&= \idotsint \left( \frac{1}{\N} \sum_{j_n = 1}^\N h(x_{0:n-1|n}^{j_n}, x_{n}^{j_n}) \right) 
\\ 
&\quad \times 
\prod_{m = 0}^{n - 1} \prod_{i_{m + 1} = 1}^\N \int \sum_{j_m = 1}^\N \frac{\ud{m}(x_m^{j_m}, x_{m + 1}^{i_{m + 1}})}{\sum_{j'_m = 1}^\N \ud{m}(x_m^{j'_m}, x_{m + 1}^{i_{m + 1}})}
\delta_{x_{0:m|m}^{j_m}}(\rmd x_{0:m | m+1}^{i_{m + 1}}),
\end{align*}
where $x_{0:-1|0}^i = \emptyset$ for all $i \in \intvect{1}{\N}$ by convention.
We will show that for every $k \in \intvect{0}{n}$, $H_{k, n} \equiv H_n$, where
\begin{equation*} \label{eq:induction:hyp}
H_{k, n}(\xmb[0:n]) \eqdef \frac{1}{\N} \sum_{j_n = 1}^\N \cdots \sum_{j_k = 1}^\N \prod_{\ell = k}^{n - 1} \frac{\ud{\ell}(x_\ell^{j_\ell}, x_{\ell + 1}^{j_{\ell + 1}})}{\sum_{j_\ell' = 1}^\N \ud{\ell}(x_\ell^{j'_\ell}, x_{\ell + 1}^{j_{\ell + 1}})}
a_{k,n}(\xmb[0], \ldots, \xmb[k - 1], x_k^{j_k}, \ldots, x_n^{j_n})
\end{equation*}
with 
\begin{multline*}
    a_{k,n}(\xmb[0], \ldots, \xmb[k - 1], x_k^{j_k}, \ldots, x_n^{j_n}) \\ 
    = \int \prod_{m = 0}^{k - 1} \prod_{i_{m + 1} = 1}^\N
\sum_{j_m = 1}^\N \frac{\ud{m}(x_m^{j_m}, x_{m + 1}^{i_{m + 1}})}{\sum_{j'_m = 1}^\N \ud{m}(x_m^{j'_m}, x_{m + 1}^{i_{m + 1}})}
\delta_{x_{0:m|m}^{j_m}}(\rmd x_{0:m|m+1}^{i_{m + 1}}) h(x_{0:k-1|k}^{j_k}, x_{k}^{j_k}, \ldots, x_n^{j_n}). 
\end{multline*}
Since, by convention, $\prod_{\ell = n}^{n - 1} \ldots = 1$, $H_{n, n}(\xmb[0:n]) = \N^{-1} \sum_{j_n = 1}^\N a_{n,n}(\xmb[0], \ldots, \xmb[n - 1], x_n^{j_n})$, 
and we note that 
$H_n \equiv H_{n, n}$.
We now show that $H_{k, n} \equiv H_{k - 1, n}$ for every $k \in \intvect{1}{n}$; for this purpose, note that
\begin{multline*}
    a_{k,n}(\xmb[0], \ldots, \xmb[k - 1], x_k^{j_k}, \ldots, x_n^{j_n}) \\ = \int \prod_{m = 0}^{k - 2} \prod_{i_{m + 1} = 1}^\N \sum_{j_m = 1}^\N \frac{\ud{m}(x_m^{j_m}, x_{m + 1}^{i_{m + 1}})}{\sum_{j'_m = 1}^\N \ud{m}(x_m^{j'_m}, x_{m + 1}^{i_{m + 1}})}
\delta_{x_{0:m|m}^{j_m}}(\rmd x_{0:m|m+1}^{i_{m + 1}}) \\
    \times \int \prod_{i_{k} = 1}^\N
    \sum_{j_{k-1} = 1}^\N \frac{\ud{k-1}(x_{k-1}^{j_{k-1}}, x_{k}^{i_{k}})}{\sum_{j'_{k-1} = 1}^\N \ud{k-1}(x_{k-1}^{j'_{k-1}}, x_{k}^{i_{k}})}
    \delta_{x_{0:k-1|k-1}^{j_{k-1}}}(\rmd x_{0:{k-1}|k}^{i_{k}}) \, h(x_{0:k-1|k}^{j_k}, x_{k}^{j_k}, \ldots, x_n^{j_n}),
\end{multline*}
and since $x_{0:k-1|k-1}^{j_{k-1}} = (x_{0:k-2|k-1}^{j_{k-1}}, x_{k-1}^{j_{k-1}})$, it holds that 
\begin{multline*}
\int \prod_{i_{k} = 1}^\N
\sum_{j_{k-1} = 1}^\N \frac{\ud{k-1}(x_{k-1}^{j_{k-1}}, x_{k}^{i_{k}})}{\sum_{j'_{k-1} = 1}^\N \ud{k-1}(x_{k-1}^{j'_{k-1}}, x_{k}^{i_{k}})}
\delta_{x_{0:k-1|k-1}^{j_{k-1}}}(\rmd x_{0:{k-1}|k}^{i_{k}}) \, h(x_{0:k-1|k}^{j_k}, x_{k}^{j_k}, \ldots, x_n^{j_n}) \\
    =\sum_{j_{k-1} = 1}^\N \frac{\ud{k-1}(x_{k-1}^{j_{k-1}}, x_{k}^{j_{k}})}{\sum_{j'_{k-1} = 1}^\N \ud{k-1}(x_{k-1}^{j'_{k-1}}, x_{k}^{j_{k}})}  h(x_{0:k-2|k-1}^{j_{k-1}}, x_{k-1}^{j_{k-1}}, x_{k}^{j_k}, \ldots, x_n^{j_n}).
\end{multline*}
Therefore, we obtain 
\begin{multline*}
a_{k,n}(\xmb[0], \ldots, \xmb[k - 1], x_k^{j_k}, \ldots, x_n^{j_n}) \\
= \int \prod_{m = 0}^{k - 2} \prod_{i_{m + 1} = 1}^\N
\sum_{j_m = 1}^\N \frac{\ud{m}(x_m^{j_m}, x_{m + 1}^{i_{m + 1}})}{\sum_{j'_m = 1}^\N \ud{m}(x_m^{j'_m}, x_{m + 1}^{i_{m + 1}})}
\delta_{x_{0:m|m}^{j_m}}(\rmd x_{0:m|m + 1}^{i_{m + 1}}) \\
\times \sum_{j_{k - 1} = 1}^\N \frac{\ud{k - 1}(x_{k - 1}^{j_{k - 1}}, x_k^{j_k})}{\sum_{j'_{k - 1} = 1}^\N \ud{k - 1}(x_{k - 1}^{j'_{k - 1}}, x_k^{j_k})} h(x_{0:k - 2|k - 1}^{j_{k - 1}}, x_{k - 1}^{j_{k - 1}}, x_k^{j_k},
\ldots, x_n^{j_n}).
\end{multline*}
Now, changing the order of summation with respect to $j_{k - 1}$ and integration on the right hand side of the previous display yields
\begin{multline*}
a_{k,n}(\xmb[0], \ldots, \xmb[k - 1], x_k^{j_k}, \ldots, x_n^{j_n})
\\ = \sum_{j_{k - 1} = 1}^\N \frac{\ud{k - 1}(x_{k - 1}^{j_{k - 1}}, x_k^{j_k})}{\sum_{j'_{k - 1} = 1}^\N \ud{k - 1}(x_{k - 1}^{j'_{k - 1}}, x_k^{j_k})} a_{k - 1,n}(\xmb[0], \ldots, \xmb[k - 2], x_{k - 1}^{j_{k - 1}}, \ldots, x_n^{j_n}).
\end{multline*}
Thus,
\begin{align*}
\lefteqn{H_{k, n}(\xmb[0:n])} \\
&= \frac{1}{\N} \sum_{j_n = 1}^\N \cdots \sum_{j_k = 1}^\N \prod_{\ell = k}^{n - 1} \frac{\ud{\ell}(x_\ell^{j_\ell}, x_{\ell + 1}^{j_{\ell + 1}})}{\sum_{j_\ell' = 1}^\N \ud{\ell}(x_\ell^{j'_\ell}, x_{\ell + 1}^{j_{\ell + 1}})} \\
& \quad \times \sum_{j_{k - 1} = 1}^\N \frac{\ud{k - 1}(x_{k - 1}^{j_{k - 1}}, x_k^{j_k})}{\sum_{j'_{k - 1} = 1}^\N \ud{k - 1}(x_{k - 1}^{j'_{k - 1}}, x_k^{j_k})} a_{k - 1,n}(\xmb[0], \ldots, \xmb[k - 2], x_{k - 1}^{j_{k - 1}}, \ldots, x_n^{j_n}) \\
&= \frac{1}{\N} \sum_{j_n = 1}^\N \cdots \sum_{j_{k - 1} = 1}^\N \prod_{\ell = k - 1}^{n - 1} \frac{\ud{\ell}(x_\ell^{j_\ell}, x_{\ell + 1}^{j_{\ell + 1}})}{\sum_{j_\ell' = 1}^\N \ud{\ell}(x_\ell^{j'_\ell}, x_{\ell + 1}^{j_{\ell + 1}})} a_{k - 1,n}(\xmb[0], \ldots, \xmb[k - 2], x_{k - 1}^{j_{k - 1}}, \ldots, x_n^{j_n}) \\
&= H_{k-1, n}(\xmb[0:n]),
\end{align*}
which establishes the recursion. Therefore, $H_n \equiv H_{0, n}$ and we may now conclude the proof by noting that $\bdpart{n}h \equiv H_{0, n}$. 


\subsection{Proof of \Cref{cor:exp:conc:cond:paris}}
\label{sec:proof:exp:concentration}

In order to establish \Cref{cor:exp:conc:cond:paris} we will prove the following more general result, of which \Cref{cor:exp:conc:cond:paris} is a direct consequence.

\begin{proposition}
\label{thm:exp:conc:cond:paris}
For every $n \in \nset$ and $M \in \nsetpos$ there exist $\cstcondparisc_n > 0$ and $\cstcondparisd_n > 0$ such that for every $\N \in \nsetpos$, $z_{0:n} \in \xpsp{0}{n}$, $(f_n, \tilde{f}_n) \in \bmf(\xfd{n})^2$, and $\varepsilon > 0$,
\begin{multline*}
\int \mbjt{n} \ckjt{n}(z_{0:n}, \rmd \statlmb_n) \\
\times \indin{\left|\frac{1}{\N} \sum_{i = 1}^\N \{b_{n}^i f_n(x_{n|n}^{i}) + \tilde{f}_n(x_{n|n}^{i}) \} - \targ[z_{0:n}]{n}(f_n \rk[z_{0:n - 1}]{n} \af{n} + \tilde{f}_n) \right|\geq \varepsilon} \\
\leq \cstcondparisc_n \exp \left( -   \frac{\cstcondparisd_n \N \varepsilon^2}{2 \upkappa_n^2} \right),
\end{multline*}
where
\begin{equation}
\label{eq:definition:upkappa}
\upkappa_n \eqdef \| f_n \|_\infty \sum_{m = 0}^{n - 1} \| \afterm{m} \|_\infty + \| \tilde{f}_n \|_\infty.
\end{equation}
\end{proposition}

To prove \Cref{thm:exp:conc:cond:paris} we need the following technical lemma.

\begin{lemma} \label{lem:key:identity:untarg:affine:ii}
For every $n \in \nset$, $(f_{n + 1}, \tilde{f}_{n + 1}) \in \bmf(\xfd{n + 1})^2$, $z_{0:n + 1} \in \xpsp{0}{n + 1}$, and $\N \in \nsetpos$,
\begin{multline*}
\untarg[z_{0:n + 1}]{n + 1}(f_{n + 1} \rk[{\chunk{z}{0}{n}}]{n + 1} \af{n + 1} + \tilde{f}_{n + 1}) \\
= \left( 1- \frac{1}{\N} \right) \untarg[{\chunk{z}{0}{n}}]{n}\{ \uk{n} f_{n + 1} \rk[z_{0:n - 1}]{n} \af{n} + \uk{n}(\afterm{n} f_{n + 1} + \tilde{f}_{n + 1}) \} \\
+ \frac{1}{\N} \untarg[{\chunk{z}{0}{n}}]{n} \pot{n} \left( f_{n + 1}(z_{n + 1}) \rk[{\chunk{z}{0}{n}}]{n + 1} \af{n + 1}(z_{n + 1}) + \tilde{f}_{n + 1}(z_{n + 1}) \right).
\end{multline*}
\end{lemma}

\begin{proof}
    Since \Cref{lem:key:identity:untarg:affine} holds also for the Feynman--Kac model with a frozen path, we obtain
    \begin{multline*}
        \untarg[{\chunk{z}{0}{n+1}}]{n + 1}(f_{n + 1} \rk[{\chunk{z}{0}{n}}]{n + 1} \af{n + 1} + \tilde{f}_{n + 1}) \\
        = \untarg[{\chunk{z}{0}{n}}]{n}\{ \uk[z_{n+1}]{n} f_{n + 1} \rk[{\chunk{z}{0}{n}}]{n} \af{n} + \uk[z_{n+1}]{n}(\afterm{n} f_{n + 1} + \tilde{f}_{n + 1}) \}.
    \end{multline*}
    Thus, the proof is concluded by noting that for every $x_n \in \xsp{n}$ and $h \in \bmf(\xfd{n:n+1})$,
    $$
    \uk[z_{n+1}]{n}h(x_n) = \left(1 - \frac{1}{\N} \right) \uk{n} h(x_n) + \frac{1}{\N} g(x_n) h(x_n, z_{n + 1}).
    $$
\end{proof}

Finally, before proceeding to the proof of \Cref{thm:exp:conc:cond:paris}, we introduce the law of the {\PARIS} evolving conditionally on a frozen path $\zpath = \{z_m\}_{m \in \nset}$.
Define, for $m \in \nset$ and $z_{m + 1} \in \xsp{m + 1}$,
\begin{equation*} 
\pk{m}[z_{m + 1}] : \yspmb{m} \times \yfdmb{m + 1} \ni (\ymb[m], A) \mapsto \int
\, \mkmb{m}[z_{m + 1}](\xmb[m|m], \rmd \xmb[m + 1]) \, \ck{m}(\ymb[m], \xmb[m + 1], A).
\end{equation*}
For any given initial distribution $\pinit \in \probmeas(\yfdmb{0})$, let $\canlaw[\pkl, \zpath]{\pinit}$ be the distribution of the canonical Markov chain induced by the Markov kernels $\{ \pk{m}[z_{m + 1}] \}_{m \in \nset}$ and the initial distribution $\pinit$. By abuse of notation we write
$\canlaw[\pkl, \zpath]{\initmb}$ instead of $\canlaw[\pkl, \zpath]{\pinit[\initmb[z_0]]}$, where the extension $\pinit[\initmb]$ is defined in \Cref{sec:proof:unbiasedness}.

\begin{proof}[Proof of \Cref{thm:exp:conc:cond:paris}]
We proceed by forward induction over $n$. Let the $\sigma$-fields $\partfiltbar{n}$ and $\partfilt{n}$ be defined as in the proof of \Cref{thm:unbiasedness}, but for the conditional {\PARIS} dual process. Then, under the law $\canlaw[\pkl, \zpath]{\initmb}$, reusing \eqref{eq:cond:exp:beta},
\begin{align*}
&\canlawexp[\pkl, \zpath]{\initmb} \left[ \stat{n}{1} f_n(\epart{n}{1}) + \tilde{f}_n(\epart{n}{1}) \mid \partfiltbar{n - 1} \right] \nonumber \\
&= \canlawexp[\pkl, \zpath]{\initmb} \left[ \canlawexp[\pkl, \zpath]{\initmb} \left[ \stat{n}{1} \mid \partfilt{n} \right] f_n(\epart{n}{1}) + \tilde{f}_n(\epart{n}{1}) \mid \partfiltbar{n - 1} \right] \nonumber \\
&= \canlawexp[\pkl, \zpath]{\initmb} \left[ f_n(\epart{n}{1}) \sum_{\ell = 1}^\N \frac{\ud{n - 1}(\epart{n - 1}{\ell}, \epart{n}{1})}{\sum_{\ell' = 1}^\N \ud{n - 1}(\epart{n - 1}{\ell'}, \epart{n}{1})} \left( \stat{n - 1}{\ell} + \afterm{n - 1}(\epart{n - 1}{\ell}, \epart{n}{1}) \right)
\vphantom{\sum_{\ell' = 1}^\N} + \tilde{f}_n(\epart{n}{1}) \mid \partfiltbar{n - 1} \right]. 
\end{align*}
Using \eqref{eq:kernel:cond:dual}, we get
\begin{multline} \label{eq:cond:exp:split}
\canlawexp[\pkl, \zpath]{\initmb} \left[ \stat{n}{1} f_n(\epart{n}{1}) + \tilde{f}_n(\epart{n}{1}) \mid \partfiltbar{n - 1} \right] \\
= \left( 1 - \frac{1}{\N} \right) \frac{\sum_{\ell = 1}^\N \{\stat{n - 1}{\ell} \uk{n - 1} f_n (\epart{n - 1}{\ell}) + \uk{n - 1}(\afterm{n - 1} f_n + \tilde{f}_n)(\epart{n - 1}{\ell})\}}{\sum_{\ell' = 1}^\N \pot{n - 1}(\epart{n - 1}{\ell'})} \\
+ \frac{1}{\N} \left( f_n(z_n) \sum_{\ell = 1}^\N \frac{\ud{n - 1}(\epart{n - 1}{\ell}, z_n)}{\sum_{\ell' = 1}^\N \ud{n - 1}(\epart{n - 1}{\ell'}, z_n)} \left( \stat{n - 1}{\ell} + \afterm{n}(\epart{n - 1}{\ell}, z_n) \right) + \tilde{f}_n(z_n) \right).
\end{multline}
In order to apply the induction hypothesis to each term on the right-hand side of the previous identity, note that
$$
\rk[z_{0:n - 1}]{n} \af{n}(z_n)
= \frac{\targ[z_{0:n - 1}]{n - 1}[ \ud{n - 1}(\cdot, z_n) \{ \rk[z_{0:n - 2}]{n - 1} \af{n - 1}(\cdot) + \afterm{n - 1}(\cdot, z_n) \} ] }{\targ[z_{0:n - 1}]{n - 1} [ \ud{n - 1}(\cdot, z_n) ]}.
$$
Therefore, using \Cref{lem:key:identity:untarg:affine:ii} and noting that
$\untarg[{\chunk{z}{0}{n}}]{n}\indi{\xsp{n}}/ \untarg[{\chunk{z}{0}{n}}]{n-1} \indi{\xsp{n-1}} = \targ[{\chunk{z}{0}{n-1}}]{n-1}\pot{n-1}$ yields
\begin{multline}
    \label{eq:cond:exp:recursion}
    \targ[{\chunk{z}{0}{n}}]{n}(f_{n} \rk[\chunk{z}{0}{n-1}]{n} \af{n} + \tilde{f}_{n}) = \frac{1}{\N} \left( f_{n}(z_{n}) \rk[\chunk{z}{0}{n-1}]{n} \af{n}(z_{n}) + \tilde{f}_{n}(z_{n}) \right) \\
    + \left( 1- \frac{1}{\N} \right) \frac{\targ[z_{0:n - 1}]{n - 1} \{ \uk{n - 1} f_n \rk[z_{0:n - 2}]{n-1} \af{n} + \uk{n - 1}(\afterm{n - 1} f_n + \tilde{f}_n) \}}{\targ[z_{0:n - 1}]{n - 1} \pot{n - 1}}.
\end{multline}
By combining \eqref{eq:cond:exp:split} with \eqref{eq:cond:exp:recursion}, we decompose the error according to
\begin{align} \label{eq:error_decomposition_paris}
\lefteqn{\frac{1}{\N} \sum_{i = 1}^\N \{\stat{n}{i} f_n(\epart{n|n}{i}) + \tilde{f}_n(\epart{n|n}{i}) \} - \targ[{\chunk{z}{0}{n}}]{n}(f_n \rk[z_{0:n - 1}]{n} \af{n} + \tilde{f}_n)} \hspace{10mm} \nonumber \\
&= \frac{1}{\N} \sum_{i = 1}^\N \{\stat{n}{i} f_n(\epart{n|n}{i}) + \tilde{f}_n(\epart{n|n}{i}) \} -
\canlawexp[\pkl, \zpath]{\initmb} \left[ \stat{n}{1} f_n(\epart{n}{1}) + \tilde{f}_n(\epart{n}{1})
\mid \partfiltbar{n - 1} \right] \nonumber \\
&\hspace{10mm} + \canlawexp[\pkl, \zpath]{\initmb} \left[
\stat{n}{1} f_n(\epart{n}{1}) + \tilde{f}_n(\epart{n}{1}) \mid \partfiltbar{n - 1} \right] - \targ[{\chunk{z}{0}{n}}]{n}(f_n \rk[z_{0:n - 1}]{n} \af{n} + \tilde{f}_n) \nonumber \\
&= \operatorname{I}^{(1)}_{\N} + \left( 1 - \frac{1}{\N} \right) \operatorname{I}^{(2)}_{\N} + \frac{1}{\N} \operatorname{I}^{(3)}_{\N},
\end{align}
where
\begin{align}
\operatorname{I}^{(1)}_{\N} &\eqdef \frac{1}{\N} \sum_{i = 1}^\N \{\stat{n}{i} f_n(\epart{n}{i}) + \tilde{f}_n(\epart{n}{i}) \}
-  \canlawexp[\pkl, \zpath]{\initmb} \left[ \stat{n}{1} f_n(\epart{n}{1}) + \tilde{f}_n(\epart{n}{1}) \mid \partfiltbar{n - 1} \right], \nonumber \\
\operatorname{I}^{(2)}_{\N} &\eqdef \frac{\sum_{\ell = 1}^\N \{\stat{n - 1}{\ell} \uk{n - 1} f_n (\epart{n - 1}{\ell}) + \uk{n - 1}(\afterm{n - 1} f_n + \tilde{f}_n)(\epart{n - 1}{\ell})\}}{\sum_{\ell' = 1}^\N \pot{n - 1}(\epart{n - 1}{\ell'})} \nonumber \\
&\hspace{10mm} - \frac{\targ[z_{0:n - 1}]{n - 1} \{ \uk{n - 1} f_n \rk[z_{0:n - 1}]{n} \af{n} + \uk{n - 1}(\afterm{n - 1} f_n + \tilde{f}_n) \}}{\targ[z_{0:n - 1}]{n - 1} \pot{n - 1}}, \label{eq:def:b_N}
\end{align}
and
\begin{multline} \label{eq:def:c_N}
\operatorname{I}^{(3)}_{\N} \eqdef f_n(z_n) \sum_{\ell = 1}^\N \frac{\ud{n - 1}(\epart{n - 1}{\ell}, z_n)}{\sum_{\ell' = 1}^\N \ud{n - 1}(\epart{n - 1}{\ell'}, z_n)} \left( \stat{n - 1}{\ell} + \afterm{n - 1}(\epart{n - 1}{\ell}, z_n) \right) \\
- f_n(z_n) \frac{\targ[z_{0:n - 1}]{n - 1}[ \ud{n - 1}(\cdot, z_n) \{ \rk[z_{0:n - 2}]{n - 1} \af{n - 1}(\cdot) + \afterm{n - 1}(\cdot, z_n) \} ] }{\targ[z_{0:n - 1}]{n - 1} [ \ud{n - 1}(\cdot, z_n) ]}.
\end{multline}

The proof is now completed by treating the terms $\operatorname{I}^{(1)}_{\N}$, $\operatorname{I}^{(2)}_{\N}$, and $\operatorname{I}^{(3)}_{\N}$ separately, using Hoeffding's inequality and its generalisation in \cite[Lemma~4] {douc:garivier:moulines:olsson:2009}. Choose $\varepsilon > 0$; then, by Hoeffding's inequality,
\begin{equation} \label{eq:a_N:bound}
\canlaw[\pkl, \zpath]{\initmb} \left( | \operatorname{I}^{(1)}_{\N} |\geq \varepsilon \right) \leq 2\exp \left( - \frac{1}{2} \frac{\varepsilon^2}{\upkappa_n^2}  \N \right).
\end{equation}
To treat $\operatorname{I}^{(2)}_{\N}$, we apply the induction hypothesis to the numerator and denominator, each normalised by $1 / \N$, yielding, since $\| \uk{n - 1} h \|_\infty \leq \pothigh{n - 1} \| h \|_\infty$ for all $h \in \bmf(\xfd{n - 1} \tensprod \xfd{n})$,
\begin{multline*}
\canlaw[\pkl, \zpath]{\initmb}  \left( \left| \frac{1}{\N} \sum_{\ell = 1}^\N \{\stat{n - 1}{\ell} \uk{n - 1} f_n (\epart{n - 1}{\ell}) + \uk{n - 1}(\afterm{n - 1} f_n + \tilde{f}_n)(\epart{n - 1}{\ell}) \} \right. \right. \\
\left. \left. - \targ[z_{0:n - 1}]{n - 1} \{ \uk{n - 1} f_n \rk[z_{0:n - 1}]{n} \af{n} + \uk{n - 1}(\afterm{n - 1} f_n + \tilde{f}_n) \} \vphantom{\sum_{\ell = 1}^\N} \right| \geq \varepsilon \right) \\
\leq \cstcondparisc_{n - 1} \exp \left( - \cstcondparisd_{n - 1} \frac{\varepsilon^2}{\pothigh{n - 1}^2 \upkappa_n^2}  \N \right)
\end{multline*}
and
$$
\canlaw[\pkl, \zpath]{\initmb} \left( \left| \frac{1}{\N} \sum_{\ell = 1}^\N \pot{n - 1}(\epart{n - 1}{\ell}) - \targ[z_{0:n - 1}]{n - 1} \pot{n - 1} \right|\geq \varepsilon \right)
\leq \cstcondparisc_{n - 1} \exp \left( - \cstcondparisd_{n - 1} \frac{\varepsilon^2}{\pothigh{n - 1}^2} \N \right).
$$
Combining the previous two bounds with the generalised Hoeffding inequality in \cite[Lemma~4]{douc:garivier:moulines:olsson:2009} yields, using also the bounds
$$
\frac{\sum_{\ell = 1}^\N \{\stat{n - 1}{\ell} \uk{n - 1} f_n (\epart{n - 1}{\ell}) + \uk{n - 1}(\afterm{n - 1} f_n + \tilde{f}_n)(\epart{n - 1}{\ell})\}}{\sum_{\ell' = 1}^\N \pot{n - 1}(\epart{n - 1}{\ell'})} \\
\leq \upkappa_n
$$
and $\targ[z_{0:n - 1}]{n - 1} \pot{n - 1} \geq \potlow{n - 1}$,
the inequality
\begin{equation} \label{eq:b_N:bound}
\canlaw[\pkl, \zpath]{\initmb} \left( | \operatorname{I}^{(2)}_{\N} |\geq \varepsilon \right)
\leq \cstcondparisc_{n - 1} \exp \left( - \cstcondparisd_{n - 1}  \frac{\potlow{n - 1}^2 \varepsilon^2}{ \pothigh{n - 1}^2 \upkappa^2_n}  \N \right).
\end{equation}
The last term $\operatorname{I}^{(3)}_{\N}$ is treated along similar lines; indeed, by the induction hypothesis, since $\|\ud{n - 1} \|_\infty \leq \pothigh{n - 1} \mdhigh{n - 1}$,
\begin{multline*}
\canlaw[\pkl, \zpath]{\initmb} \left( \left| \frac{1}{\N} \sum_{\ell = 1}^\N \ud{n - 1}(\epart{n - 1}{\ell}, z_n) \left( \stat{n - 1}{\ell} + \afterm{n - 1}(\epart{n - 1}{\ell}, z_n) \right) \right. \right. \\
\left. \left. \vphantom{\sum_{\ell = 1}^\N} - \targ[z_{0:n - 1}]{n - 1}[ \ud{n - 1}(\cdot, z_n) \{ \rk[z_{0:n - 1}]{n - 1} \af{n - 1}(\cdot) + \afterm{n - 1}(\cdot, z_n) \} ] \right| \geq \varepsilon \right) \\
\leq \cstcondparisc_{n - 1} \exp \left( - \cstcondparisd_{n - 1} \left( \frac{\varepsilon}{\pothigh{n - 1} \mdhigh{n - 1} \sum_{m = 0}^{n - 1} \| \afterm{m} \|_\infty} \right)^2 \N \right)
\end{multline*}
and
\begin{multline*}
\canlaw[\pkl, \zpath]{\initmb} \left( \left| \frac{1}{\N} \sum_{\ell = 1}^\N \ud{n - 1}(\epart{n - 1}{\ell}, z_n) - \targ[z_{0:n - 1}]{n - 1} [ \ud{n - 1}(\cdot, z_n) ] \right| \geq \varepsilon \right) \\
\leq \cstcondparisc_{n - 1} \exp \left( - \cstcondparisd_{n - 1} \left( \frac{\varepsilon}{\pothigh{n - 1} \mdhigh{n - 1}} \right)^2 \N \right).
\end{multline*}
Thus, since
$$
\sum_{\ell = 1}^\N \frac{\ud{n - 1}(\epart{n - 1}{\ell}, z_n)}{\sum_{\ell' = 1}^\N \ud{n - 1}(\epart{n - 1}{\ell'}, z_n)} \left( \stat{n - 1}{\ell} + \afterm{n - 1}(\epart{n - 1}{\ell}, z_n) \right) \leq \sum_{m = 0}^{n - 1} \| \afterm{m} \|_\infty
$$
and $\targ[z_{0:n - 1}]{n - 1} [ \ud{n - 1}(\cdot, z_n) ] \geq \potlow{n - 1}$, the generalised Hoeffding inequality provides 
\begin{equation} \label{eq:c_N:bound}
\canlaw[\pkl, \zpath]{\initmb} \left( | \operatorname{I}^{(3)}_{\N} |\geq \varepsilon \right)
\leq \cstcondparisc_{n - 1} \exp \left( - \cstcondparisd_{n - 1} \left( \frac{\potlow{n - 1} \varepsilon}{2 \pothigh{n - 1} \mdhigh{n - 1} \|f_n \|_\infty \sum_{m = 0}^{n - 1} \| \afterm{m} \|_\infty} \right)^2 \N \right).
\end{equation}
Finally, combining the bounds (\ref{eq:a_N:bound}--\ref{eq:c_N:bound}) completes the proof. 
\end{proof}


\subsection{Proof of \Cref{cor:Lp:cond:paris}}
\label{subsec:proof:cor:Lp:cond:paris}

The statement of \Cref{cor:Lp:cond:paris} is implied by the following more general result, which we will prove below.

\begin{proposition} \label{cor:Lp:cond:paris:affine}
For every $n \in \nset$, $M \in \nsets$, $\N \in \nsetpos$, $z_{0:n} \in \xpsp{0}{n}$, $(f_n, \tilde{f}_n) \in \bmf(\xfd{n})^2$, and $p \geq 2$, it holds that
\begin{multline*}
\int \mbjt{n} \ckjt{n}(z_{0:n}, \rmd \statlmb_n) \left| \frac{1}{\N} \sum_{i = 1}^\N \{b_{n}^i f_n(x_{n|n}^{i}) + \tilde{f}_n(x_{n|n}^{i}) \} - \targ[z_{0:n}]{n}(f_n \rk[z_{0:n - 1}]{n} \af{n} + \tilde{f}_n) \right|^p \\
\leq \cstcondparisc_n (p / \cstcondparisd_n)^{p/2} 
\N^{- p/2} \upkappa_n^p,
\end{multline*}
where $\cstcondparisc_n > 0$, $\cstcondparisd_n > 0$ and $\upkappa_n$ are defined in \Cref{thm:exp:conc:cond:paris} and \eqref{eq:definition:upkappa}, respectively.
\end{proposition}

Before proving \Cref{cor:Lp:cond:paris:affine}, we establish the following result.
\begin{lemma}
\label{lem:Lp}
 Let $X$ be an $\rset^d$-valued random variable, defined on some probability space $(\Omega, \mathcal{F}, \prob)$, satisfying $\mathbb{P}(|X| \geq t) \leq  c \exp(-t^2 / (2 \sigma^{2}))$ for every $t \geq 0$ and some $c > 0$ and $\sigma > 0$. Then for every $p \geq 2$ it holds that $\mathbb{E}[|X|^p] \leq c p^{p / 2} \sigma^{p}$.
 \end{lemma}
\begin{proof}
Using Fubini's theorem and the change of variable formula,
$$
\E \left[ |X|^p \right]=\int_{0}^{\infty} p t^{p-1} \mathbb{P}(|X| \geq t) \, \rmd t = c p 2^{p / 2 - 1} \sigma^{p} \Gamma(p / 2),
$$
where $\Gamma$ is the Gamma function. It remains to apply the bound $\Gamma(p / 2) \leq(p / 2)^{p / 2-1}$ (see \cite{anderson:qiu:1997}), which holds for $p \geq 2$ by [2, Theorem 1.5].
\end{proof}
\begin{proof}[Proof of \Cref{cor:Lp:cond:paris:affine}]
By combining \Cref{thm:exp:conc:cond:paris} and  \Cref{lem:Lp} we obtain
\begin{multline*}
\N \, \int \mbjt{n} \ckjt{n}(z_{0:n}, \rmd \statlmb_n) \left| \frac{1}{\N} \sum\nolimits_{i = 1}^\N \{b_{n}^i f_n(x_{n|n}^{i}) + \tilde{f}_n(x_{n|n}^{i}) \} - \targ[z_{0:n}]{n}(f_n \rk[z_{0:n - 1}]{n} \af{n} + \tilde{f}_n) \right|^2 \\ 
\leq \cstcondparisc_n (p / \cstcondparisd_n)^{p/2} 
\N^{- p/2}
\left( \| f_n \|_\infty \sum_{m = 0}^{n - 1} \| \afterm{m} \|_\infty + \| \tilde{f}_n \|_\infty \right)^p,
\end{multline*}
which was to be established.
\end{proof}

\subsection{Proof of \Cref{prop:bias:cond:paris}}
\label{subsec:prop:bias:cond:paris}

Like previously, we establish \Cref{prop:bias:cond:paris} via a more general result, namely the following.

\begin{proposition} \label{prop:bias:cond:paris:affine}
For every $n \in \nset$, the exists $\cstcondparisbias< \infty$ such that for every $M \in \nsets$, $\N \in \nsetpos$, $z_{0:n} \in \xpsp{0}{n}$, and $(f_n, \tilde{f}_n) \in \bmf(\xfd{n})^2$, 
\begin{multline*}
 \left| \int \mbjt{n} \ckjt{n}(z_{0:n}, \rmd \statlmb_n) \frac{1}{\N} \sum_{i = 1}^\N \{b_{n}^i f_n(x_{n|n}^{i}) + \tilde{f}_n(x_{n|n}^{i}) \} - \targ[z_{0:n}]{n}(f_n \rk[z_{0:n - 1}]{n} \af{n} + \tilde{f}_n) \right| \\
\leq \cstcondparisbias \upkappa_n  \N^{- 1}, 
\end{multline*}
where $\upkappa_n$ is defined in \eqref{eq:definition:upkappa}.
\end{proposition}

We preface the proof of \Cref{prop:bias:cond:paris:affine} by a technical lemma providing a bound on the bias of ratios of random variables.
\begin{lemma}
\label{lemma:bias-estimator-general}
Let $\upalpha$ and $\upbeta$ be (possibly dependent) random variables defined on some probability space $(\Omega, \mathcal{F}, \prob)$ and such that $\E[\upalpha^2] < \infty$ and $\E[\upbeta^2] < \infty$. Moreover, assume that there exist $c > 0$ and $d > 0$ such that $|\upalpha / \upbeta| \leq c$, $\prob$-a.s., $|a/b| \leq c$, $\E[ (\upalpha - a)^2] \leq c^2 d^2$, and $\E[ (\upbeta - b)^2] \leq d^2$. Then
\begin{equation}
\left| \E[ \upalpha / \upbeta] - a / b  \right| \leq  2 c(d/b)^2 + c |\E[\upbeta - b]|/|b| + |\E[\upalpha - a]|/|b|.
\end{equation}
\end{lemma}

\begin{proof}
Using the identity
$$
\E[\upalpha / \upbeta] - a/b = \E[ (\upalpha/ \upbeta)(b - \upbeta)^2] / b^2 + \E[(\upalpha - a)(b - \upbeta)] / b^2  + a\E[b - \upbeta] / b^2 + \E[\upalpha - a]/b,
$$
the claim is established by applying the Cauchy--Schwarz inequality and the assumptions of the lemma according to
\begin{align}
\lefteqn{\left|\E[\upalpha / \upbeta] - a/b\right|} \hspace{5mm}  \nonumber \\
&\leq c \E[ (\upbeta - b)^2] / b^2 + \{\E[(\upalpha - a)^2] \E[(\upbeta - b)^2]\}^{1/2} / b^2 + |a||\E[\upbeta - b]|/b^2 + |\E[\upalpha - a]|/b^2  \nonumber \\
&\leq 2 c (d / b)^2 + c|\E[\upbeta - b]|/|b| + |\E[\upalpha - a]|/|b|. \nonumber
\end{align}
\end{proof}

\begin{proof}[Proof of \Cref{prop:bias:cond:paris}]
We proceed by induction and assume that the claim holds true for $n - 1$. Reusing the error decomposition \eqref{eq:error_decomposition_paris}, it is enough to bound the expectations of the terms $\operatorname{I}^{(2)}_\N$ and $\operatorname{I}^{(3)}_\N$ given in \eqref{eq:def:b_N} and \eqref{eq:def:c_N}, respectively (since $\canlawexp[\pkl, \zpath]{\initmb}[\operatorname{I}^{(1)}_\N] = 0$). This will be done using the induction hypothesis, \Cref{lemma:bias-estimator-general}, and \Cref{cor:Lp:cond:paris:affine}. More precisely, to bound the expectation of $\operatorname{I}^{(2)}_\N$, we use \Cref{lemma:bias-estimator-general} with $\upalpha \leftarrow \upalpha_n$, $\upbeta \leftarrow \upbeta_n$, $a \leftarrow a_n$, and $b \leftarrow b_n$, where
\begin{align*}
&\upalpha_n \eqdef \frac{1}{\N} \sum_{\ell = 1}^\N \{\stat{n - 1}{\ell} \uk{n - 1} f_n (\epart{n - 1}{\ell}) + \uk{n - 1}(\afterm{n - 1} f_n + \tilde{f}_n)(\epart{n - 1}{\ell})\}, \quad
\upbeta_n \eqdef \frac{1}{\N} \sum_{\ell = 1}^\N \pot{n - 1}(\epart{n - 1}{\ell}), \\
& a_n \eqdef \targ[z_{0:n - 1}]{n - 1} \{ \uk{n - 1} f_n \rk[z_{0:n - 1}]{n} \af{n} + \uk{n - 1}(\afterm{n - 1} f_n + \tilde{f}_n) \}, \quad
b_n \eqdef \targ[z_{0:n - 1}]{n - 1} \pot{n - 1}.
\end{align*}
For this purpose, note that $|\upalpha_n / \upbeta_n | \leq \upkappa_n$ and $|a_n/b_n| \leq \upkappa_n$, where $\upkappa_n$ is defined in \eqref{eq:definition:upkappa}. On the other hand, using \Cref{cor:Lp:cond:paris:affine}  (applied with $p = 2$), we obtain
\begin{equation*}
\canlawexp[\pkl, \zpath]{\initmb}[(\upalpha_n - a_n)^2] \leq d_n^2 \upkappa_n^2  \quad \text{and} \quad \canlawexp[\pkl, \zpath]{\initmb}[(\upbeta_n - b_n )^2] \leq d_n^2,
\end{equation*}
where $d_n^2 \eqdef \cstcondparisc_n \gsupbound{n-1}^2 / (\cstcondparisd_n \N)$. 
Using the induction assumption, we get
\begin{equation*}
|\canlawexp[\pkl, \zpath]{\initmb}[\upalpha_n] - a_n| \leq \cstcondparisbias[n-1] \N^{-1} \gsupbound{n-1} \upkappa_n
\quad \text{and} \quad
|\canlawexp[\pkl, \zpath]{\initmb}[\upbeta_n] - b_n| \leq \cstcondparisbias[n-1] \N^{-1} \gsupbound{n-1}.
\end{equation*}
Hence, the conditions of \Cref{lemma:bias-estimator-general} are satisfied and we deduce that
\[
|\canlawexp[\pkl, \zpath]{\initmb}[\operatorname{I}^{(2)}_\N] | = |\canlawexp[\pkl, \zpath]{\initmb}[\upalpha_n/\upbeta_n] - a_n/b_n| \leq 2 \upkappa_n \frac{\cstcondparisc_n}{\cstcondparisd_n \N} \frac{\gsupbound{n-1}^2}{\ginfbound{n-1}^2} +
2 \cstcondparisbias[n-1] \upkappa_n  \frac{\gsupbound{n-1}}{\ginfbound{n-1} \N} .
\]
The bound on $|\canlawexp[\pkl, \zpath]{\initmb}[\operatorname{I}^{(2)}_\N]|$ is obtained along the same lines.
\end{proof}

\subsection{Proof of \Cref{theo:bias-mse-rolling}}
\label{sec:proof:theo:bias-mse-rolling}
We first consider the bias, which can be bounded according to 
\begin{align*}
    \left| \E_{\upxi}[\rollingestim[\ki_0][\ki][N][f]] - \targ{0:n} \af{n} \right| &\leq  (\ki - \ki_0)^{-1}  \sum_{\ell=\ki_0 + 1}^{\ki} \left| \E_{\upxi} \occm(\statmb{n}[\ell])(\operatorname{id})  - \targ{0:n} \af{n} \right| \\
    &\leq (\ki - \ki_0)^{-1} N^{-1} \cstparisbias   \left(\sum_{m=0}^{n-1} \| \afterm{m} \|_\infty \right) 
    \sum_{\ell=\ki_0 + 1}^{\ki} \kappa_{\N,n}^{\ell},
\end{align*}
from which the bound \eqref{eq:theo:bias-mse-rolling:bias} follows immediately. 

We turn to the MSE. Using the decomposition
\begin{align*}
    \E_{\upxi}[(\rollingestim[\ki_0][\ki][N][f] - \targ{0:n} \af{n})^2] \leq (\ki-\ki_0)^{-2} \left\{ \sum_{\ell=\ki_0 + 1}^{\ki}   \E_{\upxi}[ (\occm(\statmb{n}[\ell])(\operatorname{id}) - \targ{0:n} \af{n} )^2] \right.\\
    + \left. 2 \sum_{\ell=\ki_0 + 1}^{\ki} \sum_{j= \ell + 1}^{\ki} \E_{\upxi}[(\occm(\statmb{n}[\ell])(\operatorname{id}) - \targ{0:n} \af{n})(\occm(\statmb{n}[j])(\operatorname{id}) - \targ{0:n} \af{n})] \right\}, 
\end{align*}
the MSE bound in \Cref{thm:bias:bound} implies that 
\begin{equation*}
\sum_{\ell=\ki_0 + 1}^{\ki} \E_{\upxi}[(\occm(\statmb{n}[\ell])(\operatorname{id}) - \targ{0:n} \af{n})^2] \leq 
\cstparismse \left(\sum_{m=0}^{n - 1} \| \afterm{m} \|_\infty \right)^2 \N^{-1} (\ki - \ki_0).
\end{equation*}
Moreover, using the covariance bound in \Cref{thm:bias:bound}, we deduce that 
\begin{multline*}
    \sum_{\ell=\ki_0 + 1}^{\ki} \sum_{j= \ell + 1}^{\ki} \E_{\upxi}[(\occm(\statmb{n}[\ell])(\operatorname{id}) - \targ{0:n} \af{n})(\occm(\statmb{n}[j])(\operatorname{id}) - \targ{0:n} \af{n})] \\
    \leq \cstpariscov  \left(\sum_{m=0}^{n - 1} \| \afterm{m} \|_\infty \right)^2 \N^{-3/2}  \left( \sum_{\ell=\ki_0 + 1}^{\ki} \sum_{j= \ell + 1}^{\ki} \kappa_{\N,n}^{(j - \ell)} \right)  .
\end{multline*}
Thus, the proof is concluded by noting that  $\sum_{\ell=\ki_0 + 1}^{\ki} \sum_{j= \ell + 1}^{\ki} \kappa_{\N,n}^{(j - \ell)}  \leq (\ki - \ki_0) / (1- \kappa_{N,n})$.

\typeout{get arXiv to do 4 passes: Label(s) may have changed. Rerun}
\end{document}